\documentclass[a4paper,UKenglish,cleveref, autoref, thm-restate]{lipics-v2021}

\newif\ifFull
\newif\ifConf

\Fulltrue
\Conffalse

\ifFull
    \pdfoutput=1 
    \hideLIPIcs  
\fi

\nolinenumbers

\title{On the Complexity of Minimising the Moving Distance for Dispersing Objects} 

\author{Nicol{\'a}s {Honorato-Droguett}}{Nagoya University, Japan \and \url{https://sites.google.com/view/nicolas-honorato-droguett}}{honorato.droguett.nicolas.n7@s.mail.nagoya-u.ac.jp}{https://orcid.org/0009-0005-1969-3649}{}
\author{Kazuhiro {Kurita}}{Nagoya University, Japan \and \url{https://sites.google.com/view/kazuhirokurita} }{kurita@i.nagoya-u.ac.jp}{https://orcid.org/0000-0002-7638-3322}{This work is partially supported by JSPS KAKENHI Grant Numbers 
JP21K17812, 
JP22H03549, 
and JST ACT-X Grant Number JPMJAX2105. 
}
\author{Tesshu {Hanaka}}{Kyushu University, Japan \and \url{https://sites.google.com/view/tesshu-hanaka/home} }{hanaka@inf.kyushu-u.ac.jp}{https://orcid.org/0000-0001-6943-856X}{This work is partially supported by JSPS KAKENHI Grant Numbers 
JP21K17707, 
JP23H04388, 
JST CRONOS Grant Number JPMJCS24K2. 
}
\author{Hirotaka {Ono}}{Nagoya University, Japan \and \url{https://researchmap.jp/hono} }{ono@i.nagoya-u.ac.jp}{https://orcid.org/0000-0003-0845-3947}{This work is partially supported by JSPS KAKENHI Grant Numbers
JP20H05967, 
JP21K19765, 
JP22H00513, 
JST CRONOS Grant Number JPMJCS24K2. 
}
\authorrunning{N. {Honorato-Droguett} and K. Kurita and T. Hanaka and H. Ono} 

\Copyright{Nicol{\'a}s {Honorato-Droguett} and Kazuhiro Kurita and Tesshu Hanaka and Hirotaka Ono} 



\category{} 

\ifConf
    \relatedversiondetails[cite=]{Full Version}{https://arxiv.org} 
\fi

\EventEditors{John Q. Open and Joan R. Access}
\EventNoEds{2}
\EventLongTitle{42nd Conference on Very Important Topics (CVIT 2016)}
\EventShortTitle{CVIT 2016}
\EventAcronym{CVIT}
\EventYear{2016}
\EventDate{December 24--27, 2016}
\EventLocation{Little Whinging, United Kingdom}
\EventLogo{}
\SeriesVolume{42}
\ArticleNo{23}
\usepackage{amsmath}
\usepackage{amsthm}

    \usepackage{bm}
    \usepackage{ascmac}
    \DeclareMathAlphabet{\mathcal}{OMS}{cmsy}{m}{n}
    \usepackage{booktabs}
    \usepackage{array}
    \newcolumntype{C}[1]{>{\centering\let\newline\\\arraybackslash\hspace{0pt}}m{#1}}
    \usepackage{float}
    \usepackage{multicol}
    \usepackage{multirow} 
    \usepackage{layout}
    \usepackage{xcolor}
    \usepackage{colortbl}
    \usepackage[title,titletoc]{appendix}
    \usepackage{adjustbox}
    \usepackage{lineno}
    \usepackage{cite}
    \usepackage[full]{complexity}
    \usepackage[linesnumbered,lined,ruled,noend,vlined]{algorithm2e}
    \makeatletter
    \@ifpackageloaded{algorithm2e}{
    \SetKwProg{Fn}{Function}{}{}
    \SetKwProg{Procedure}{Procedure}{}{}
    \SetKwProg{Subprocedure}{Subprocedure}{}{}
    \SetKwComment{tcc}{//}{}
    \SetKwFunction{Output}{Output}%
    \SetKw{Continue}{continue} 
    \SetKwInOut{AlgInput}{Input}
    \SetKwInOut{AlgOutput}{Output}
    \SetKwInOut{AlgPrecondition}{Pre-conditions}
    \SetKwInOut{AlgInvariant}{Invariants}
    }{}
    \makeatother
    \crefname{problem}{Problem}{Problems}
    \crefname{itembox}{Problem}{Problems}
    \crefname{algocf}{algorithm}{algorithms}
    \Crefname{algocf}{Algorithm}{Algorithms}
    \crefname{AlgoLine}{line}{lines}
    \Crefname{AlgoLine}{Line}{Lines}
    \crefname{observation}{observation}{observations}
    \Crefname{observation}{Observation}{Observations}
    \crefname{mtheorem}{theorem}{theorems}
    \Crefname{mtheorem}{Theorem}{Theorems}
    \crefname{mlemma}{lemma}{lemmas}
    \Crefname{mlemma}{Lemma}{Lemmas}
    \crefname{mcorollary}{corollary}{corollaries}
    \Crefname{mcorollary}{Corollary}{Corollaries}

    \usepackage[title]{appendix}
    \ifConf
         \usepackage{apxproof}
    \fi
    \ifFull
       \usepackage[appendix=inline]{apxproof}
    \fi
    
    \newtheoremrep{mtheorem}[theorem]{\ifConf$\star\,$\fi Theorem}
    \newtheoremrep{mlemma}[lemma]{\ifConf$\star\,$\fi Lemma}
    \newtheoremrep{mcorollary}[corollary]{\ifConf$\star\,$\fi Corollary}
        
    \usepackage[section]{placeins}

\def\Underline{\setbox0\hbox\bgroup\let\\\endUnderline}
\def\endUnderline{\vphantom{y}\egroup\smash{\underline{\box0}}\\}
\def\|{\verb|}

\providecommand{\A}{}\renewcommand{\A}{\mathcal{A}}
\newcommand{\B}{\mathcal{B}}
\renewcommand{\C}{\mathcal{C}}
\providecommand{\D}{}\renewcommand{\D}{\mathcal{D}}

\newcommand{\F}{\mathcal{F}}
\providecommand{\G}{}\renewcommand{\G}{\mathcal{G}}
\renewcommand{\H}{\mathcal{H}}
\providecommand{\I}{}\renewcommand{\I}{\mathcal{I}}
\newcommand{\J}{\mathcal{J}}

\newcommand{\n}{\eta}
\renewcommand{\R}{\mathcal{R}}
\renewcommand{\S}{\mathcal{S}}

\newcommand{\X}{\mathcal{X}}

\newcommand{\len}[1]{\mathrm{len}(#1)}
\newcommand{\size}[1]{|#1|}
\newcommand{\set}[1]{\{#1\}}

\newcommand{\diam}[1]{\mathrm{diam}\:#1}

\newcommand{\gged}[0]{\textsc{Geometric Graph Edit Distance}}
\newcommand{\ggedmm}[0]{\textsc{minimax-Geometric Graph Edit Distance}}
\newcommand{\idisp}[0]{\textsc{Interval Dispersal}}

\newcommand{\pthreesat}[0]{\textsc{Planar 3-SAT}}
\newcommand{\threepartition}[0]{{\scshape 3-Partition}}

\makeatletter
\DeclareMathOperator*{\argmin}{\smash[b]{\operator@font arg\,min}}
\DeclareMathOperator*{\argmax}{\smash[b]{\operator@font arg\,max}}
\makeatother

\definecolor{ered}{HTML}{dc133b}

\theoremstyle{plain}
\begin{document}

\maketitle

\begin{abstract}
    We study {\gged} (GGED), a graph-editing model to compute the minimum edit distance of intersection graphs that uses moving objects as an edit operation.
    We first show an $O(n\log n)$-time algorithm that minimises the total moving distance to disperse unit intervals.
    This algorithm is applied to render a given unit interval graph (i) edgeless, (ii) acyclic and (iii) $k$-clique-free. 
    We next show that GGED becomes strongly \NP-hard when rendering a weighted interval graph (i) edgeless, (ii) acyclic and (iii) $k$-clique-free. 
    Lastly, we prove that minimising the maximum moving distance for rendering a unit disk graph edgeless is strongly \NP-hard over the $L_1$ and $L_2$ distances.
\end{abstract}

\section{Introduction}\label{sec:intro}
Graph modification is a fundamental topic to address graph similarity and dissimilarity, where a given graph is deformed by adding or deleting vertices or edges to satisfy a specific non-trivial graph property, while minimising the cost of edit operations. The problem of determining this cost is commonly known as \emph{graph modification problem} (GMP) and has applications in various disciplines, such as computer vision~\cite{Chung1994}, network interdiction~\cite{Hoang2023}, and molecular biology~\cite{Hellmuth2020}. GMPs are often categorised into vertex and edge modification problems, with edit operations restricted to the vertex and edge sets, respectively.

The cost of a single edit operation in a GMP is often determined by the specific application. In theoretical studies, a unit-cost model is often assumed, where each addition or deletion of a vertex or edge has a uniform cost. However, for such models, it is known that determining whether a graph can be modified to satisfy a given property is \NP-hard for a wide range of graph classes and properties~\cite{Lewis1980,Burzyn2006,Fomin2015,Sritharan2016}. These negative bounds of GMPs motivate alternative formulations for graph editing that consider domain-specific constraints and cost measures.

The choice of edit operations and their associated costs is a crucial aspect of GMPs, as different formulations capture different structural properties and computational challenges. Analogous to string similarity analysis, where modifications are based on biologically significant operations such as DNA mutations and repeats~\cite{Li2009}, graph modification problems should reflect the inherent constraints and structural properties of the graphs being studied. In particular, \emph{geometric intersection graphs} (hereafter intersection graphs) provide a suitable framework for studying GMPs for scenarios where graphs represent spatial relationships (see, e.g.,~\cite{Panolan2024,Berg2019,fomin2023}). Given a collection of geometric objects $\S$, an \emph{intersection graph} $(G,\S)$ is a graph where there is a one-to-one correspondence between the vertex set $V(G)$ and $\S$, and two vertices are adjacent if and only if their corresponding objects intersect. This model includes many well-known graph classes, such as interval graphs and disk graphs. These graphs can be frequently found in real-world applications such as network modelling and bioinformatics~\cite{McKee1999}.

Motivated by this context, this paper investigates GMPs for intersection graphs. In this context, two natural questions arise: \begin{enumerate} \item Are standard graph edit operations suitable for modifying intersection graphs? \item How can the geometric properties of objects be exploited to overcome the hardness of GMPs? \end{enumerate} To answer these questions, we introduce {\gged}, a model for modifying intersection graphs from a geometric perspective.

In the intersection graph model, a natural edit operation is to move the objects in $\S$. We treat this movement as a graph edit operation and focus on minimising the cost required to modify an intersection graph to satisfy a specific graph property. The cost is quantified by the total moving distance, which is the sum of the distances by which objects in $\S$ are moved. More precisely, we define the problem as follows:

\begin{itembox}[l]{{\gged}}\label{pro:uig_general} \begin{description} \item[Input:] An intersection graph $(G,\S)$ and a graph property $\Pi$. \item[Output:] The minimum total moving distance of the objects in $\S$ so that the resulting intersection graph satisfies $\Pi$. \end{description} \end{itembox}

We assume that $\Pi$ is given by an oracle, i.e. we have an algorithm to determine whether a given intersection graph $(G, \S)$ satisfies $\Pi$.

\subsection*{Related work}
Numerous GMPs are known to be computationally hard. In the early 1980s, Lewis and Yannakakis~\cite{Lewis1980} showed that vertex-deletion problems are \NP-complete for any hereditary graph property. Similarly, many edge modification problems have been shown to be \NP-complete, such as transforming a graph into a perfect, chordal, or interval graph~\cite{Burzyn2006}. As a result, the past decade has seen a growing interest in addressing these problems from the perspective of parameterised complexity. The recent survey by Crespelle et al.\cite{Crespelle2023} provides a comprehensive overview of this subject (see also\cite{Drange2015}).

Although classical GMPs focus on structural modifications of graphs, recent studies have explored models that include geometric constraints. Honorato-Droguett et al.~\cite{HonoratoDroguett2024} introduced the above geometric approach to graph modification, demonstrating that certain properties, such as graph completeness and the existence of a $k$-clique, can be efficiently satisfied on interval graphs. 
Their work highlights how the underlying geometric properties of intersection graphs can be exploited to design appropriate modification models.

In a similar vein, Fomin et al.~\cite{fomin2023} studied the \emph{disk dispersal} problem, where a set $\S$ of $n$ disks, an integer $k \geq 0$, and a real number $d \geq 0$ are given, and the goal is to determine whether an edgeless disk graph can be realised by moving at most $k$ disks by at most $d$ distance each. They proved that this problem is {\NP-hard} when $d=2$ and $k = n$ and also {\FPT} when parameterised by $k+d$. Furthermore, they showed that the problem becomes \W[1]-hard when parameterised by $k$ when disk movement is restricted to rectilinear directions.

Expanding on this line of research, Fomin et al.~\cite{Fomin2025} conducted a parameterised complexity study of edge modification problems where \emph{scaling} objects is considered as the edit operation. Their results illustrate how alternative edit operations in geometric intersection graphs can impact computational complexity, enabling further study of geometric modification graph models.
In particular, their work includes several {\FPT} results to achieve independence, acyclity and connectivity on disk graphs.

Our work continues these developments by introducing {\gged}, a model that considers object movement as an edit operation to modify intersection graphs. Unlike prior studies that focus on vertex and edge modifications or object scaling, our approach explicitly considers movement costs by quantifying the total moving distance required to satisfy a given graph property. This approach enables the exploration of new algorithmic and complexity-theoretic questions in the context of geometric intersection graphs.

\subsection*{Our contribution}
Our results are mainly focused on interval graphs and summarised in \Cref{tab:summary}. In this paper, we deal with the following graph properties: \begin{itemize} 
\item $\Pi_{\texttt{edgeless}}$ (edgeless graphs), 
\item $\Pi_{\texttt{acyc}}$ (acyclic graphs), and 
\item $\overline{\Pi_{k\texttt{-clique}}}$ ($k$-clique-free graphs).
\end{itemize} 

\begin{table}[!b]
    \setlength\extrarowheight{2pt}
    \centering
    \caption{Summary of our results. In this table, $L_1$ and $L_2$ are the Manhattan and Euclidean distances, respectively. The terms IG, UIG and UDG are abbreviations of interval graphs, unit interval graphs and unit disk graphs, respectively.}\label{tab:summary}
    \begin{tabular}{|c|c|c|c|c|c|}
        \hline
        {\bfseries Problem} &  &  &  & {\bfseries Target Graph} & \\
        {\bfseries Type} & \multirow{-2}{*}{{\bf Metric}} & \multirow{-2}{*}{{\bf Weighted}} & \multirow{-2}{*}{\textbf{Graph}} & {\bfseries Property} & \multirow{-2}{*}{\textbf{Complexity}} \\
        \hline
        \multirow{6}{*}{\texttt{minsum}} & \multirow{7}{*}{$L_2(=L_1)$}  & \multirow{4}{*}{Yes} & \multirow{4}{*}{\textbf{IG}} &  $\Pi_{\texttt{edgeless}}$ & strongly \NP-hard\\\cline{5-6}
        & &  &  &   $\Pi_{\texttt{acyc}}$  & strongly \NP-hard \\\cline{5-6}
        & &  &  &   \multirow{2}{*}{$\overline{\Pi_{k\texttt{-clique}}}$} & {strongly \NP-hard} \\
        & & & & & {for any $1\le k \le n$} \\\cline{3-6}
        &  & \multirow{3}{*}{No} &\multirow{3}{*}{\textbf{UIG}}  &  $\Pi_{\texttt{edgeless}}$ & $O(n\log n)$ \\\cline{5-6}
        & &  &  &   $\Pi_{\texttt{acyc}}$  & $O(n\log n)$ \\\cline{5-6}
        & &  &  &   $\overline{\Pi_{k\texttt{-clique}}}$ & $O(n\log n)$ \\
       \hline
        \multirow{1}{*}{\texttt{minimax}} & \multirow{1}{*}{$L_1$, $L_2$} & Yes &\multirow{1}{*}{\textbf{UDG}} &   \multirow{1}{*}{$\Pi_{\texttt{edgeless}}$} & \multirow{1}{*}{strongly \NP-hard} \\\hline
    \end{tabular}
\end{table}

In~\cite{HonoratoDroguett2024}, the model presented is studied mainly for properties of dense graphs.
This inspires the present paper as a subsequent work, where we instead focus on properties for sparse graphs.
As two fundamental classes of sparse graphs, we consider edgeless graphs ($\Pi_{\texttt{edgeless}}$) and acyclic graphs ($\Pi_{\texttt{acyc}}$). These properties have also been studied in related work on geometric intersection graphs\cite{fomin2023,Fomin2025}.

As we shall detail, $\Pi_{\texttt{acyc}}$ is contained in $\overline{\Pi_{k\texttt{-clique}}}$ in our context.
As a result, one might argue that the distinction of both properties is irrelevant.
However, we still consider them distinctively, as forests are a well-known class of graphs.
Our analysis highlights the computational complexity of modifying intersection graphs while considering movement-based edit operations, a perspective distinct from prior work that focuses on exclusively modifying the graph structure.

\subsection*{Paper Organisation}
\Cref{sec:preliminaries} formally describes the definitions needed to address the above ideas. \Cref{sec:edg_uig} presents the problem {\idisp} and shows that it can be solved in $O(n\log n)$ time. Using this algorithm, we establish that {\gged} can also be solved in $O(n\log n)$ time for satisfying $\Pi_{\texttt{edgeless}}$, $\Pi_{\texttt{acyc}}$, and $\overline{\Pi_{k\texttt{-clique}}}$ on unit interval graphs.
\Cref{sec:edg_ig} demonstrates that {\gged} becomes strongly \NP-hard on weighted interval graphs for satisfying $\Pi_{\texttt{edgeless}}$, $\Pi_{\texttt{acyc}}$, and $\overline{\Pi_{k\texttt{-clique}}}$.
\Cref{sec:disk_edgeless} shows that the minimax version of {\gged} is strongly \NP-hard on weighted unit disk graphs when satisfying $\Pi_{\texttt{edgeless}}$ under both the $L_1$ and $L_2$ distance metrics.
\Cref{sec:conclu} concludes with remarks on our results and potential future directions.

\ifConf
Due to space restrictions, we omit in-depth explanations and all full proofs of statements with a $\star$-mark. The reader is referred to the full version of this paper~[] for these details.
\fi

\section{Preliminaries}\label{sec:preliminaries}
    This section provides the main definitions used in the paper, referencing geometry, graph, and convexity terminology from textbooks~\cite{Preparata1985,cormen2009,Diestel2017,boyd2004}.
\ifConf
\begin{toappendix}
\section{Auxiliary Definitions}
\else
\begin{toappendix}
\fi
\ifFull
\paragraph*{Convexity}
\fi
    A set $C$ is \emph{convex} if the line segment between any two points in $C$ lies entirely in $C$. Such a set is called \emph{convex set}. 
    A function $f:\mathbb{R}^n \rightarrow \mathbb{R}$ is \emph{convex} if its domain $\mathrm{dom} f$ is a convex set, and if for all $x,y \in \mathrm{dom} f$ and $\theta \in [0,1]$, the inequality
    $f(\theta x + (1-\theta)y) \leq \theta f(x) + (1-\theta)f(y)$ holds. 
    A function satisfying the above is called a \emph{convex function}. 
    The \emph{convex hull} of a finite set of points $P \subset \mathbb{R}^2$ is the smallest convex set that contains $P$ and is denoted by $\C(P)$.
    A \emph{polygon} is defined by a finite set of segments such that every segment endpoint is shared by exactly two edges, and no subset of edges has the same property. 
    The segments are the \emph{edges} and their endpoints are the \emph{vertices} of the polygon. 
    An $n$-vertex polygon is called an $n$-gon.
    A polygon is \emph{simple} if no pair of non-consecutive edges that share a point exists.
    A simple polygon is \emph{convex} if its interior is a convex set.
    Similarly, the \emph{boundary} of a convex hull is a convex polygon.

\end{toappendix}

\paragraph*{Objects}
    An \emph{interval} $I$ is a line segment on the real line of length $\mathrm{len}(I) \in \mathbb{R}^+$.
    Intervals are assumed to be open, unless explicitly stated otherwise.
    An interval such that $\mathrm{len}(I) = 1$ is called \emph{unit interval}.
    The \emph{left endpoint} $\ell(I)$ of an interval $I$ is the point that satisfies $\ell(I) \le y$ for any $y \in I$. 
    Similarly, the \emph{right endpoint} $r(I)$ of $I$ is the point that satisfies $y \le r(I)$ for any $y \in I$.
    The \emph{centre} $c(I)$ of $I$ is the point $c(I) = (r(I) - \ell(I))/2$.
\begin{toappendix}
    The \emph{left interval set} $L(\mathcal{I}, x)$ of a collection of intervals $\mathcal{I}$ is the subcollection of intervals to the `left' of a given point $x$.
    That is, $L(\mathcal{I}, x) = \{I \in \mathcal{I}:\: r(I) < x\}$. 
    Similarly, the \emph{right interval set} $R(\mathcal{I},x)$ is defined as $R(\mathcal{I},x) = \{I\in \mathcal{I}:\:\ell(I) > x\}$.
    Additionally, the \emph{leftmost endpoint} $\ell(\mathcal I)$ of a collection of intervals $\mathcal{I}$ is defined as $\min_{I \in \mathcal I} \ell(I)$.
    Similarly, the \emph{rightmost endpoint} $r(\mathcal I)$ of $\mathcal{I}$ is defined as $\max_{I \in \mathcal I} r(I)$.
\end{toappendix}
    Throughout the paper, we assume that the indices of a collection of intervals $\I = \set{I_1,\ldots,I_n}$ follow the order given by centres of intervals.
    That is, $c(I_{i}) \le c(I_{i+1})$ for all $1\le i\le n-1$.
    However, it is not assumed that collections are ordered when given as the input graph.
    %
    Given a radius $r>0$ and a point $p\in \mathbb{R}$, a \emph{disk} $D$ centred at $p$ is the set $D = \set{x\in \mathbb{R}^2\mid \lVert x,p \rVert_2 \le r}$.
    An \emph{open disk} $D$ is a disk without its boundary circle; that is, $D = \set{x\in \mathbb{R}^2\mid \lVert x,p \rVert_2 < r}$.
    We assume that the disks are open, unless we mention the contrary.
    A \emph{unit disk} is a disk of radius $r = 1/2$.
    \begin{toappendix}
    The \emph{minimax centre} $p$ of a finite set of points $P \subset \mathbb{R}^2$ is the centre of the smallest circle that contains $P$, which is the point that minimises $\max_{p'\in P} \lVert p,p'\rVert_m$ for $m \in \set{1,2}$.
    The \emph{diameter} $\diam{P}$ of a finite set of points $P \subset \mathbb{R}^2$ is the distance of the farthest pair of points in $P$.
    Similarly, the diameter $\diam{S}$ of a convex polygon $\S$ is the diameter of its vertices.
    \end{toappendix}
    The \emph{$L_m$ distance} for a $m\ge 1$ defines the distance between two points $p = (p_1,\ldots,p_d)$ and $q = (q_1,\ldots,q_d)$ in $\mathbb{R}^d$ as $\lVert p,q\rVert_m = \sqrt[m]{(p_1-q_1)^m+\cdots + (p_d - q_d)^m}$.
    In all subsequent sections, we use the $L_2$ distance (also known as the Euclidean distance) and the $L_1$ distance (also known as the Manhattan distance).

\paragraph*{Graphs}
    Throughout the paper, a graph $G = (V,E)$ is assumed to be a simple, finite, and undirected graph with vertex set $V$ and edge set $E$.
    An \emph{edgeless graph} is a graph $G = (V, E)$ such that $E = \emptyset$.
    A $k$\emph{-clique} of a graph $G = (V, E)$ is a subset $W\subseteq V$ such that $|W| = k$ and for all $u,v \in W,\: u\neq v$, $\{u,v\} \in E$, for $k \le n$.
    If such $W$ exists in $V$, we say that $G$ \emph{contains a $k$-clique}.
    An \emph{interval graph} is an intersection graph $G = (V,E)$ where the vertex set $V = \{v_1,\ldots,v_n\}$ corresponds to a collection of intervals $\mathcal{I} = \{I_1,\ldots,I_n\}$ and an edge $\{v_i,v_j\} \in E$ exists if and only if $I_i \cap I_j \neq \emptyset$, for any $1\leq i,j \leq n,\: i \neq j$. 
    An interval graph is called \emph{unit interval graph} if $\len{I} = 1$ for all $I \in \mathcal{I}$.
    Similarly, a \emph{disk graph} is an intersection graph $G = (V,E)$ where the vertex set $V = \{v_1,\ldots,v_n\}$ corresponds to a disk collection $\mathcal{D} = \{D_1,\ldots,D_n\}$. An edge $(v_i,v_j) \in E$ exists if and only if $D_i \cap D_j \neq \emptyset$, for any $1\leq i,j \leq n,\: i \neq j$.
    A \emph{unit disk graph} is a disk graph in which the collection contains exclusively unit disks.
    Unless stated otherwise, all intersection graphs are assumed to be \emph{unweighted}. 
    A \emph{weighted intersection graph} assigns a multiplicative weight, called the \emph{distance weight}, to the moving distance function of each object. 
    The formal definition of distance weight appears in later sections when required.
    An (infinite) set of graphs $\Pi$ is a \emph{graph property} (or simply a property), and we say that \emph{$G$ satisfies $\Pi$} if $G \in \Pi$.
    A graph property $\Pi$ is \emph{non-trivial} if infinitely many graphs belong to $\Pi$ and infinitely many graphs do not belong to $\Pi$.
    In this paper, we deal with the following non-trivial properties:
    (i)  $\Pi_{\texttt{edgeless}} = \{G :G\text{ is an edgeless graph.}\}$,
    (ii) $\Pi_{\texttt{acyc}} = \{G :G\text{ is an acyclic graph.}\}$,  
    (iii)  $\Pi_{k\texttt{-clique}} = \{G :G\text{ contains a $k$-clique.}\}$ and
    (iv)   $\overline{\Pi_{k\texttt{-clique}}} = \{G :G \not\in \Pi_{k\texttt{-clique}}\}$.

\section{Satisfying \texorpdfstring{$\Pi_{\texttt{edgeless}}$}{} on Unit Interval Graphs in \texorpdfstring{$O(n\log n)$}{} time}\label{sec:edg_uig}

We show that $\Pi_{\texttt{edgeless}}$ can be satisfied in $O(n\log n)$ time given a unit interval graph of $n$ intervals.
We start by defining a problem that we call {\idisp} and then use the algorithm designed to satisfy the properties $\Pi_{\texttt{edgeless}}$, $\Pi_{\texttt{acyc}}$ and $\overline{\Pi_{k\texttt{-clique}}}$. 
\ifConf
{\idisp} receives as input a collection $\I$ of $n$ intervals and a real $s \ge 1$, and asks for the minimum value of the total moving distance to obtain a collection $\I'$ that satisfies $c(I'_j)-c(I'_i) \ge s$ for each $I'_i, I'_j \in \I'$, $i<j$.
\fi
\ifFull
The problem is defined as follows:

\begin{itembox}[l]{\idisp}
    \begin{description}
        \item[Input:] A collection $\I$ of $n$ intervals and a real $s \ge 1$.
        \item[Output:] The minimum value of the total moving distance for obtaining a collection $\I'$ that satisfies $c(I'_j)-c(I'_i) \ge s$ for each $I'_i, I'_j \in \I'$, $i<j$.
    \end{description}
\end{itembox}
\fi
When $s = 1$, {\idisp} is equivalent to {\gged} on unit interval graphs for satisfying $\Pi_{\texttt{edgeless}}$.
%
%
For simplicity, the intervals are assumed to be open. This avoids the need to address infinitesimally small distances required to separate closed intervals.
We must first introduce some basic definitions and notation to describe the algorithm.
Given a collection of $n$ intervals $\I= \set{I_1,\ldots,I_{n}}$, let $D = (d_1,\ldots,d_{n})$ be a vector such that $d_i$ is the moving distance applied to $I_i$. 
We denote by $\I^D = \{I^D_1\ldots, I^D_{n}\}$ the collection of intervals such that $c(I^D_i) = c(I_i) +d_i$.
The set $\D(\I) \subseteq \mathbb{R}^n$ is the set of vectors that describe the moving distance applied to intervals such that the condition of {\idisp} is satisfied. 
In other words, for all $D = (d_1,\ldots,d_n) \in \D(\I)$, $c(I^D_{j})+c(I^D_{i}) \ge s$ holds for $i < j$.
We use $\D^{\mathit{opt}}(\I) \subseteq \D(\I)$ to denote the subset of vectors in $\D(\I)$ that minimises the total moving distance applied to intervals; i.e. $\D^{\mathit{opt}}(\I) = \set{D=(d_1,\ldots,d_n) \in \D(\I)\mid \sum_{1\le i \le n} |d_i| = \min_{D' = (d'_1,\ldots,d'_n) \in \D(\I)}{\sum_{1\le i \le n}|d'_i|}}$.

Intuitively, we aim to find a vector $D \in \D^{\mathit{opt}}(\I)$ to move each interval so that the distance between each pair of intervals is at least $s$.
Given an arbitrary $D \in \D^{\mathit{opt}}(\I)$, the order of $\I^D$ may be different from the order of $\I$.
However, it was previously shown~\cite{HonoratoDroguett2024} that the there are always a vector $D \in \D^{\mathit{opt}}(\I)$ such that the order of $\I^D$ preserves the order of $\I$.
This implies that there always exists an optimal solution of {\idisp} for which checking the inequality $(c(I_{i+1})+d_{i+1}) - (c(I_{i})+d_{i}) \ge s \text{ for } \leq i \leq n-1$ is sufficient.

%
%
We now define the \emph{equispace function}, which moves intervals so that the distance between their centres is exactly $s$, maintaining the order induced by interval centres.
\begin{definition}[\textit{Equispace function}]\label{def:tmd_equispace}
    Let $(\I,s)$ be an instance of {\idisp} where $\I$ is a collection of unit intervals. The \emph{equispace function} of $\I$ to a point $x$ is a function $E: \I\times \mathbb{R} \rightarrow \mathbb{R}$ defined as:
    \begin{align*}
        E(\I,x) = \sum_{i = 1}^n f_i(x),\quad f_i(x) = |x-c(I_i) - (n-i)s|.
    \end{align*}
    The vector that describes the moving distances given by $E(\I,x)$ is defined as $E_x(\I) = (e_1,\ldots,e_n)= \left(\alpha_1 f_1(x),\ldots,\alpha_n f_n(x)\right)$ where $\alpha_i = 1$ if $x\ge c(I_i)+(n-i)s$ and $\alpha_i = -1$ otherwise, for $1\le i \le n$.
    We also denote by $\I^{E_x(\I)} = \{I^{E_x(\I)}_1\ldots, I^{E_x(\I)}_{n}\}$ the collection of intervals where $c(I^{E_x(\I)}_i) = c(I_i) +\alpha_if_i(x)$ for $1\le i \le n$.
\end{definition}
By the above, $E_x(\I) \in \D(\I)$ for all $x \in \mathbb{R}$.
Moreover, $c(I^{E_x(\I)}_{i+1}) - c(I^{E_x(\I)}_{i}) = s$ for all $1 \le i \le n-1$.
We first prove that for certain collections of intervals, minimising $E$ gives a vector contained in $\D^{\mathit{opt}}(\I)$.
\begin{mlemmarep}\label{lem:tmd_convex}
    The equispace function $E(\I,x)$ is a piecewise-linear convex function.
\end{mlemmarep}
\begin{proof}
    By \Cref{def:tmd_equispace}, the function $E(\I,x)$ is a function of the form $f(x) = f_n(x)+\cdots+ f_1(x) = |x - c(I_n)| + |x - c(I_{n-1}) - s| + \cdots + |x - c(I_1) - (n-1)s|$. 
    The absolute function is convex; hence each $f_i$ is convex.
    Consequently, $f(x)$ is also convex as it is a sum of convex functions.
    On the other hand, the piecewise linearity of $E$ is given by the fact that each absolute function is piecewise linear.
    Therefore, $E$ is a piecewise-linear convex function.
\end{proof}

We define the \emph{set of breakpoints of $E(\I,x)$} to be the set $B_\I =\set{b_1^\I,\ldots,b_n^\I} = \set{c(I_i) + (n-i)s\mid I_i \in \I,\: 1\le i\le n}$.
Given a collection of intervals $\I$, we define the equispace function $E(\I, x)$ as a sequence of linear functions $E_1(\I, x), \ldots, E_{\size{\I}+1}(\I, x)$. 
The slope of $E_i(\I, x)$ is less than the slope of $E_j(\I, x)$ for $1 \le i < j \le \size{\I}$.
Since the equispace function is convex and piecewise linear, the points that minimise $E$ are located within a range $b_\ell \leq x \leq b_r$, where $b_{\ell} \le b_r$ and $b_\ell, b_r \in B_\I$.
We prove that $b_\ell$ and $b_r$ can be easily found.

\begin{mlemmarep}\label{lem:tmd_opt}
     The minimum value of $E(\I,x)$ is given by the breakpoint $b^{\I}_{(n+1)/2}$ if $n$ is odd, and by breakpoints $b^{\I}_{n/2}$ and $b^{\I}_{(n/2)+1}$ otherwise.
\end{mlemmarep}
\begin{proof}
    Let $s_i$ be the slope of $E_i(\I,x)$.
    The function $E(\I,x)$ is a function of the form $f_1(x)+\cdots+f_n(x)$ where $f_i(x) = |x-c(I_i) -(n-i)s|$.
    The slope of $f_i$ is equal to $1$ if $x\ge c(I_i) +(n-i)s$ and $-1$ otherwise, which implies that $s_i = -n + 2(i-1)$ for $1\le i \le n +1$ and thus $s_{i+1}-s_i = 2$ for $1\le i \le n$.
    Suppose first that $n$ is odd.
    Then, 
    \begin{align*}
        s_{(n+1)/2} & = -n + 2((n+1)/2 - 1) = -n +n+1 - 2 = -1\text{ and}\\ 
        s_{(n+1)/2 + 1} &= -n + 2((n+1)/2 + 1 - 1) = -n+n+1 = 1.
    \end{align*}
    Hence, it follows that the $((n+1)/2)$th breakpoint minimises $E$.
    Suppose that $n$ is even.
    In this case, we have that
    \begin{align*}
        s_{n/2} & = -n + 2(n/2 - 1) = -n +n - 2 = -2,\\ 
        s_{n/2 + 1} &= -n + 2(n/2 + 1 - 1) = -n+n = 0\text{ and},\\
        s_{n/2 + 2} &= -n + 2(n/2+2 - 1) = -n + n +4 -2 = 2.
    \end{align*}
    Hence, any point $x\in \mathbb{R}$ such that $b^{\I}_{n/2} \le x \le b^{\I}_{n/2+1}$ minimises $E$, implying that $b^{\I}_{n/2}$ and $b^{\I}_{(n/2)+1}$ also minimise $E$.
    This concludes the proof.
\end{proof}

By \Cref{lem:tmd_opt}, the minimum value of $E$ for an arbitrary collection of intervals $\I$ is given by the median value(s) of $B_{\I}$.
%
%
We now show which collections allow minimising $E$ to obtain a vector in $\D^{\mathit{opt}}(\I)$, characterised as follows:
\begin{definition}[\textit{Optimally Equispaceable Collections}]\label{def:equispace_collections}
    Given a collection of intervals $\I$, we say that $\I$ is \emph{optimally equispaceable} if there exists a $D \in \D^{\mathit{opt}}(\I)$ such that $D = E_{x^*}(\I)$ and $x^* \in \argmin_{x \in \mathbb{R}} E(\I,x)$. Equivalently, $\I$ is optimally equispaceable if $E_{x^*}(\I) \in \D^{\mathit{opt}}(\I)$ for all $x^* \in \argmin_{x \in \mathbb{R}} E(\I,x)$.
\end{definition}

\begin{lemma}\label{lem:partition_opt}
    Let $\I = \set{I_1,\ldots,I_n}$ be a collection of unit intervals such that $c(I_{i+1}) - c(I_{i}) \le s$ for $1\le i \le n-1$.
    Then $\I$ is optimally equispaceable.
    Moreover, there exists a $D \in \D^{\mathit{opt}}(\I)$ such that $c(I^D_{i+1}) -c(I^D_i) = s$ holds for all $1 \le i \le n-1$.
\end{lemma}
\begin{proof}
    We only prove the latter, as the existence of $D$ in $\D^{\mathit{opt}}(\I)$ directly implies the optimal equispaceability of $\I$.
    That is, we show that $\I^D$ satisfies $c(I^D_{i+1}) - c(I^D_i) = s$, for $1 \le i \le n-1$.
    By the definition of {\idisp}, we have $c(I^D_{i+1}) \ge c(I^D_{i})$ and $c(I^D_{i+1}) - c(I^D_{i}) \ge s$ for $1\le i \le n-1$.
    Suppose that there exists a pair of intervals $I_i$ and $I_{i+1}$ that satisfies $c(I^D_{i+1}) - c(I^D_{i}) > s$. 
    Let $s' = c(I^D_{i+1}) - c(I^D_{i})$ and $\delta = s'-s$.
    We show how to obtain a total moving distance $D'$ such that $\sum_{d\in D'} |d|< \sum_{d\in D} |d|$ and $c(I^{D'}_{i+1}) - c(I^{D'}_{i}) = s$.
    
    We divide the proof into three cases: (i) $d_i \ge 0$, (ii) $d_{i+1} \le 0$ and (iii) $d_i\le 0$ and $d_{i+1}\ge 0$. 
    For case (i), it follows that $d_j \ge d_{j-1}\ge 0$ for $i+1\le j \le n$ and $(c(I^D_{i+1})-\delta) - c(I^D_{i})  = c(I_{i+1})+(d_{i+1}-\delta) - (c(I_i) + d_i) = s$ holds.
    Let $D' (d'_1,\ldots,d'_n) = (d_1,\ldots,d_i,d_{i+1}-\delta,\ldots,d_{n}-\delta)$.
    The dispersal condition is satisfied by $\I^{D'}$.
    Furthermore, since $\delta > 0$, the total moving distance satisfies 
    $\sum_{d\in D'} |d| = \sum_{j=1}^i |d_j| + \sum_{j=i+1}^n d_j - \delta < \sum_{d\in D} |d|$, which contradicts the optimality of $D$.
    
    For case (ii), $d_j \le d_{j+1}$ for $1\le j \le i$ holds, and the argument for case (i) applies analogously for $D' = (d'_1,\ldots,d'_n) = (d_1+\delta,\ldots,d_i+\delta,d_{i+1},\ldots,d_{n})$.
    
    We only need to prove case (iii).
    Let $\delta = s' - s$ as in the previous cases.
    If $\delta \le d_{i+1}$, then we move the intervals as in the first case. If $\delta \le -d_{i}$, then we move intervals as in the second case. 
    In both cases, the same argument applies and the total moving distance contradicts the optimality of $D$.
    Thus we assume that $\delta > d_{i+1}, -d_i$ holds.
    Without loss of generality, we move intervals $I_j$ for $i+1\le j \le n$ by $d_{i+1}$ to the left by $\delta' = d_{i+1}$ and intervals $I_j$ for $1\le j \le i$ to the right by $\delta'' = (c(I^D_{i+1})-\delta')- c(I^D_i) - s$. 
    Then $(c(I^D_{i+1})-\delta') - (c(I^D_i) +\delta'') = s$ holds since $d_{i+1}-\delta' = 0$. 
    Let $D' = (d'_1,\ldots,d'_n) = (d_1+\delta'',\ldots,d_i+\delta'',d_{i+1}-\delta',\ldots,d_n-\delta')$. 
    The inequality $\sum_{d\in D'} |d| = \sum_{j=1}^i d_j+\delta'' + \sum_{j=i+1}^n d_j - \delta' < \sum_{d\in D} |d|$ holds since $\delta',\delta''>0$, which contradicts the optimality of $D$.
    Therefore, in an optimal solution, $\I$ must satisfy $c(I_{i+1})+d_{i+1} - (c(I_i) + d_i) = c(I^D_{i+1})-c(I^D_i) = s$, for $1 \le i \le n-1$.
\end{proof}

Let $\I = \set{I_1,\ldots,I_n}$ and $\J = \set{J_1,\ldots,J_m}$ be two collections of unit intervals and let $x_1,x_2\in \argmin_{x\in\mathbb{R}} E(\I,x)$, $x_1\le x_2$, and $y_1,y_2\in \argmin_{x\in\mathbb{R}} E(\J,x)$, $y_1\le y_2$, be the breakpoints that minimise $E$ for $\I$ and $\J$, respectively.
We say that \emph{$\I$ and $\J$ intersect when equispaced} when $y_1 \le x_2 + \size{\J}s$.
In other words, $\I$ and $\J$ intersect when equispaced whenever there exist points $x_1\le x \le x_2$ and $y_1 \le y \le y_2$ such that there exist $I \in \I^{E_x(\I)}$ and $I \in \J^{E_y(\J)}$ for which $c(J) - c(I) < s$.

\begin{mlemmarep}\label{lem:tmd_opt_of_opts}
    Given that $\I \cup \J = \set{I_1,\ldots,I_n,J_1,\ldots,J_m}$, $\I \cup \J$ is optimally equispaceable if and only if $y_1 \le x_2 + \size{\J}s$.
\end{mlemmarep}
\begin{proof}
    Let $\H = \I \cup \J = \set{I_1,\ldots,I_n,J_1,\ldots,J_m}= \set{I_1,\ldots,I_{n+m}}$ and assume that $\H$ optimally equispaceable.
    In other words, there exists a vector $D = (d_1,\ldots,d_{n+m}) \in \D^{\mathit{opt}}(\H)$ such that $c(I^D_{i+1}) - c(I^D_{i}) = s$ for $1\le i\le n+m-1$.
    Suppose first that $y_1 > x_2 + \size{\J} s$.
    We show that this assumption contradicts the optimality of the vector $D$.
    First, assume that $c(I^D_{n+m}) \ge y_1$.
    Then, it follows that $x_2 < c(I^D_n)$ since $c(I^D_{n+m}) - c(I^D_n) = \size{\J}s$.
    Moreover, $\sum_{i=1}^n |d_i| = E(\I,c(I^D_{n}))$ holds.
    By \Cref{lem:tmd_convex}, $E(\I,x_2) < E(\I,c(I^D_{n}))$ also holds, which contradicts the optimality of the vector $D$.
    Moreover, $c(I^D_{n+1})-x_2 >s$ holds, and hence the inequality $c(I^D_{i+1}) - c(I^D_{i}) = s$ fails when $i=n$.
    Suppose instead that $c(I^D_{n+m}) < y_1$.
    In this case, \Cref{lem:tmd_convex} tells us that $\sum_{i=n}^{n+m} |d_i| = E(\J,c(I^D_{n+m})) > E(\J,y_1)$, which contradicts the optimality of $D$.
    Moreover, if $\delta = c(I^D_{n+m})-y_1$ is the value with which $\J$ is moved from $c(I^D_{n+m})$ to $y_1$, then $(c(I_{n+1}) + \delta) - c(I^D_n) > s$ holds.
    Consequently, $c(I^D_{i+1}) - c(I^D_{i}) = s$ fails when $i = n$.
    In both cases, the initial assumption that $D$ is an optimal solution is contradicted.
    Therefore, $y_1 \le x_2 + \size{\J}s$ must hold.

    In the other direction, assume that $y_1 \le x_2 + \size{\J}s$.
    We show that $\I\cup \J$ is optimally equispaceable.
    That is, we prove the existence of a vector $D = (d_1,\ldots,d_{n+m}) \in \D^{\mathit{opt}}(\I\cup \J)$ such that $c(I^D_{i+1}) - c(I^D_{i}) = s$ for $1\le i\le n+m-1$.
    Let $D' = (d'_1,\ldots,d'_{n+m}) \in \D^{\mathit{opt}}(\I\cup \J)$ be a vector for which the inequality $c(I^D_{i+1}) - c(I^D_{i}) = s$ does not hold for an $1\le i \le n+m-1$.
    We know that $\I$ and $\J$ are optimally equispaceable, thus we can assume that $c(I^{D'}_{i+1}) - c(I^{D'}_{i}) = s$ fails for $D'$ when $i=n$ without loss of generality.
    Let $x = c(I^{D'}_n)$ and $y=c(I^{D'}_{n+m})$.
    We have $x\le x_2$, $y\ge y_1$, and $\sum_{i=1}^{n+m} |d'_i| = E(\I,x) + E(\J,y)$.
    We show that $\sum_{i=1}^{n+m} |d_i| \le \sum_{i=1}^{n+m} |d'_i|$ for all values of $x$ and $y$.
    
    First, suppose that $x = x_2$.
    We have that $y_1 \le x+\size{\J}s$ and $x+\size{\J}s \le y$, otherwise $I^{D'}_n$ intersects with $I^{D'}_{n+1}$.
    Moreover, $x+\size{\J}s < y$, otherwise $c(I^{D'}_{i+1}) - c(I^{D'}_{i}) = s$ holds when $i=n$.
    Hence $E(\J,x+\size{\J}s) \le E(\J,y)$ holds by \Cref{lem:tmd_convex}.
    If we set $(d_n,\ldots,d_{n+m}) = E_{x+\size{\J}s}(\J)$ and $(d_1,\ldots,d_n) = (d'_1,\ldots,d'_n)$, then $c(I^D_{i+1}) - c(I^D_{i}) = s$ and $\sum_{i=1}^{n+m} |d_i| \le \sum_{i=1}^{n+m} |d'_i|$ hold.

    Suppose now that $y = y_1$.
    Analogously, $x< y - \size{\J}s$, otherwise $I^{D'}_n$ intersects with $I^{D'}_{n+1}$ or $c(I^{D'}_{i+1}) - c(I^{D'}_{i}) = s$ holds when $i=n$.
    Again, $E(\J,y-\size{\J}s) \le E(\J,x)$ holds by \Cref{lem:tmd_convex}.
    If we set $(d_1,\ldots,d_n) = E_{y-\size{\J}s}(\I)$  and $(d_n,\ldots,d_{n+m}) = (d'_n,\ldots,d'_{n+m})$, then $c(I^D_{i+1}) - c(I^D_{i}) = s$ and $\sum_{i=1}^{n+m} |d_i| \le \sum_{i=1}^{n+m} |d'_i|$ hold.

    Lastly, suppose that $x<x_2$ and $y>y_1$.
    Let $x',y'$ be two arbitrary points such that $x\le x' \le x_2$, $y_1 \le y' \le y$ and $y' = x'+\size{\J}s$.
    These points exist since $y > x+\size{\J}s$ and $y_1 \le x_2+\size{\J}s$.
    If we set $(d_1,\ldots,d_n) = E_{x'}(\I)$ and $(d_n,\ldots,d_{n+m}) = E_{y'}(\J)$, then $\sum_{i=1}^{n+m} |d_i| \le \sum_{i=1}^{n+m} |d'_i|$ holds by \Cref{lem:tmd_convex}.
    Moreover, $c(I^D_{i+1}) - c(I^D_{i}) = s$ also holds since $y' = x' +\size{\J}s$. 
    
    In all cases, we obtained a vector $D = (d_1,\ldots,d_{n+m})$ in $\D^{\mathit{opt}}(\I)$ such that $c(I^D_{i+1}) - c(I^D_{i}) = s$ for $1\le i\le n+m-1$.
    Consequently there exists a $D \in \D^{\mathit{opt}}(\I\cup \J)$ such that $D = E_{x^*}(\I\cup \J)$ and $x^* \in \argmin_{x \in \mathbb{R}} E(\I\cup\J,x)$.
    Therefore, $\I \cup \J$ is optimally equispaceable if and only if $y_1 \le x_2 + \size{\J}s$.
\end{proof}

\Cref{cor:tmd_opt_no_intersect} is directly implied by \Cref{lem:tmd_opt_of_opts}.

\begin{corollary}\label{cor:tmd_opt_no_intersect}
    If $y_1 > x_2 + \size{\J} s$, then $\I \cup \J$ is not optimally equispaceable.
    Moreover, the minimum total moving distance for dispersing $\I \cup \J$ is equal to $E(\I,x) + E(\J,y)$ for arbitrary $x_1\le x\le x_2$ and $y_1\le y\le y_2$.
\end{corollary}

Given a collection $\I$ of $n$ unit intervals, we note that $\I$ can be partitioned into $m\le n$ subcollections $\I_{a_1,b_1},\ldots,\I_{a_{m},b_{m}}$ such that for all $1\le i \le m$, $c(I_{j+1})-c(I_{j}) \le s$ for $a_i \le j \le b_i -1$.
By \Cref{lem:partition_opt}, each $\I_{a_i,b_i}$ is an optimally equispaceable collection.
We use \Cref{lem:tmd_opt_of_opts} and prove the statement of \Cref{lem:consec_partition_opt}.

\begin{mlemmarep}\label{lem:consec_partition_opt}
    Let $\I = \set{I_1,\ldots,I_n} = \I_{a_1,b_1}\cup\cdots\cup\I_{a_{m},b_{m}}$ be a collection of $n$ unit intervals partitioned as above.
    If there exist integers $\alpha_1, \ldots, \alpha_k$ such that $\I_{a_{\alpha_i},b_{\alpha_i}}$ and $\I_{a_{\alpha_i + 1},b_{\alpha_i + 1}}$ intersect when equispaced,
    then there exists an optimal solution for dispersing $\I$ that disperses the intervals in a way that $c(I_{j+1})+d_{j+1} - (c(I_j) + d_j) = s$ holds for $1 \le i \le k$ and $a_{\alpha_i} \le j < b_{\alpha_i + 1}$.
\end{mlemmarep}
\begin{proof}
    Consider the subcollections $\I_{a_{\alpha_{i}},b_{\alpha_{i}}}$ and $\I_{a_{\alpha_{i}+1},b_{\alpha_{i}+1}}$.
    These subcollections intersect when equispaced, hence $\I_{a_{\alpha_{i}},b_{\alpha_{i}}} \cup \I_{a_{\alpha_{i}+1},b_{\alpha_{i}+1}}$ is optimally equispaceable by \Cref{lem:tmd_opt_of_opts}.
    Let $D^* \in \D^{\mathit{opt}}(\I_{a_{\alpha_{i}},b_{\alpha_{i}}} \cup \I_{a_{\alpha_{i}+1},b_{\alpha_{i}+1}})$.
    By the definition of $E$, we have that $c(I^{D^*}_{j+1}) - c(I^{D^*}_{j}) = s$ holds for $a_{\alpha_i} \le j < b_{\alpha_i + 1}$, which implies that the intervals are dispersed such that $c(I_{j+1})+d_{j+1} - \left(c(I_j) + d_j\right) = s$ holds for $1 \le i \le k$ and $a_{\alpha_i} \le j < b_{\alpha_i + 1}$.
    We only need to show that this dispersal is part of an optimal solution for dispersing $\I$.
    Without loss of generality, suppose that $\I_{a_{\alpha_{i}},b_{\alpha_{i}}} \cup \I_{a_{\alpha_{i}+1},b_{\alpha_{i}+1}}$ and $\I_{a_{\alpha_{i}+2},b_{\alpha_{i}+2}}$ intersect when equispaced. 
    Then, \Cref{lem:tmd_opt_of_opts} ensures that $\I_{a_{\alpha_{i}},b_{\alpha_{i}}} \cup \I_{a_{\alpha_{i}+1},b_{\alpha_{i}+1}} \cup \I_{a_{\alpha_{i}+2},b_{\alpha_{i}+2}}$ is optimally equispaceable.
    Applying \Cref{lem:tmd_opt_of_opts} recursively results in $k\le m$ partitions of $\I = \I_1,\ldots,\I_k$, where each $\I_i$ is dispersed to a point $x_i \in \argmin_{x\in \mathbb{R}} E(\I_i,x)$ and there exists no pair $\I_i,\I_{j}$, $i\neq j$, such that $\I_i$ and $\I_j$ intersect when equispaced.
    Moreover, the minimum total moving distance is given by $\sum_{i=1}^k E(\I_i,x_i)$ by \Cref{cor:tmd_opt_no_intersect}.
    In other words, $(x_1,\ldots,x_k)$ describes an optimal solution to disperse $\I$.
    Therefore, there exists an optimal solution to disperse $\I$ that disperses the intervals in a way that $c(I_{j+1})+d_{j+1} - \left(c(I_j) + d_j\right) = s$ holds for $1 \le i \le k$ and $a_{\alpha_i} \le j < b_{\alpha_i + 1}$.
\end{proof}

\paragraph*{Outline of \Cref{alg:dispersing-intervals}} 
Given a collection of unit intervals $\I$ and a dispersal value $s\ge1$, the algorithm starts by sorting and partitioning $\I$ into $m\le n$ disjoint subcollections $\I_{a_1,b_1},\ldots,\I_{a_{m},b_{m}}$ such that each $\I_{a_{i},b_{i}}$ satisfies \Cref{lem:partition_opt}.
Subsequently, the optimal breakpoints are determined for each $E(\I_{a_{i},b_{i}},x)$.
Whenever there exist two subcollections $\I_{a_{i},b_{j}},\:i\le j$ and $\I_{a_{k},b_{\ell}},\:k\le \ell$ that intersect when equispaced, the algorithm considers both subcollections as a unique subcollection $\I_{a_i,b_{\ell}} = \I_{a_i,b_{j}} \cup \I_{a_k,b_{\ell}}$ and recursively determines the optimal breakpoints of $E(\I_{a_i,b_{\ell}},x)$ using the breakpoint sets of $E(\I_{a_i,b_{j}},x)$ and $E(\I_{a_k,b_{\ell}},x)$.
%
%
\Cref{lem:consec_partition_opt} ensures that this recursion partitions $\I$ into non-intersecting subcollections when equispaced.
Lastly, the algorithm returns the total moving distance, which is calculated as the sum of the optimal values of $E$ for each subcollection.


Before showing the complexity of the algorithm, we must characterise the set of breakpoints further.
When a collection of unit intervals $\I =\set{I_1,\ldots,I_n}$ is partitioned into $m$ disjoint subcollections $\I_{a_1,b_1},\ldots,\I_{a_{m},b_{m}}$ of intervals that satisfy \Cref{lem:partition_opt}, 
the set of breakpoints $B_{\I_{a_i,b_i}}$ is equal to $\set{c(I_j)+(\size{\I_{a_i,b_i}}-j)s\mid I_i\in \I,\:a_i\le j \le b_i}$ for each $1\le i \le m$.
Consequently, $B_\I$ can be reformulated as follows:
\begin{gather*}
    B_\I = \left\{c(I_j) + \left(\size{\I_{a_i,b_i}}-j + \sum_{k=i+1}^m \size{\I_{a_k,b_k}}\right)s\mid 1\le i \le m,\: a_i\le j \le b_i\right\}.
\end{gather*}
As a result, if $b$ and $b'$ are the breakpoints for $I$ in $B_{\I_{a_i,b_i}}$ and $B_\I$, respectively, then $b' = b -\sum_{j=i+1}^m \size{\I_{a_j,b_j}}$ holds. 
Moreover, the breakpoints of any union of subcollections $\I_{a_i,b_j} = \I_{a_i,b_i}\cup\cdots\cup\I_{a_j,b_j}$ can be calculated in the same way by subtracting $\sum_{k=j+1}^m \size{\I_{a_k,b_k}}$ from any breakpoint $b\in B_\I$ calculated using an interval $I \in \I_{a_i,b_j}$.
It follows that the order of $B_{\I_{a_i,b_j}}$ is the same as the order of the corresponding breakpoints in $B_\I$.

The above implies that the breakpoints of any (union of) subcollection(s) can be obtained from $B_\I$.
We denote the set $\bigcup_{i\le k \le j}\left\{b+s\sum_{l=k+1}^m \size{\I_{a_l,b_l}}\mid b \in B_{\I_{a_k,b_k}}\right\}$ by $B^*_{\I_{a_i,b_j}}$ and call it the \emph{cumulative set of breakpoints of $B_{\I_{a_i,b_j}}$}.
We prove that $B^*_{\I_{a_1,b_1}},\ldots,B^*_{\I_{a_m,b_m}}$ can be found in $O(n\log n)$ time.

\begin{algorithm}[tb]
    \caption{Dispersing $n$ unit intervals in $O(n\log n)$ time.}
    \label{alg:dispersing-intervals}
    \Procedure{\rm{DispersingIntervals}($\I$,$s$)}{
        Sort and partition $\I$ into $m\le n$ subcollections $\I_{a_1,b_1},\ldots,\I_{a_{m},b_{m}}$ such that for all $1\le i \le m$, $c(I_{j+1})-c(I_{j}) \le s$ for $a_i \le j \le b_i -1$.\;\nllabel{alg:disp-sort}
        Compute and sort $B^*_{\I_{a_i,b_i}}$ for all $1\le i \le n$.\;\nllabel{alg:breakpoint_sort}
        $x^1_{a_i,b_{i}},x^2_{a_i,b_{i}}$ are the breakpoint $b^{\I_{a_i,b_{i}}}_{(n+1)/2}$ if $\size{\I_{a_i,b_{i}}}$ is odd and $b^{\I_{a_i,b_{i}}}_{n/2}$ and $b^{\I_{a_i,b_{i}}}_{(n/2)+1}$ otherwise.\;\nllabel{alg:breakpoint_1}
        $D^{\mathit{opt}} \gets \bigcup_{1\le i \le n} \set{(B^*_{\I_{a_i,b_i}},x^1_{a_i,b_{i}},x^2_{a_i,b_{i}})}$\;
        \While{$x^1_{a_{k},b_{\ell}}\le x^2_{a_{i},b_{j}}+\size{\I_{a_{k},b_{\ell}}}s$, $1\le i\le j < k\le \ell \le n$}{\nllabel{alg:disp-while}
            $B^*_{\I_{a_i,b_l}} \gets \mathit{merge}(B^*_{\I_{a_i,b_j}},B^*_{\I_{a_k,b_\ell}})$\;\nllabel{alg:merge}
            $x^1_{a_i,b_{\ell}}, x^2_{a_i,b_{\ell}} \gets b^{\I_{a_i,b_{\ell}}}_{(n+1)/2}$ if $\size{B^*_{\I_{a_i,b_{\ell}}}}$ is odd and $x^1_{a_i,b_{\ell}} \gets b^{\I_{a_i,b_{\ell}}}_{n/2}$,  $x^2_{a_i,b_{\ell}} \gets b^{\I_{a_i,b_{\ell}}}_{(n/2)+1}$ otherwise.\;
            $D^{\mathit{opt}} \gets \left(D^{\mathit{opt}} \setminus \left\{(B^*_{\I_{a_i,b_j}},x^1_{a_i,b_{j}},x^2_{a_i,b_{j}}),(B^*_{\I_{a_k,b_l}},x^1_{a_k,b_{\ell}},x^2_{a_k,b_{\ell}})\right\}\right) \cup \left\{(B^*_{\I_{a_i,b_\ell}},x^1_{a_i,b_{\ell}},x^2_{a_i,b_{\ell}})\right\}$\;\nllabel{alg:add_delete}
        }
        \Return $\sum_{(B^*_{\I_{a_i,b_j}},x_1,x_2) \in D^{\mathit{opt}}} E(\I_{a_i,b_i}\cup\cdots\cup \I_{a_j,b_j},x_1)$\;\nllabel{alg:disp_tmd_calc}
    }
\end{algorithm}

\begin{mlemmarep}\label{lem:breakpoint_sort}
    Let $\I = \set{I_1,\ldots,I_n} = \I_{a_1,b_1}\cup\cdots\cup\I_{a_{m},b_{m}}$ be a collection of $n$ unit intervals partitioned as above. 
    Then the cumulative sets of breakpoints $\B^*_{\I_{a_1,b_1}},\ldots,\B^*_{\I_{a_m,b_m}}$ such that each $B^*_{\I_{a_i,b_i}}$ is sorted can be obtained in $O(n\log n)$ total time.
\end{mlemmarep}
\begin{proof}
    Let $n_i$ be the size $\size{\I_{a_i,b_i}}$ and $n_{i,j}$ be the cumulative sum of sizes $n_i+\cdots+n_j$, for $i\le j$.
    We observe that $n_1+\cdots+n_m = n$ since the given partitions are disjoint.
    We first determine $n_{i,m}$ for each $2\le i \le m$ in $O(m)$ time.
    For each $1\le i \le m$, we compute $\B^*_{\I_{a_i,b_i}} = \set{c(I_j)+(n_i-j + n_{i+1,m})s\mid a_i\le j \le b_i}$ in $O(n_i)$ time.
    The total running time of this procedure is $O(n_1+\cdots+n_m) = O(n)$ time.
    We then sort $B^*_{\I_{a_i,b_i}}$ for each $1\le i \le m$ in $O(n_i \log n_i)$ time. The total running time $T(n_1,\ldots,n_m)$ is given as follows:
    \begin{align*}
        T(n_1,\ldots,n_m) & = n_1 \log n_1 + \cdots + n_m\log n_m\le n_1\log{n} + \cdots +n_m\log{n}\\
        &= (n_1+\cdots+n_m)\log{n}= n \log{n}.
    \end{align*}
    Therefore, obtaining sets $\B^*_{\I_{a_1,b_1}},\ldots,\B^*_{\I_{a_m,b_m}}$ such that each $\B^*_{\I_{a_i,b_i}}$ is sorted can be done in $O(n\log n)$ total time.
\end{proof}
\begin{mlemmarep}\label{lem:total_merge_complexity}
    Let $\I = \set{I_1,\ldots,I_n} = \I_{a_1,b_1}\cup\cdots\cup\I_{a_{m},b_{m}}$ be a collection of $n$ unit intervals partitioned as above.
    If cumulative breakpoint sets $B^*_{\I_{a_1,b_1}},\ldots,B^*_{\I_{a_m,b_m}}$ are given so that each $B^*_{\I_{a_i,b_i}}$ is sorted, then merging them into one sorted set can be done in $O(n\log n)$ total time.
\end{mlemmarep}
\begin{proof}
    We proceed using an unbalanced merge sort approach.
    Given two sorted sets $A$ and $B$ of numbers, it is known that $A$ and $B$ can be merged into one sorted set in $O(\size{A} \log{(\size{B}/\size{A})})$ time, assuming that $\size{A}\le \size{B}$~\cite{Brown1979}.
    We use this algorithm and show that the sets $B^*_{\I_{a_1,b_1}},\ldots,B^*_{\I_{a_m,b_m}}$ can be merged in $O(n \log n)$ total time.
    Let $n_i$ be the size $\size{\I_{a_i,b_i}}$ and $n_{i,j}$, $i\le j$, be the cumulative sum of sizes $n_i+\cdots+n_j$.
    We prove the lemma by induction on $m$.
    Given $m = 1$, no sort is performed since each $B^*_{\I_{a_i,b_i}}$ is already sorted, implying that the time is bounded by $n \log n$.
    Thus, we assume that the lemma holds for $1 < p \le m -1$ sets and prove that it also holds for $p = m$.
    Without loss of generality, assume that the merging of $B^*_{\I_{a_1,b_1}},\ldots,B^*_{\I_{a_m,b_m}}$ is done by merging two already merged sets $B^*_{\I_{a_1,b_i}} = B^*_{\I_{a_1,b_1}}\cup \cdots\cup B^*_{\I_{a_i,b_i}}$ and $B^*_{\I_{a_{i+1},b_m}} = B^*_{\I_{a_{i+1},b_{i+1}}}\cup \cdots \cup B^*_{\I_{a_m,b_m}}$, for an arbitrary $1\le i \le m-1$.
    Thus, we have $\size{B^*_{\I_{a_1,b_i}}},\size{B^*_{\I_{a_{i+1},b_m}}}\le m-1$ and $n = n_{1,i} + n_{i+1,m}$.
    Let $T(a,b)$ denote the number of steps necessary to merge two sets of size $a$ and $b$, $n_{1,i}$ be $ n_{1,j}+n_{j+1,i}$ and $n_{i+1,m}$ be $n_{i+1,k}+n_{k+1,m}$ for arbitrary $1\le j \le i$ and $i+1\le k \le m$.
    Without loss of generality, assume that $n_{1,i} \le n_{i+1,m}$.
    The value of $T(n_{1,i},n_{i+1,m})$ is bounded as follows:
    \begin{align*}
    T(n_{1,i}&,n_{i+1,m}) =\\
    & = T(n_{1,j},n_{j+1,i})+T(n_{i+1,k},n_{k+1,m}) + n_{1,i} \log{(n_{i+1,m}/n_{1,i})}\\
    & \le n_{1,i} \log n_{1,i} + n_{i+1,m} \log n_{i+1,m} + n_{1,i} \log{(n_{i+1,m}/n_{1,i})} & \text{(IH)}\\
    & = n_{1,i} \log n_{1,i} + n_{i+1,m} \log n_{i+1,m} + n_{1,i}\log{n_{i+1,m}} - n_{1,i} \log{n_{1,i}}\\
    & = n_{i+1,m} \log n_{i+1,m} + n_{1,i}\log{n_{i+1,m}}= (n_{1,i} + n_{i+1,m})\log{n_{i+1,m}}\\
    & \le n\log n.
    \end{align*}
    We have proved that if the lemma is true for $p \le m -1$ sets, then the lemma is also true for $p =m$.
    This concludes the induction and lemma proof.
\end{proof}

\begin{theorem}\label{thm:interval_dispersion}
    Given a collection of unit intervals $\I$ and a value $s\ge 1$, {\idisp} can be solved in $O(n\log n)$ time.
\end{theorem}
\begin{proof}
    We show the complexity of \Cref{alg:dispersing-intervals}. 
    Line~\ref{alg:disp-sort} can be done in $O(n\log n)$ time for sorting and $O(n)$ time to determine the initial $m$ partitions.
    Similarly, line~\ref{alg:breakpoint_sort} can be done in $O(n\log n)$ time by \Cref{lem:breakpoint_sort}.
    Given that each $\B^*_{\I_{a_i,b_i}}$ is sorted, the $((\size{\I_{a_i,b_i}}+1)/2)$th element (resp. $(\size{\I_{a_i,b_i}}/2)$th and $((\size{\I_{a_i,b_i}}/2)+1)$th element) can be calculated in $O(\log \size{\I_{a_i,b_i}})$ time using binary search on $\B^*_{\I_{a_i,b_i}}$.
    This ensures that line~\ref{alg:breakpoint_1} is done for all $1\le i \le m$ in $O(m\log n)$ total time.
    We initialise $D^{\mathit{opt}}$ as a doubly linked list where each node $i$ contains the information of $(B^*_{\I_{a_i,b_i}},x^1_{a_i,b_{i}},x^2_{a_i,b_{i}})$.
    We show the complexity of the loop in line~\ref{alg:disp-while}.
    We merge both $B^*_{\I_{a_i,b_j}}$ and $B^*_{\I_{a_k,b_\ell}}$ to obtain a sorted $B^*_{\I_{a_i,b_\ell}}$.
    Hence, the median value(s) of $B^*_{\I_{a_i,b_\ell}}$ can be calculated in $O(\log n)$ time by binary search.
    At each execution of line~\ref{alg:merge}, two partitions are merged; thus the number of partitions is reduced by one unit at each iteration.
    Initially, there exist $m$ partitions, and hence the loop of line~\ref{alg:disp-while} iterates at most $m-1$ times.
    Moreover, merging $m$ cumulative sets of breakpoints into one sorted set can be done in $O(n\log n)$ time by \Cref{lem:total_merge_complexity}, which implies that any partial merge of these sets is also bounded by $O(n\log n)$.
    Consequently, the total running time of line~\ref{alg:disp-while} is $O(n\log n)$ time.
    Lastly, in line~\ref{alg:add_delete} the two merged sets are deleted and the new one is added. 
    Since $D^{\mathit{opt}}$ is a doubly linked list, this can be done in $O(1)$ time by connecting the previous and next node of $B^*_{\I_{a_i,b_j}}$ and $B^*_{\I_{a_k,b_\ell}}$ to a new node containing $B^*_{\I_{a_i,b_\ell}}$, respectively.
    Once there is no pair of subcollections left to merge, the total moving distance is calculated in $O(n)$ time in line~\ref{alg:disp_tmd_calc} following the definition of cumulative set of breakpoints, which concludes that the total running time of \Cref{alg:dispersing-intervals} is $O(n\log n)$ time.
\end{proof}


\Cref{thm:interval_dispersion} implies the following result for satisfying $\Pi_{\texttt{edgeless}}$ on unit interval graphs when $s=1$.

\begin{mcorollaryrep}
    Given a unit interval graph $(G,\I)$, {\gged} can be solved in $O(n\log n)$ time for satisfying $\Pi_{\texttt{edgeless}}$.
\end{mcorollaryrep}
\begin{proof}
    If $G$ is edgeless, then there exists an optimal solution that satisfies $c(I_{i+1}) - c(I_i) \ge 1$ for $1\le i \le n-1$. 
    Dispersing the intervals using $s = 1$ results in intervals separated by a distance of at least one unit. 
    That is, the resulting unit interval graph satisfies the edgeless condition. Therefore, \Cref{thm:interval_dispersion} works for satisfying $\Pi_{\texttt{edgeless}}$ on unit intervals graphs when $s=1$. 
\end{proof}

\subsection{Satisfying \texorpdfstring{$\Pi_{\texttt{acyc}}$}{} and \texorpdfstring{$\overline{\Pi_{k\texttt{-clique}}}$}{} on Unit Interval Graphs}
\label{ssec:acyc_kclique_uig}

This section shows how to use \Cref{alg:dispersing-intervals} for satisfying $\Pi_{\texttt{acyc}}$ and $\overline{\Pi_{k\texttt{-clique}}}$ on unit interval graphs. 
We first show how to satisfy $\overline{\Pi_{k\texttt{-clique}}}$. 

It is shown in~\cite{HonoratoDroguett2024} that given a unit interval graph $(G,\I)$, $G$ does not contain a $k$-clique if and only if $c(I_{i+k-1}) - c(I_i) \ge  1$ for all $1 \le i \le n-k+1$.
%
%
%
%
This inequality can be decomposed into $k-1$ inequalities of the following form:
    for each $0\le r \le k-2$, $c(I_{i+k-1}) - c(I_i) \ge 1$ for all $1\le i \le n-k+1$ such that $i\bmod{k-1} = r$.
If $\I$ is decomposed into $k-1$ subcollections such that $\I = \bigcup_{1\le i \le k-1} \I_i$, $\I_i = \{I_j \in \I\mid 1\le j \le n,\: j\pmod{k-1}=i\}$, then \Cref{alg:dispersing-intervals} can be applied to each $\I_i$ independently for $s = 1$ to satisfy the above inequalities.
Since unit interval graphs are chordal, $G$ is acyclic if it is triangle-free; i.e. $G$ is contained in $\overline{\Pi_{3\texttt{-clique}}}$.
Consequently $\Pi_{\texttt{acyc}}$ can be satisfied by satisfying $\overline{\Pi_{k\texttt{-clique}}}$ when $k = 3$.
The above ideas imply \Cref{cor:nokclique}.

\begin{corollary}\label{cor:nokclique}
    Given an interval graph $(G,\I)$, {\gged} can be solved in $O(n\log n)$ time for satisfying $\Pi_{\texttt{acyc}}$ and $\overline{\Pi_{k\texttt{-clique}}}$.
\end{corollary}

\ifFull
An interval graph $G$ is bipartite if $G$ does not contain an odd cycle. Since interval graphs are chordal, any existence of a cycle implies also the existence of an odd cycle. Thus it is sufficient to remove all cycles to obtain a bipartite graph. 
\begin{corollary}
    Given a unit interval graph $(G,\I)$, {\gged} can be solved in $O(n\log n)$ time for satisfying $\Pi_{\texttt{bipar}}$.
\end{corollary}
\fi

\section{Minimising the Total Moving Distance for Satisfying \texorpdfstring{$\Pi_{\texttt{edgeless}}$}{} on Weighted Interval Graphs is Hard}\label{sec:edg_ig}

In this section we show that {\gged} is strongly \NP-hard on weighted interval graphs for satisfying $\Pi_{\texttt{edgeless}}$. We show a reduction from {\threepartition}~\cite{garey1979}.
\ifConf
{\threepartition} receives as input a set $A$ of $3m$ elements, a bound $B \in \mathbb{Z}^+$ and a size $s(a) \in \mathbb{Z}^+$ such that $B/4 <s(a) < B/2$ and $\sum_{a \in A} s(a) = mB$, and the task is to decide whether $A$ can be partitioned into $m$ disjoint sets $A_1,\ldots,A_m$ such that for $1\le i \le m$, $\size{A_i} = 3$ and $\sum_{a\in A_i} s(a) = B$.
\fi
\ifFull
\begin{itembox}[l]{{\threepartition}~\cite{garey1979}}
    \begin{description}
        \item[Input:] Set $A$ of $3m$ elements; a bound $B \in \mathbb{Z}^+$; a size $s(a) \in \mathbb{Z}^+$ such that $B/4 <s(a) < B/2$ and $\sum_{a \in A} s(a) = mB$.
        \item[Task:] Decide whether $A$ can be partitioned into $m$ disjoint sets $A_1,\ldots,A_m$ such that for $1\le i \le m$, $\size{A_i} = 3$ and $\sum_{a\in A_i} s(a) = B$.
    \end{description}
\end{itembox}
\fi

Given an instance $(A,B,s)$ of {\threepartition}, we construct a collection of intervals $\I_A$ and show that $A$ can be partitioned if and only if $\Pi_{\texttt{edgeless}}$ can be satisfied on $\I_A$ with at most total moving distance $T$.
Given two intervals $I,I'$ such that $c(I)\le c(I')$, we say that $I$ and $I'$ intersect if $c(I')-c(I) < (\len{I'}+\len{I})/2$.

We show the construction of $\I_A$ (see \Cref{fig:reduction_overview_ig_hard}).
We define $\I_A$ as the collection $\I\cup \I^s \cup \I^b$ where $\I = \{I_1,\ldots,I_{3m}\},\: \I^s = \{I^s_1,\ldots,I^s_{m-1}\},\: \I^b =  \{I_\ell,I_r\}$ and,
\begin{description}
    \item[(i)] for $1\le i\le 3m$, $I_i$ is an interval such that $\len{I_i} = s(a_i)$ and $c(I_i) = -s(a_i)/2$ (that is, $r(I_i) = 0$),
    \item[(ii)] for $1\le i \le m-1$, $I^s_i$ is an interval such that $\len{I^s_i} = B$ and $c(I^s_i) = (2i-1)B + B/2$ and
    \item[(iii)] $I_\ell$ and $I_r$ are intervals such that $\len{I_\ell} = \len{I_r} = 3Bm^2 + \max_{a\in A}{s(a)}$, $c(I_\ell) = -3Bm^2/2$ and $c(I_r) = (2m-1)B+3Bm^2/2$.
\end{description}
    
    \begin{figure}[!b]
        \centering
        \includegraphics[scale=1,page=1]{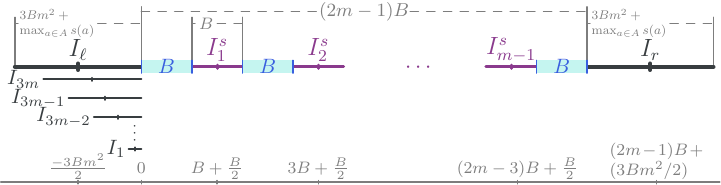}
        \caption{Reduction Overview}
        \label{fig:reduction_overview_ig_hard}
    \end{figure}  
    
For an interval $I\in \I_A$, we define the moving distance function $d_I:\mathbb{R}\rightarrow\mathbb{R}$ as:
\begin{align*}
    d_I(x) = \begin{cases}
        |c(I)-x|,&\quad I \in \I,\\
        12Bm^2|c(I)-x|,& \quad I \in \I^s \cup \I^b.
    \end{cases}
\end{align*}
Given an instance $(A,B,s)$ of {\threepartition}, we show the following properties.
\begin{mlemmarep}\label{lem:cumulative_sum_bound}
    Given an arbitrary partition of $A$ of $m$ disjoint sets $A_1,\ldots,A_m$ such that $A_i =\set{a^i_1,a^i_2,a^i_3}$ for $1\le i \le m$, $\sum_{i=1}^m 6(i-1)B + \sum^m_{i=1} (3a^i_1+2a^i_2+a^i_3) < 3Bm^2$ holds.
\end{mlemmarep}
\begin{proof}
    We first simplify the first sum:
    \begin{align*}
        \sum_{i=1}^m 6(i-1)B & = 6B\left(\sum_{i=1}^m i-1\right) = 6B\left(\sum_{i=1}^m i - \sum_{i=1}^m 1\right)\\
        & = 6B\left(\frac{m(m+1)}{2} - m\right) = 6B\left(\frac{m(m-1)}{2}\right).
    \end{align*}
    We obtain the upper bound using the fact that $s(a)< B/2$ for any $a \in A$:
    \begin{align*} \sum^m_{i=1} (3a^i_1+2a^i_2+a^i_3) <\sum_{i=1}^m 3\frac{B}{2}+2\frac{B}{2}+\frac{B}{2} = \sum_{i=1}^m 6\frac{B}{2} = \sum_{i=1}^m 3B = 3mB.
    \end{align*}
    Now we have that
    \begin{align*}
        \sum_{i=1}^m 6(i-1)B + \sum^m_{i=1} (3a^i_1+2a^i_2+a^i_3) & < 6B\left(\frac{m(m-1)}{2}\right) + 3mB\\
        & = 3Bm^2 -3mB +3mB = 3Bm^2.
    \end{align*}
    Therefore, the lemma statement is true.
\end{proof}
We note that \Cref{lem:cumulative_sum_bound} works for any partition of $A$ as described above, even without the restrictions of the {\threepartition} output.
\begin{mlemmarep}\label{lem:3p_iff_gged_edgeless}
    Given an instance $(A,B,s)$ of {\threepartition}, $A$ can be partitioned into $m$ disjoint sets $A_1,\ldots,A_m$ such that for $1\le i \le m$ $A_i = \set{a^i_1,a^i_2,a^i_3}$, $\size{A_i} = 3$ and $\sum_{a\in A_i} s(a) = B$ if and only if $\Pi_{\texttt{edgeless}}$ can be satisfied on $\I_A$ with total moving distance of at most $3Bm^2$.
\end{mlemmarep}
\begin{proof}
    Assume that $A$ can be partitioned into $m$ disjoint sets $A_1,\ldots,A_m$ such that for $1\le i \le m$, $\size{A_i} = 3$ and $\sum_{a\in A_i} s(a) = B$.
    Let $D = (d_1,\ldots,d_{3m}),\:D^s = (d^s_{1},\ldots, d^s_{m-1}),\:D^b = (d^b_\ell, d^b_r)$ be vectors that describe the moving distances of $\I$, $\I^s$ and $\I^b$ for satisfying $\Pi_{\texttt{edgeless}}$, respectively.
    We show that $D,D^s,D^b$ exist such that $\sum_{d\in D\cup D^s \cup D^b} |d| \le 3Bm^2$.
    
    Without loss of generality, we assume that the first three elements of $\I$ correspond to $A_1$, the next three elements to $A_2$, and so forth.
    We start by setting $D^s = 0$ and $D^b = 0$.
    We set $d_1 = a^1_1 $, $d_2 = a^1_1 + a^1_2 $ and $d_3 = a^1_1 + a^1_2 + a^1_3$.
    The centre of $I^D_1$ is equal to $c(I^D_1) = c(I_1) + d_1 = -a^1_1/2 + a^1_1 = a^1_1/2$.
    Similarly, $c(I^D_2) = a^1_1 + a^1_2/2$ and $c(I^D_3) = a^1_1 + a^1_2 + a^1_3/2$.
    Furthermore, it holds that $c(I^D_{i+1}) - c(I^D_i) = (\len{I^D_{i+1}}+\len{I^D_i})/2$ for $i\in \set{1,2}$, hence $I^D_1,I^D_2,I^D_3$ do not intersect each other.
    Moreover, $c(I^D_1) - c(I_\ell) = (\len{I^D_1}+\len{I_\ell})/2$ and $c(I^s_1) - c(I^D_3) = (\len{I^s_1}+\len{I^D_3})/2$.
    That is, $I_1,I_2,I_3$ were moved to the area of length $B$ between $I_\ell$ and $I^s_1$ without introducing new intersections.
    For $2 \le i \le m$, we set $d_{3i-2} = 2(i-1)B + a^i_1$, $d_{3i-1} = 2(i-1)B + a^i_1 + a^i_2$ and $d_{3i-1} = 2(i-1)B + a^i_1 + a^i_2 + a^i_3$.
    Analogous to $I_1,I_2,I_3$, it is easy to see that $I^D_{3i-2}, I^D_{3i-1}, I^D_{3i}$ are moved to the area of length $B$ between $I^s_{i-1}$ and $I^s_{i}$ without introducing intersections.
    This implies that $D$, $D^s$, and $D^b$ describe moving distances to make $\I_A$ satisfy $\Pi_{\texttt{edgeless}}$.
    For $1\le i \le m$, the total moving distance for moving $I_{3i-2},I_{3i-1},I_{3i}$ as described is given by $d_{3i-2}+d_{3i-1}+d_{3i} = 6B(i-1) + 3a^i_1+2a^i_2+a^i_3$.
    Consequently, $T = \sum D + \sum D^s + \sum D^b = \sum D = \sum_{i=1}^m 6B(i-1) + 3a^i_1+2a^i_2+a^i_3$.
    By \Cref{lem:cumulative_sum_bound}, $T< 3Bm^2$ holds.
    Therefore, $\Pi_{\texttt{edgeless}}$ can be satisfied on $\I_A$ with total moving distance of at most $3Bm^2$.

    In the other direction, assume that $\Pi_{\texttt{edgeless}}$ can be satisfied on $\I_A$ with total moving distance of at most $3Bm^2$.
    We let $D,D^s,D^b$ be the vectors that describe such a solution.
    We show that $A$ can be partitioned into $m$ disjoint sets $A_1,\ldots,A_m$ such that for $1\le i \le m$, $A_i = \set{a^i_1,a^i_2,a^i_3}$, $\size{A_i} = 3$ and $\sum_{a\in A_i} s(a) = B$.
    Observe that if an interval $I\in \I$ is moved to a point $x \le \ell(I_\ell)$, then $d_I(x) \ge \len{I_\ell}-\len{I}/2 = 3Bm^2 + \max_{a\in A}{s(a)} - \len{I}/2> 3Bm^2$.
    Thus no interval is moved to the left side of $I_\ell$.
    Analogously, no interval is moved to the right side of $I_r$.
    Moreover, for $I\in \I^b$, $c(I)-1/4\le c(I^{D^b})\le c(I)+1/4$ holds, otherwise $d_I(x) > 3Bm^2$ by the definition of the moving distance function.
    This argument also holds for any $I \in \I^s$, since the moving distance is the same.
    Hence the intervals in $\I^D$ must be between $r(I_\ell)$ and $\ell(I_r)$.
    
    There exist $m$ areas of length $B$ between $r(I_{\ell})$ and $\ell(I_r)$ divided by the $m-1$ intervals of $\I^s$. 
    By definition of $\I_A = \I \cup \I^s\cup \I^b$, $\sum_{I\in \I } \len{I} = mB$ and $\size{\I} = 3m$. 
    For any subcollection of intervals $\I' \subseteq \I$ such that $\size{\I'} \ge 4$, the constraint $B/4 < s(a) < B/2$ ensures that $\sum_{I \in \I'} \len{I} > B$ holds.
    Consequently, $\I$ must admit a partition into $m$ subcollections $\I_1,\dots,\I_m$ of three intervals such that $\sum_{I \in \I_i} \len{I} = B$ for each $i$.
    Otherwise the intervals do not fit into the $m$ areas divided intervals of $\I^s$.
    On the other hand, it holds that the length of an area is in the range of $B\pm 1/2$ since $c(I)-1/4\le c(I^{D^b})\le c(I)+1/4$ holds for $I \in \I^s\cup \I^b$.
    Given that $s(a) \in \mathbb{Z}^+$ for all $a \in A$, this ensures that regardless of the movement of intervals in $\I^s$, the three intervals in each area must satisfy $\sum_{I \in \I_i} \len{I} = B$ for each $i$.
    Without loss of generality, assume that $\I$ is moved such that $c(I^D_{i+1}) \ge c(I^D_i)$ for $1\le i \le 3m-1$ and that $\len{I^D_i} = s(a'_i)$ for $a'_i \in A$.
    The $m$ disjoint subcollections $\set{I_{3i-2},I_{3i-1},I_{3i}}$, $1\le i \le m$, satisfy $\len{I_{3i-2}}+\len{I_{3i-1}}+\len{I_{3i}} = B$. That is, this partition gives $m$ disjoint subsets of the form $A_i = \set{a'_{3i-2},a'_{3i-1},a'_{3i}}$ such that $\sum_{a\in A_i} s(a)  = B$.
    Therefore, $A$ can be partitioned into $m$ disjoint sets $A_1,\ldots,A_m$ such that for $1\le i \le m$ $A_i = \set{a^i_1,a^i_2,a^i_3}$, $\size{A_i} = 3$ and $\sum_{a\in A_i} s(a) = B$.
\end{proof}
Lastly, we remark that the polynomial construction of $\I_A$ is straightforward by iterating over $A$ and following the definitions given at the beginning of the section. We summarise the main result of this section as follows:
\begin{theorem}\label{thm:ig_nphard_edgeless}
    {\gged} is strongly \NP-hard on weighted interval graphs for satisfying $\Pi_{\texttt{edgeless}}$.
\end{theorem}

\ifFull
Let $\Pi_{k\texttt{-deg}}$ be the class of graphs with maximum degree $k$. 
We slightly modify the reduction of \Cref{thm:ig_nphard_edgeless} and show that {\gged} is also strongly \NP-hard for satisfying $\Pi_{k-\texttt{deg}}$.

\begin{mtheoremrep}\label{thm:ig_nphard_kdeg}
    {\gged} is strongly \NP-hard on weighted interval graphs for satisfying $\Pi_{k-\texttt{deg}}$.
\end{mtheoremrep}
\begin{proof}
    We extend the reduction from {\threepartition} of \Cref{thm:ig_nphard_edgeless}.
    Given an instance $(A,B,s)$ of {\threepartition} and $\I_A = \I \cup \I^s \cup \I^b$ as defined above, we construct an instance $\J_A = \J \cup \J^s \cup \J^b \cup \J^f$ such that (i) $\J = \I$, (ii) for each $I^s_i \in \I^s$, $\J$ contains $k+1$ intervals $J_1,\ldots,J_{k+1}$ such that $\len{J_j} = \len{I^s_i}$ and $c(J_j) = c(I^s_i)$ for $1\le j \le k+1$, and (iii) $\J^b = \set{J_\ell^1,\ldots,J_\ell^{k+1}}\cup \set{J_r^1,\ldots,J_r^{k+1}}$ such that $\len{J_\ell^j} = \len{I_\ell}$ and $c(J^j_\ell) = c(I_\ell)$ for $1\le j \le k+1$, and $\len{J_r^j} = \len{I_r}$ and $c(J^j_r) = c(I_r)$ for $1\le j \le k+1$.
    We define the moving distance function analogously for $\J$ and $\J^s \cup \J^b$.
    Let $c_1,\ldots,c_m = (\ell(I^s_1)-r(I_\ell))/2,(\ell(I^s_2)-r(I^s_1))/2,\ldots, (\ell(I^s_{m-1})-r(I^s_{m-2}))/2,(\ell(I_r)-r(I^s_{m-1}))/2$.
    For each $1\le i \le m$, the subcollection $\J^f$ contains $k$ intervals $J_1,\ldots,J_k$ such that $c(J_j) = c_i$, $\len{J_j} = B$ and $d_{J_j}(x) = 12Bm^2|c(J_j)-x|$ for all $1\le j \le k$.
    The resulting collection consists of a $(k+1)$-clique for each copy of an interval $I \in \I^s\cup \I^b$ and $k$-cliques of intervals centred at $c_1,\ldots,c_m$ in the areas of size $B$.
    We set the maximum total moving distance to $T = 3Bm^2$ as above.
    Analogously to $\I_A$, the intervals in $\J$ must be placed in the areas of the $k$-cliques; otherwise, a $(k+2)$-clique is formed by intersecting with the intervals in $\J^s \cup \J^b$ or the total moving distance is greater than $3Bm^2$.
    This implies that the proof of \Cref{lem:3p_iff_gged_edgeless} also shows that $\Pi_{k-\texttt{deg}}$ can be satisfied on $\J_a$ with total moving distance of at most $3Bm^2$ if and only if $(A,B,s)$ is a yes-instance of {\threepartition}.
    Lastly, $\J_A$ is constructed in polynomial time by iterating $A$ since $k$ is bounded by $n$.
    Therefore, the theorem statement is true.
\end{proof}

We notice that \Cref{thm:ig_nphard_kdeg} can be used to show that the cases for properties $\Pi_{\texttt{acyc}}$ and $\overline{\Pi_{k\texttt{-clique}}}$ are also strongly \NP-hard.
In particular, for $\Pi_{1\texttt{-deg}}$, the intervals $I,J^1_\ell,J^2_\ell$ form a cycle in $\J_A$ for any $I \in \J$. Consequently, moving the intervals of $\J$ with total moving distance of at most $3Bm^2$ is equivalent to removing all cycles from $\J_A$ with total moving distance of at most $3Bm^2$.
Similarly, for any $I\in \J$, the intervals $I,J^1_\ell,\ldots,J^k_\ell$ form a $k$-clique in $\I_A$ for satisfying $\Pi_{(k-1)\texttt{-deg}}$. Consequently, moving the intervals of $\J$ with total moving distance of at most $3Bm^2$ is equivalent to removing all $k$-cliques from $\J_A$ with total moving distance of at most $3Bm^2$.
\fi
\ifConf
We notice that \Cref{thm:ig_nphard_edgeless} can be extended to show that satisfying $\Pi_{\texttt{acyc}}$ and $\overline{\Pi_{k\texttt{-clique}}}$ is also strongly \NP-hard.
In particular, when satisfying $\overline{\Pi_{k\texttt{-clique}}}$, we create $k-1$ overlapping copies of the intervals in $\I^s\cup\I^b$ and add $k-1$ overlapping intervals of size $B$ into the spaces between intervals of $\I^s\cup\I^b$ with the same moving distance function.
Any interval forms a $k$-clique with the $k$ copies of overlapping intervals.
Consequently, moving the intervals of $\I$ with total moving distance of at most $3Bm^2$ is equivalent to removing all $k$-cliques from $\I_A$ with at most the same distance.
Moreover, by the chordality of interval graphs, it is sufficient to satisfy $\overline{\Pi_{k\texttt{-clique}}}$ when $k =3$ to satisfy $\Pi_{\texttt{acyc}}$.
\fi
As a result, \Cref{cor:ig_nphard_acyc_and_nokclique} is obtained.

\begin{corollary}\label{cor:ig_nphard_acyc_and_nokclique}
    {\gged} is strongly \NP-hard on weighted interval graphs for satisfying $\Pi_{\texttt{acyc}}$ and $\overline{\Pi_{k\texttt{-clique}}}$.
\end{corollary}

\section{Minimising the Maximum Moving Distance for \texorpdfstring{$\Pi_{\texttt{edgeless}}$}{} on Unit Disk Graphs is Hard}\label{sec:disk_edgeless}

    \newcommand{\bdisk}[1]{B\langle #1\rangle }
    \newcommand{\ldisk}[1]{L\langle #1\rangle }
    \newcommand{\hdisk}[1]{H\langle #1\rangle }
    \newcommand{\cstate}{\D_\Phi}
    \newcommand{\istate}{\D^{(i)}_\Phi}
    \newcommand{\mstate}{\D^{(m)}_\Phi}
    \newcommand{\mstatefinal}[1]{(#1)^{(\mathit{moved})}}
    \newcommand{\ismoved}[1]{\X(#1)}
    \newcommand{\movedpos}[1]{\X_{\mathit{pos}}(#1)}
%
%
In this section, we deal with the minimax version of {\gged}, defined as follows:
%
\begin{itembox}[l]{{\ggedmm}}\label{pro:edg_disk}
    \begin{description}
        \item[Input:] An intersection graph $(G,\S)$, a graph property $\Pi$ and a real $K>0$.
        \item[Task:] Decide whether $\Pi$ can be satisfied by moving objects such that for all $S\in \S$, the moving distance of $S$ is at most $K$.
    \end{description}
\end{itembox}

We show that {\ggedmm} is strongly {\NP-hard} on unit disk graphs for satisfying $\Pi = \Pi_{\texttt{edgeless}}$ over the $L_1$ and $L_2$ distances by reducing from {\pthreesat}.
Specifically, we show a proof for \Cref{thm:edgeless_np_hard}.

\begin{restatable}{theorem}{edgelessNPHard}\label{thm:edgeless_np_hard}
    {\ggedmm} is strongly {\NP-hard} on unit disk graphs for satisfying $\Pi_{\texttt{edgeless}}$ over the $L_1$ and $L_2$ distances.
\end{restatable}

\ifConf
Due to space constraints, we only give an overview of the reduction. The complete reduction and proofs can be found in the full-version of the paper~[].
\fi
\subsection[Proof Overview of \texorpdfstring{\Cref{thm:edgeless_np_hard}}{}]{Proof Overview of \texorpdfstring{\Cref{thm:edgeless_np_hard}}{}: Reducing {\pthreesat} to {\ggedmm}}
\ifConf
    We show a reduction from the following \NP-complete variation of {\pthreesat}~\cite{Lichtenstein1982,Knuth1992,Tovey1984}.
    Given CNF formula $\Phi$ equipped with a planar rectilinear embedding $G_\Phi$, a set $X$ of $n$ variables, a set $C$ of $m$ clauses over $X$ such that each $c \in C$ has length $|c| \le 3$, each variable $x \in X$ appears in at most three clauses, and $\Phi = \bigwedge_{c\in C} c$, {\pthreesat} asks whether $\Phi$ is satisfiable.
\fi
\ifFull
Given a boolean formula $\Phi$ and its planar incidence graph $G_{\Phi}$, {\pthreesat}~\cite{Lichtenstein1982} asks whether $\Phi$ is satisfiable. 
An instance of {\pthreesat} can always be described using a rectilinear embedding \cite{Knuth1992}.
Moreover, this problem is \NP-complete even if the appearance of variables in clauses is restricted to at most three \cite{Tovey1984}.
We use these restrictions in the reduction and define the problem as follows.
\begin{itembox}[l]{{\pthreesat}}
    \begin{description}
        \item[Input:] A CNF formula $\Phi$ equipped with a planar rectilinear embedding $G_\Phi$. Set $X$ of $n$ variables, set $C$ of $m$ clauses over $X$ such that each $c \in C$ has length $|c| \le 3$, each variable $x \in X$ appears in at most three clauses, and $\Phi = \bigwedge_{c\in C} c$.
        \item[Task:] Decide whether $\Phi$ is satisfiable.
    \end{description}
\end{itembox}
\fi
\ifFull
\paragraph*{Reduction Overview} 
\fi
We give a simplified overview of the reduction. The idea is to emulate each component (clauses, variables and connectors) of $G_{\Phi}$ using \emph{disk gadgets} and construct a collection of disks $\D_\Phi$ equivalent to $G_{\Phi}$. 
That is, our objective is to construct a $\D_\Phi$ such that $\Phi$ is satisfiable if and only if $\D_\Phi$ is a yes-instance of {\ggedmm} for $\Pi_{\texttt{edgeless}}$.
To do this, we emulate the truth assignment using a proper movement of disks. 
To force the disk movement, we deliberately insert intersecting disks in $\D_\Phi$. 
In particular, we insert intersecting disks in clause gadgets and restrict the movement of such disks to moving a sequence of disks such that a \emph{free slot} of a variable gadget is used.
To allow the removal of the intersection, the gadgets are connected following the structure of $G_\Phi$ using consecutive disks separated by distance $K$.
%
%
For example, consider the boolean formula $\Phi$ 
and its rectilinear embedding $G_\Phi$, illustrated in \Cref{fig:reduction_overview_a}. 
\begin{figure}[!bt]
    \centering
    \includegraphics[scale=1,page=32]{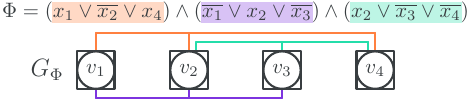}
    \caption{Reduction Overview: An arbitrary instance $\Phi$ of {\pthreesat} with its rectilinear embedding $G_\Phi$.}
    \label{fig:reduction_overview_a}
\end{figure}
A skeleton of the reduction is shown in \Cref{fig:reduction_overview_bc}(a), where representations of clause and variable gadgets are connected following $G_\Phi$. 
\begin{figure}[!tb]
    \centering
    \includegraphics[scale=1,page=33]{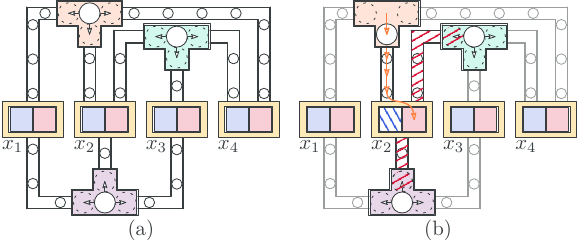}
    \caption{Reduction Overview: (a) The skeleton given by the instance $(\Phi, G_\Phi)$ of \Cref{fig:reduction_overview_a}; (b) The intersection of the gadget for $c = (x_1 \lor \overline{x_2} \lor x_4)$ is removed by moving disks in a way that a free slot of the gadget for $x_2$ is used. Since $c = \mathit{true}$ when $x_2 = \mathit{false}$, the free slots for the other two gadgets become blocked, being unable to remove their intersection using the variable gadget for $x_2$.}
    \label{fig:reduction_overview_bc}
\end{figure}
\ifFull
We remark that the gadgets in the figure are solely representations, and we shall show their detailed construction using disks later.
\fi
Let $c = (x_1 \lor \overline{x_2} \lor x_4)$ and suppose that $x_2$ is assigned to $\mathit{false}$.
This assignment implies a movement of disks that (i) removes the intersections in the clause gadget for $c$ and (ii) blocks the truth value of the variable gadget for $x_2$ (see \Cref{fig:reduction_overview_bc}(b)).
%
%
We must block the truth value of the variable gadget so that another clause gadget $c'$ does not use the free slot in the variable gadget for $x_2$ when $x_2 = \mathit{true}$.
%
Consequently, their intersections must be removed using other gadgets.
It can be shown that removing all intersections in this way is equivalent to a valid assignment of variables for which $\Phi = \mathit{true}$. 
\ifFull
\paragraph*{Reduction Overview: Moving disks} 
\fi
The disks are \emph{moved} by assigning a new location, and the distance is calculated using a function that we call \emph{moving distance function}, which is the $L_1$ or $L_2$ distance metric multiplied by a \emph{distance weight}.
We employ two types of disks classified by their distance weight, called \emph{transition disk} and \emph{heavy disk}. 
The transition disks are the disks that we aim to move, whereas heavy disks are used to restrict the movement of transition disks.
%
%
The moving distance function of a heavy disk is intuitively defined such that any significant movement that alters the construction exceeds a distance of $K$. We show that a solution that allows removing all intersections from $\D_\Phi$ with minimum maximum moving distance $K$ exclusively relies on the movement of transition disks.
We remark that, although heavy disks can move, their movement is negligible. 
Combining this condition and the above construction, it can be shown that $\Phi$ is satisfiable if and only if $\Pi_{\texttt{edgeless}}$ can be satisfied in $\cstate$ using minimum maximum moving distance $K$.

\ifConf
\begin{toappendix}
\section{Definitions and Proofs of \Cref{sec:disk_edgeless}}
\else
\begin{toappendix}
\fi
In the subsequent sections, we formally define the disks and gadgets of the reduction.
The gadgets are based on the gadgets presented in~\cite{Breu1998} and their coordinates are given in \ref{apx:coordinates}.
All coordinates are rational numbers; thus, all centres of disks can be described using a finite number of bits.
We also consider instances for which $K=1$ exclusively.
Before presenting the details in the following subsections, we make some remarks to aid in understanding the reduction.
\begin{itemize}
    \item In the following figures, we omit heavy disks that are properly inserted into the blank spaces to highlight the shape of the gadgets.
    \item When a collection of disks representing a gadget is given, it is sometimes conveniently assumed that omitted heavy disks are contained in the collection.
    \item In the following figures, the distance between consecutive transition disks is highlighted using shaded concentric circles with radius $3$ for the $L_1$ and $L_2$ distances.
\end{itemize}

\subsubsection{General Definitions}

Given an arbitrary instance $\Phi$ of {\pthreesat} with $n$ variables and $m$ clauses, we denote the collection of $\n$ unit disks produced by the reduction as $\cstate$ where $\n = f(n,m)$ is a polynomial of $n$ and $m$.
The \emph{moving distance function} of an arbitrary disk $D \in \D_\Phi$ is a function of the form $d_{D}:\mathbb{R}^2 \rightarrow \mathbb{R}$ such that $d_{D}(p) = w_D\lVert c(D),p \rVert_\alpha$ for $\alpha \in \set{1,2}$. The real $w_D >0$ is the \emph{moving weight} of $D$.
We differentiate disks according to the value of $w_D$.
The disk $D$ is called \emph{transition disk} when $w_D = K/3$.
The disk $D$ is a \emph{$k$-heavy disk} when $w_D = 2^kK$, for $k\ge 1$.
We sometimes identify the heavy disk centred at an arbitrary point $p$ as $\hdisk{p}$.
If such a heavy disk does not exist, $\hdisk{p} = \emptyset$.
Given disks $D,D' \in \cstate$, we say that \emph{$D$ is consecutive to $D'$} if $d_D(c(D')) \le K$.
Lastly, we also refer to a transition disk concentric with a heavy disk as \emph{intersection disk}.

The transition and heavy disks used in the reduction are shown in \Cref{fig:md_functions} with their corresponding moving distance functions.

\begin{figure}[!htb]
    \centering
    \includegraphics[scale=1,page=3]{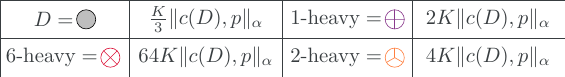}
    \caption{Disks used in the reduction: Transition disk $D$ and $k$-heavy disks, $k \in\set{1,2,6}$, with their corresponding moving distance function.}
    \label{fig:md_functions}
\end{figure}

We formally define the movement of the disks in $\cstate$. 
Let $\X : \cstate \to \set{0,1}$ be an indicator function that tells whether a disk in $\cstate$ was moved. If $D \in \cstate$ is has not been moved, then $\ismoved{D} = 0$, otherwise $\ismoved{D} = 1$. 
Let also $\X_{\mathit{pos}}: \cstate \to \mathbb{R}\times \mathbb{R}$ be a function that returns the position of a disk.
If $\ismoved{D} = 0$, then $\movedpos{D} = c(D)$ for any disk $D \in \cstate$.
Given a disk $D \in \cstate$ such that $\ismoved{D} = 0$ and a point $p \in \mathbb{R}$, we say that \emph{$D$ is moved to $p$} to refer to setting $\ismoved{D} = 1$ and $\movedpos{D} = p$.
We then define $\X_{\mathit{pos}}$ as follows:
\[
    \movedpos{D} = \begin{cases}
        c(D),&\quad \ismoved{D} = 0,\\
        p, &\quad \ismoved{D} = 1.
    \end{cases}
\]
We also define $\istate = \set{D \in \cstate\mid \ismoved{D} = 0}$ and $\mstate = \set{D \in \cstate\mid \ismoved{D} = 1}$ as the subcollections of disks that represent unmoved and moved disks, respectively.
By the definition of movement described above, we see that $\istate$ and $\mstate$ form a partition of $\cstate$. That is, $\cstate = \istate \cup \mstate$.

Given an arbitrary disk $D \in \cstate$, the \emph{range of movement of $D$}, denoted by $\A_D \subseteq \mathbb{R}^2$, is the set of points where $D$ can be moved with minimum maximum moving distance $K$. That is, 
\[
\A_D = \begin{cases}
        \set{p \in \mathbb{R}^2\mid d_D(p) \le K}, & \quad \ismoved{D} = 0,\\
        \emptyset, & \quad \ismoved{D} = 1.
    \end{cases}
\]
A \emph{blocked zone by} $D \in \cstate$, denoted by $\B_D$, is the zone where an arbitrary disk $D'\in \istate$, $D'\neq D$, cannot be moved avoiding intersecting $D$ even after moving $D$ to a point $p\in \A_D$. In particular,
\[
    \B_D  = \begin{cases}
        \emptyset,&\: \text{$D$ is a transition disk and $\ismoved{D} = 0$} ,\\
        \{p \in \mathbb{R}^2\mid \lVert c(D),p\rVert_\alpha < \frac{2^{k}-1}{2^k}\},&\: \text{$D$ is a $k$-heavy disk and $\ismoved{D} = 0$},\\
        \{p \in \mathbb{R}^2\mid \lVert \movedpos{D},p\rVert_\alpha  < 1\},&\: \ismoved{D} = 1.\\
    \end{cases}
\]
Let $r(D)$ be the radius of $D$.
Equivalently, the blocked zone of a $k$-heavy disk $D \in \istate$ is an open disk centred at $c(D)$ with radius $r(D) =(2^{k}-1)/2^k$.
The blocked zone of a disk $D \in \mstate$ is an open disk centred at $\movedpos{D}$ with radius $r(D) = 1$.

Let $D \in \cstate$ be a disk.
A point $p$ is a \emph{feasible position of movement} or simply \emph{feasible position} of $D$ if $D$ can be moved to $p$ such that $p \in \A_D$ and intersections can be removed with minimum maximum moving distance $K$ after moving $D$.
That is,
    $p \in \A_D \setminus \cup_{D' \in \left(\cstate\setminus\set{D}\right)}\B_{D'}$.
A \emph{feasible area of movement} or simply \emph{feasible area} $\F_D$ of $D$ is the union of subsets of $\mathbb{R}^2$ such that for any $p \in \F_D$, $p$ is a feasible position of $D$ (see \Cref{fig:feasible_area}).
In particular, 
\[
        \F_D =  \A_D \setminus \cup_{D' \in \left(\cstate\setminus\set{D}\right)}\B_{D'}\:.
\]

\begin{figure}[!htb]
    \centering
    \includegraphics[scale=1,page=9]{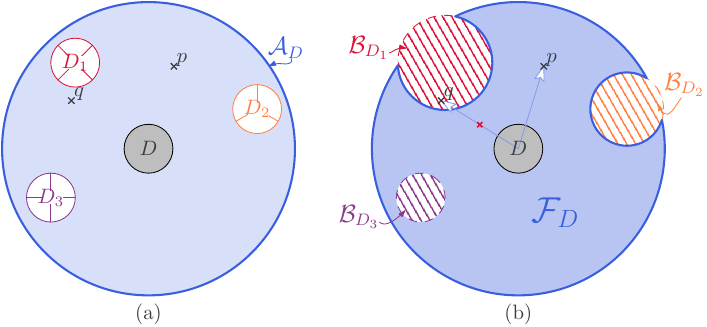}
    \caption{Range of movement and feasible area: (a) illustration of a collection of disks $\D = \set{D,D_1,D_2,D_3} \subseteq \istate$ and the range of movement $\A_D$; (b) the feasible area of $D$, $\F_D$. In particular, $\F_D = \A_D \setminus \{\B_{D_1} \cup \B_{D_2} \cup \B_{D_3}\}$ is the region marked with a bold dotted line. The disk $D$ can be moved to an arbitrary point $p$ contained in $\F_D$. On the other hand, $D$ cannot be moved to the point $q \in \A_D$ without exceeding the minimum moving distance $K$ even if $d_D(q) \le K$ holds.}
    \label{fig:feasible_area}
\end{figure}

\Cref{tab:summary_of_notation} summarises the main notation used throughout \Cref{sec:disk_edgeless}.
\begin{table}[!hbt]
    \centering
    \caption{Summary of Notation}
    \begin{tabular}{C{0.15\textwidth}p{0.79\textwidth}}
        \textbf{Symbol} & \textbf{Explanation} \\\midrule
        \multicolumn{2}{c}{Reduction instance}\\\midrule
        $\cstate$       & Collection of disks constructed by using $\Phi$. \\
        $d_D$ & Moving distance function of the disk $D$.\\
        $W_D$ & Moving weight of the disk $D$.\\
        $\ismoved{D}$ & Indicator function that returns $1$ if the disk $D \in \cstate$ was moved and $0$ otherwise.\\
        $\movedpos{D} $ & Function that returns the position of $D \in \cstate$.\\
        $\istate$ & The disks $D \in \cstate$ for which $\ismoved{D} = 0$ (unmoved disks).\\
        $\mstate$ & The disks $D \in \cstate$ for which $\ismoved{D} = 1$ (moved disks).\\
        $\mstatefinal{\D}$ & The collection of disks in which there exists a disk $D' \in \mstatefinal{D}$, $c(D') = p$ for a disk $D \in \D$ such that $\movedpos{D} = p$, given that for all $D \in \D$, $\ismoved{D} = 1$.\\
        \midrule
        \multicolumn{2}{c}{Cell Gadget, Clause Gadget and Clause Component}\\\midrule
        $\G_{(x,y)}$ & Cell gadget with its transition disk centred at $(x,y)$.\\
        $(\G)^{(\mathit{heavy})}$ & The subcollection of the heavy disks contained in $\G$.\\
        $\H_{(x,y)},\H_{c(D)}$ & Interior hole of the cell gadget $\G_{(x,y)}$ with transition disk $D$.\\
        $\G_c$ & Clause gadget of the clause $c$.\\
        $T_c$ & Intersection disk of $\G_c$.\\
        $\G^c_{i,j,k}$ & Clause component composed of clause gadget $\G_c$ and variable gadgets $\G_{x_i}$, $\G_{x_j}$ and $\G_{x_k}$.\\
        \midrule
        \multicolumn{2}{c}{Variable Gadget}\\\midrule
        $\G_x$ & Variable gadget of the variable $x$.\\
        $S_x$ & Truth setter disk of $\G_x$.\\
        $\S^x_{t,i},s^x_{t,i}$ & Truth slot of the true ($i = 1$) and false side ($i = 2$) of $\G_x$ and its centre.\\
        $\bdisk{\S^x_{t,i}}$ &  Blocking disk of the true ($i = 1$) and false side ($i = 2$).\\
        $\S^x_{c,i}, s^x_{c,i}$ & Free space to position the disk coming from the clause gadget of $c_i$ and its centre.\\
        $D_{c_i}$ & Transition disk coming from clause gadget $c_i$.\\
        $\bdisk{\S^x_{c,i}}$ & Blocking disk for $c_i$.\\
        $\ldisk{\S^x_{c,i}}$ & Link disk for $c_i$.\\
        \midrule
        \multicolumn{2}{c}{General}\\\midrule
        $r(D)$ & Radius of disk $D \subseteq \mathbb{R}^2$.\\
        $\hdisk{p}$ & Heavy disk centred at $p \in \mathbb{R}^2$.\\
        $\A_D$ & Range of movement of disk $D \in \cstate$.\\
        $\B_D$ & Blocked zone by disk $D \in \cstate$.\\
        $\F_D$ & Feasible area of disk $D \in \cstate$.\\
        \bottomrule
    \end{tabular}
    \label{tab:summary_of_notation}
\end{table}


Given a subcollection $\D \subseteq \cstate$ such that for all $D \in \D, \ismoved{D} = 1$ (that is, all disks in $\D$ were moved), we denote by $\mstatefinal{\D}$ the collection of disks in which there exists a disk $D' \in \mstatefinal{\D}$ such that $c(D') = \movedpos{D}$ for a disk $D \in \D$.
If for all $D \in \cstate$, $\ismoved{D} = 1$, $d_D(\movedpos{D}) \le K$ and $\mstatefinal{\cstate}$ is contained in $\Pi$, then $(\cstate,\Pi)$ is a yes instance of {\ggedmm}.
In the reduction, we construct $\cstate$ using $\Phi$ and $G_\Phi$ and show that there exists a $\mstatefinal{\cstate}$ such that $\mstatefinal{\cstate}$ is in $\Pi_{\texttt{edgeless}}$ if and only if $\Phi$ is satisfiable.

\subsubsection{Cell Gadgets}

A \emph{cell gadget} $\G_{(x,y)}\subseteq \cstate$ consists of a transition disk centred at $(x,y)$ surrounded by $6$-heavy disks centred at points $\{(x+i,y+j)\mid i,j\in \{-1,0,1\}\}\setminus\{(x,y)\}$ (see \Cref{fig:cell_gadget}). We denote the subcollection that contains these heavy disks of $\G_{(x,y)}$ by $\G_{(x,y)}^{(\mathit{heavy})}$.

\begin{figure}[!b]
    \centering
    \includegraphics[scale=1,page=26]{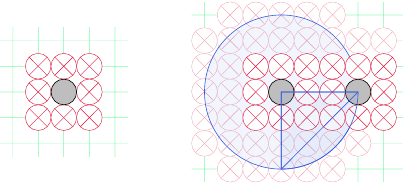}
    \caption{Cell gadget: An arbitrary cell gadget (left); two connected cell gadgets such that their transition disks are consecutive. The faded heavy disks are the omitted heavy disks to show the shape of gadgets.}
    \label{fig:cell_gadget}
\end{figure}

Given a collection of disks $\D$, we denote the convex hull of the set $\set{c(D)\mid D\in \D}$ by $\C(\D)$.
Let $\D$ be a collection of $k$-heavy disks. If $\C(\D)$ (i) is a $|\D|$-gon and (ii) $\B_{D}\cap \B_{D'} \neq \emptyset$ for any pair of disks $D,D' \in \D$ such that $c(D)$ and $c(D')$ share an edge in $\C(\D)$, then the set $\cup_{D\in\D} \B_D $ is called \emph{blocked enclosure}.
For instance, the union of blocked zones of heavy disks in a cell gadget forms a blocked enclosure.
If the blocked zones of $\D$ form a blocked enclosure, the \emph{interior hole} of $\D$ is the region given by $\C(\D) \setminus \cup_{D\in\D} \B_D$.

Given a cell gadget $\G_{(x,y)}$ with transition disk $D$, we denote its interior hole $\C(\G_{(x,y)}^{(\mathit{heavy})})\setminus (\cup_{D' \in \G_{(x,y)}^{(\mathit{heavy})}} \B_{D'})$ by $\H_{(x,y)}$ or $\H_{c(D)}$ indistinctly. \Cref{fig:interior_holes} shows an example of interior holes for heavy disks of an arbitrary cell gadget.

\begin{mlemmarep}\label{lem:no_holes_hd_squares}
    If $\D = \set{H_1,\ldots,H_4} \subseteq \istate$ is a collection of $6$-heavy disks such that $c(H_1) = (x,y)$, $c(H_2) = (x-1,y)$, $c(H_3) = (x,y-1)$ and $c(H_4) = (x-1,y-1)$, then the interior hole of $\D$ is empty.
\end{mlemmarep}
\begin{proof}
    Notice that the minimax centre of $c(H_1),\ldots,c(H_4)$ is given by the centre of the smallest circle that contains them. 
    This point is $c = (x-1/2,y-1/2)$.
    Moreover, the point $c$ is the farthest point from any $c(H_i)$ for $i\in \set{1,\ldots,4}$.
    We have $\lVert c, c(H_i)\rVert_2 = \sqrt{2}/2 < (2^6 - 1 )/2^6 = 63/64$ for any $i \in \set{1,\ldots,4}$.
    Therefore, any other point enclosed by $c(H_1),\ldots,c(H_4)$ is contained in at least one blocked zone, implying that the interior hole is empty.
\end{proof}

\begin{mlemmarep}\label{lem:holes_in_cell_gadgets}
    The interior hole $\H_{(x,y)}$ of an arbitrary cell gadget $\G_{(x,y)}$ with transition disk $D$ is a non-empty set.
    Moreover, $c(D) \in \H_{(x,y)}$.
\end{mlemmarep}
\begin{proof}
    Let $H_1,\ldots,H_8 \in \G_{(x,y)}^{(\mathit{heavy})}$ be the eight heavy disks of $\G_{(x,y)}$ such that the centres $(c(H_1),\ldots,c(H_8))$ are equal to $ ((x-1,y+1),(x,y+1),(x+1,y+1),(x-1,y),(x+1,y),(x-1,y-1),(x,y-1),(x+1,y-1)$.
    By definition, $\lVert (x,y), c(H_i)\rVert_2 \ge 1 > (2^6-1)/ 2^6 = 63/64$ holds for $i \in \set{1,\ldots,8}$.
    Therefore $c(D) \in \H_{(x,y)}$ and $\H_{(x,y)} \neq \emptyset$.
\end{proof}

\begin{observation}\label{obs:convex_polygon_one_disk}
    Let $\S$ be a convex polygon such that $\diam{\S} < 1$. The region delimited by $\S$ admits exactly one disk centred within it.
\end{observation}
\begin{proof}
    Let $p_1,p_2$ be the farthest pair of points in $\S$.
    Without loss of generality, assume that an arbitrary disk $D$ is moved to $p_1$.
    It gives $r(\B_D) = 1$ and thus $p_2 \in \B_D$ since $\diam{\S} < 1$.
    Since $p_2$ is the farthest point from $p_1$, all other points in $\S$ are also blocked by $\B_D$.
    Therefore, $\S$ admits exactly one disk centred within it.
\end{proof}

\begin{mlemmarep}\label{lem:holes_one_disk}
    The interior hole $\H_{(x,y)}$ of an arbitrary cell gadget $\G_{(x,y)}$ admits exactly one disk centred within it.
\end{mlemmarep}
\begin{proof}
    Let $\S$ be the square formed by points $(p_1,\ldots,p_4) = ((x-1/64,y+1/64),(x+1/64,y+1/64),(x-1/64,y-1/64),(x+1/64,y-1/64))$.
    The square $\S$ is a rectangle such that $\lVert p_1,p_4\rVert_2,\lVert p_2,p_3\rVert_2<1$. Consequently, $\S$ admits exactly one disk centred within it by \Cref{obs:convex_polygon_one_disk}.
    Moreover, $\H_{(x,y)} \subseteq \S \setminus \cup_{D' \in \G_{(x,y)}^{(\mathit{heavy})}} \B_{D'}$.
    Therefore, the lemma statement is true.
\end{proof}

\begin{figure}[!b]
    \centering
    \includegraphics[scale=1,page=10]{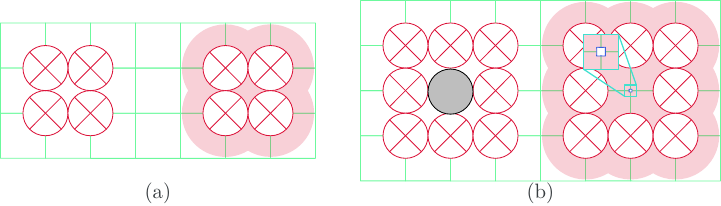}
    \caption{Interior holes: (a) A square of four disks with no interior hole (\Cref{lem:no_holes_hd_squares}); (b) union of blocked zones of heavy disks of a cell gadget (\Cref{lem:holes_in_cell_gadgets}), where the small region at the centre of the figure not covered by blocked zones is the interior hole.}
    \label{fig:interior_holes}
\end{figure}

We now show that the feasible area of a transition disk is restricted to the areas used by consecutive transition disks. 

\begin{mlemmarep}\label{lem:cg_two_subsets}
    Let $\G_{(x,y)}$ and $\G_{(x',y')}$ be two cell gadgets with transition disks $D$ and $D'$, respectively.
    If $D$ and $D'$ are consecutive, then the feasible area $\F_{D}$ is equal to $(A_D \cap \H_{(x,y)}) \cup (A_{D} \cap \H_{(x',y')})$.
    Moreover, $A_D \cap \H_{(x,y)}$ and $A_{D'} \cap \H_{(x',y')}$ are non-empty and disjoint.
\end{mlemmarep}
\begin{proof}
    First, $\H_{(x,y)}$ and $\H_{(x',y')}$ are disjoint, so their intersection with $\A_D$ is also disjoint.
    By \Cref{lem:holes_in_cell_gadgets}, $\H_{(x,y)} \neq \H_{(x',y')} \neq \emptyset$ holds.
    By definition of $\G_{(x,y)}$, $c(D) \in \H_{(x,y)}$ and $\B_D = \emptyset$ hold, so $\A_D \cap \H_{(x,y)} \neq \emptyset$ also holds.
    Since $D$ is consecutive to $D'$, $d_D(c(D')) \le K$ and $c(D') $ is contained in $ \A_D$.
    It follows that $\A_D \cap \H_{(x',y')} \neq \emptyset$.

    We prove that $\F_{D}$ is equal to $(A_D \cap \H_{(x,y)}) \cup (A_{D} \cap \H_{(x',y')})$.
    Recall that $\H_{(x,y)} = \C(\G_{(x,y)}^{(\mathit{heavy})}) \setminus \cup_{O \in \G_{(x,y)}^{(\mathit{heavy})}} \B_{O}$ and $\H_{(x',y')} = \C(\G_{(x',y')}^{(\mathit{heavy})}) \setminus \cup_{O \in \G_{(x',y')}^{(\mathit{heavy})}} \B_{O}$. By \Cref{lem:holes_in_cell_gadgets} and the definition of $\cstate$, $\H_{(x,y)}$ and $\H_{(x',y')}$ do not intersect with any blocked zone. Hence we can define
    \begin{align*}
        \H_{(x,y)} & = \C(\G_{(x,y)}^{(\mathit{heavy})}) \setminus \cup_{O \in \cstate} \B_{O} = \C(\G_{(x,y)}^{(\mathit{heavy})}) \cap (\cup_{O \in \cstate} \B_{O})^c\\
        \H_{(x',y')} & = \C(\G_{(x',y')}^{(\mathit{heavy})}) \setminus \cup_{O \in \cstate} \B_{O} = \C(\G_{(x',y')}^{(\mathit{heavy})}) \cap (\cup_{O \in \cstate} \B_{O})^c \:.
    \end{align*}
    
    On the other hand, $\F_D = \A_D \setminus \cup_{D' \in \cstate} \B_{D'} = \A_D \cap (\cup_{O \in \cstate} \B_{O})^c$. We now have that
    \begin{align*}
        (A_D \cap \H_{(x,y)}) \cup (A_D \cap \H_{(x',y')}) & = \left(\A_D \cap \left(\C(\G_{(x,y)}^{(\mathit{heavy})}) \cap (\cup_{O \in \cstate} \B_{O})^c\right)\right)\\
        &\phantom{=}\cup \left(\A_D \cap \left(\C(\G_{(x',y')}^{(\mathit{heavy})}) \cap (\cup_{O \in \cstate} \B_{O})^c\right)\right) \\
        & = \left(\A_D \cap \C(\G_{(x,y)}^{(\mathit{heavy})}) \cap (\cup_{O \in \cstate} \B_{O})^c\right)\\
        &\phantom{=} \cup \left(\A_D \cap \C(\G_{(x',y')}^{(\mathit{heavy})}) \cap (\cup_{O \in \cstate} \B_{O})^c\right)\\
        & =(\F_D \cap \C(\G_{(x,y)}^{(\mathit{heavy})}))\cup (\F_D \cap \C(\G_{(x',y')}^{(\mathit{heavy})})).
    \end{align*}
    It holds that $\C(\G_{(x,y)}^{(\mathit{heavy})})) \subset \A_D$, thus
    \begin{align*}
        (A_D \cap \H_{(x,y)}) \cup (A_D \cap \H_{(x',y')}) & =\F_D \cup (\F_D \cap \C(\G_{(x',y')}^{(\mathit{heavy})}))\\
        & = \F_D.
    \end{align*}
    This concludes the proof.
\end{proof}

\subsubsection{Clause Gadgets}

The \emph{clause gadget} $\G_c \subseteq \cstate$ for an arbitrary clause $c$ consists of an intersection disk $T_c$ centred at an arbitrary point $(x,y)$ surrounded by $6$-heavy disks centred at points $\{(x+i,y+j)\mid i,j\in \{-1,0,1\}\}\setminus\{(x,y)\}$.
It also contains three cell gadgets $\G_{(x-3,y)}$, $\G_{(x,y-3)}$ and $\G_{(x+3,y)}$ representing the three literals of $c$ (see \Cref{fig:clause_gadget}).
We interpret the movement of $T_c$ to one of the arms as the assignment of truth value to the clause by the literal corresponding to the arm. 

\begin{figure}[!b]
    \centering
    \includegraphics[scale=1,page=27]{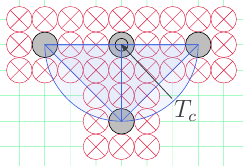}
    \caption{Clause Gadget: Clause gadget for an arbitrary clause $c$.}
    \label{fig:clause_gadget}
\end{figure}


\Cref{lem:cg_two_subsets} implies that if a transition disk is moved from its position (that is, it is moved outside $\H_{(x,y)}$), then its new position must be in $\H_{(x',y')}$ for an arbitrary $(x',y')\neq (x,y)$.
We now prove a similar property of clause gadgets.

\begin{mlemmarep}\label{lem:tc_three_subsets}
    Let $\G_{c}$ be a clause gadget such that $c(T_c) = (x,y)$.
    The feasible area $\F_{T_c}$ is equal to $(\A_{T_c} \cap \H_{(x-3,y)}) \cup (\A_{T_c} \cap \H_{(x,y-3)}) \cup (\A_{T_c} \cap \H_{(x+3,y)})$.
    Moreover, $\A_{T_c} \cap \H_{(x-3,y)}$, $\A_{T_c} \cap \H_{(x,y-3)}$ and $\A_{T_c} \cap \H_{(x+3,y)}$ are non-empty and disjoint between each other.
\end{mlemmarep}
\begin{proof}
    First, $\H_{(x-3,y)}$, $\H_{(x,y-3)}$ and $\H_{(x+3,y)}$ are disjoint between each other, so their intersections with $\A_{T_c}$ must also be disjoint.
    By \Cref{lem:holes_in_cell_gadgets}, $\H_{(x-3,y)} \neq \H_{(x,y-3)} \neq \H_{(x+3,y)} \neq \emptyset$.
    Moreover, $d_{T_c}((x-3,y)),d_{T_c}((x,y-3)),d_{T_c}((x+3,y)) \le K$ holds, so $\A_{T_c} \cap \H_{(x-3,y)}$, $\A_{T_c} \cap \H_{(x,y-3)}$ and $\A_{T_c} \cap \H_{(x+3,y)}$ are non-empty.

    We prove that $\F_{T_c}$ is equal to $(\A_{T_c} \cap \H_{(x-3,y)}) \cup (\A_{T_c} \cap \H_{(x,y-3)}) \cup (\A_{T_c} \cap \H_{(x+3,y)})$.
    Recall that we can define $\H_{(x-3,y)}$, $\H_{(x,y-3)}$ and $\H_{(x+3,y)}$ as follows:
    \begin{align*}
        \H_{(x-3,y)} & = \C(\G_{(x-3,y)}^{(\mathit{heavy})}) \setminus \cup_{D' \in \cstate} \B_{D'} = \C(\G_{(x-3,y)}^{(\mathit{heavy})}) \cap (\cup_{D' \in \cstate} \B_{D'})^c\\
        \H_{(x,y-3)} & = \C(\G_{(x,y-3)}^{(\mathit{heavy})}) \setminus \cup_{D' \in \cstate} \B_{D'} = \C(\G_{(x,y-3)}^{(\mathit{heavy})}) \cap (\cup_{D' \in \cstate} \B_{D'})^c\\
        \H_{(x+3,y)} & = \C(\G_{(x+3,y)}^{(\mathit{heavy})}) \setminus \cup_{D' \in \cstate} \B_{D'} = \C(\G_{(x+3,y)}^{(\mathit{heavy})}) \cap (\cup_{D' \in \cstate} \B_{D'})^c.
    \end{align*}
    On the other hand, $\F_{T_c} = \A_{T_c} \setminus \cup_{D' \in \cstate} \B_{D'} = \A_{T_c} \cap (\cup_{D' \in \cstate} \B_{D'})^c$. Let $\H = (\H_{(x-3,y)} \cup \H_{(x,y-3)} \cup \H_{(x+3,y)})$. Then, 
    \begin{align*}
        (\A_{T_c} \cap \H_{(x-3,y)}) &\cup (\A_{T_c} \cap \H_{(x,y-3)}) \cup (\A_{T_c} \cap \H_{(x+3,y)}) = \\
        &= \A_{T_c} \cap (\H_{(x-3,y)} \cup \H_{(x,y-3)} \cup \H_{(x+3,y)})\\
        & = \A_{T_c} \cap \H\\
        & = \A_{T_c} \cap \H \cap (\cup_{D' \in \cstate} \B_{D'})^c\\
        & = \F_{T_c} \cap \H.
    \end{align*}
    For all $x \in \A_{T_c} \cap \H$, it holds that $x \in \F_{T_c}$ by the above equation. Hence, $\A_{T_c} \cap \H \subseteq F_{T_c}$.
    In contrast, for all $x \in \F_{T_c}$, we have $x \in \A_{T_c}$ since $\F_{T_c} \subseteq \A_{T_c}$.
    Moreover, it follows from \Cref{lem:no_holes_hd_squares} and the definition of clause gadgets that there exists no point $p\in \A_D$ such that $p \notin \cup_{D' \in \cstate} \B_{D'}$ and $p \notin \H$.
    Thus $\F_{T_c} \subseteq \H_{(x-3,y)}\cup \H_{(x,y-3)}\cup \H_{(x+3,y)}$ holds, which implies that $x \in \H$.
    Consequently, $\F_{T_c} \subseteq A_{T_c} \cap \H$ holds.
    Therefore $\F_{T_c} = \A_{T_c} \cap \H = (\A_{T_c} \cap \H_{(x-3,y)}) \cup (\A_{T_c} \cap \H_{(x,y-3)}) \cup (\A_{T_c} \cap \H_{(x+3,y)})$ holds.
\end{proof}

\subsubsection{Variable Gadgets}

Lastly, the \emph{variable gadget} $\G_x \subseteq \cstate$ for an arbitrary variable $x$ is depicted in \Cref{fig:variable_gadget}. 
The variable gadget consists of two truth sides (called true and false sides) with three free slots and an intersection disk. 
There exists space for moving the truth setter disk in both truth sides.

\begin{figure}[!htb]
    \centering
    \includegraphics[scale=1,page=28]{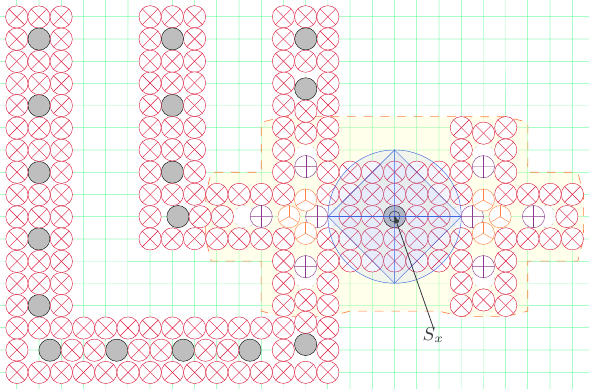}
    \caption{Variable gadget: Variable gadget for variable $S_x$ with three arms connected to the left side. The dashed region is the central part of the variable gadget. Shaded concentric circles are circles of radius $K$ for the $L_1$ and $L_2$ distances highlighting the distance between consecutive transition disks.}
    \label{fig:variable_gadget}
\end{figure}

The collection of disks enclosed by the yellow dashed region is the \emph{central part} of the gadget. 
\Cref{fig:variable_names} shows the central part of the gadget for an arbitrary variable $x$.
The disk ${S}_x$ is called \emph{truth setter disk}. 

\begin{figure}[!htb]
    \centering
    \includegraphics[scale=1,page=2]{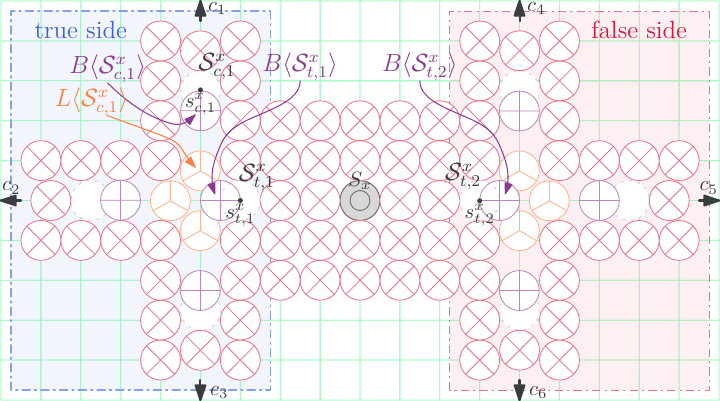}
    \caption{Central part of the gadget for an arbitrary variable $x$.}
    \label{fig:variable_names}
\end{figure}

For convenience, we show the definitions assuming that $c(S_x) = (0,0)$.
The central part consists of consecutive $6$-heavy disks surrounding $S_x$ and two \emph{truth sides} called \emph{true side} and \emph{false side}.
As we show in \Cref{lem:block_disk_restricted}, $S_x$ can be moved to free spaces $\S_{t,1}^x$ or $\S_{t,2}^x$ called \emph{truth slots}, representing the false or true value given $x$, respectively.
The centres of $\S_{t,1}^x$ and $\S_{t,2}^x$ are denoted by $s^x_{t,1} = (-3,0)$ and $s^x_{t,2} = (3,0)$, respectively.
The \emph{blocking disk of a truth side} is a $1$-heavy disk that possibly blocks $\S_{t,i}^x$ and is denoted by $\bdisk{\S^x_{t,i}}$, where $i = 1$ for the true side and $i=2$ for the false side.
Each truth side can be connected to at most three clauses, denoted in the figure by $c_1,c_2,c_3$ for the true side and $c_4,c_5,c_6$ for the false side.
For $i \in \set{1,\ldots,6}$, clause $c_i$ contains a free space $\S_{c,i}^x$ to position the transition disk moved from clause gadget $c_i$, with centre $s_{c,i}^x$.
We denote the transition disk moved from $c_i$ by $D_{c_i}$.
\Cref{fig:move_of_disks} shows the transition disks $D_{c_1},D_{c_2},D_{c_3}$ for $c_1,c_2,c_3$, respectively.

\begin{figure}[!hbt]
    \centering
    \includegraphics[scale=1,page=5]{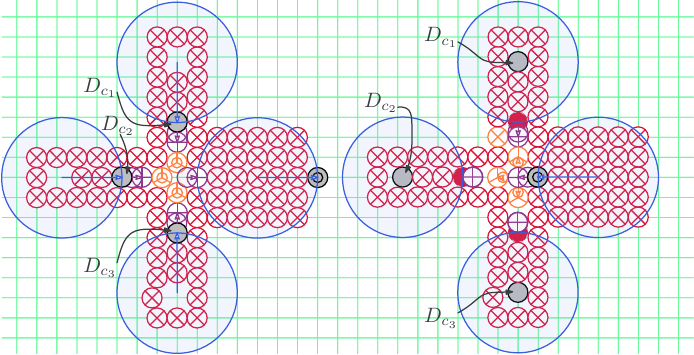}
    \caption{Left: The truth setter disk is blocking the right side, letting transition disks $D_{c_1},D_{c_2},D_{c_3}$ to be moved to free slots. Right: The truth setter disk is blocking free slots to move $D_{c_1},D_{c_2},D_{c_3}$ into the gadget.}
    \label{fig:move_of_disks}
\end{figure}

The \emph{blocking disk for $c_i$} is a $1$-heavy disk that possibly blocks $\S_{c,i}^x$ and is denoted by $\bdisk{\S^x_{c,i}}$.
The \emph{link disk for $c_i$} is a $2$-heavy disk moved close to $\bdisk{\S^x_{c,i}}$ depending on the position of the blocking disk of the truth side.
The link disk is denoted by $\ldisk{\S^x_{c,i}}$.

Each variable is connected to at most three clauses by the arms depicted in \Cref{fig:variable_gadget}connected to $c_1$, $c_2$ and $c_3$.
The central part of the gadget can be mirrored horizontally and arms can be mirrored vertically and horizontally.
We make some observations used in the subsequent lemmas.

\begin{observation}[Vertical condition of blocked zones]\label{obs:vertical_condition}
    Let $H_1$ and $H_2$ be two $k$-heavy disks centred at $(x,y)$ and $(x+1,y)$, respectively. A point $(x',y')$ such that the inequality $y - \sqrt{((2^k-1)/2^k)^2 - 1/4} \le y' \le y + \sqrt{((2^k-1)/2^k)^2 - 1/4}$ holds is in $\B_{H_1} \cup \B_{H_2}$ for any $x-1/2\le x' \le x+3/2 $.
\end{observation}
\begin{proof}
    We show that $\B_{H_1} $ and $ \B_{H_2}$ intersect at $(x+1/2,y + \sqrt{((2^k-1)/2^k)^2 - 1/4})$. 
    The proof is analogous for the lower bound.
    Let $p = (x'',y'')$ be the point where $\B_{H_1} $ and $ \B_{H_2}$ intersect.
    Straightforwardly, $x'' = x+1/2$.
    The triangle formed by points $p$, $(x,y)$ and $(x,y+1)$ is an isosceles triangle of base length $1$ and sides of length $2^k-1 /2^k$.
    Moreover, $y'' = y + h$ where $h$ is the height of the triangle.
    The height is equal to $h = \sqrt{((2^k-1)/2^k)^2 - 1/4}$.
    Hence $p = (x+1/2,y + \sqrt{((2^k-1)/2^k)^2 - 1/4})$.

    Furthermore, notice that $(x-1/2,y'') \in \B_{H_1}$ and $(x+3/2,y'') \in \B_{H_2}$.
    Therefore, for any $x-1/2\le x' \le x+3/2 $, $(x',y + \sqrt{((2^k-1)/2^k)^2 - 1/4}) \in \B_{H_1} \cup \B_{H_2}$.
\end{proof}

\begin{observation}[Horizontal condition of blocked zones]\label{obs:horizontal_condition}
    Let $H_1$ and $H_2$ be two $k$-heavy disks centred at $(x,y)$ and $(x,y+1)$, respectively. A point $(x',y')$ such that $x - \sqrt{((2^k-1)/2^k)^2 - 1/4} \le x' \le x + \sqrt{((2^k-1)/2^k)^2 - 1/4}$ is in $\B_{H_1} \cup \B_{H_2}$ for any $y-1/2\le y' \le y+3/2 $.
\end{observation}
    
\begin{figure}[tb]
    \centering
    \includegraphics[scale=1,page=14]{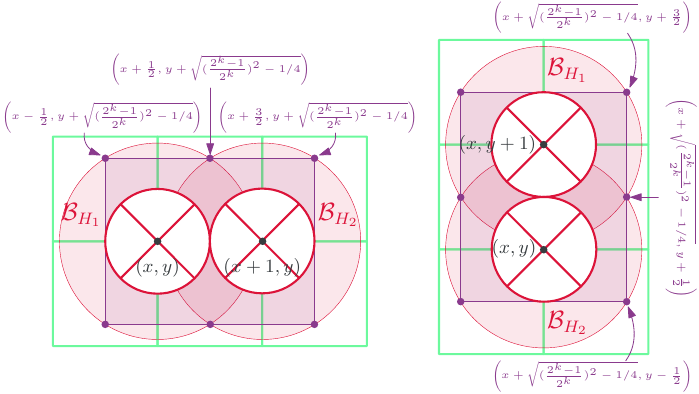}
    \caption{Illustration of \Cref{obs:vertical_condition} (left) and \Cref{obs:horizontal_condition} (right).}
\end{figure}

\Cref{lem:tc_three_subsets} defines the removal of the intersection of $T_c$. In particular, $T_c$ must be moved to one of the three non-empty disjoint subsets given by $(\A_{T_c} \cap \H_{(x-3,y)})$, $(\A_{T_c} \cap \H_{(x,y-3)})$ and $(\A_{T_c} \cap \H_{(x+3,y)})$.

\begin{mlemmarep}\label{lem:sx_two_subsets}
    Let $\G_{x}$ be a variable gadget such that $c(S_x) = (x,y)$ and $H_1^{(i)},\ldots,H_5^{(i)}$ the $6$-heavy disks surrounding $s^x_{t,i}$ for $i \in \set{1,2}$.
    The feasible area $\F_{S_x}$ is equal to $\F_{S_x}^{(1)} \cup \F_{S_x}^{(2)}$, where $\F_{S_x}^{(i)} = (\A_{S_x}\cap \S^x_{t,i})\setminus \cup_{j=1}^5 \B_{H_j^{(i)}}$ for $i \in \set{1,2}$.
    Moreover, $\F_{S_x}^{(t)}$ and $\F_{S_x}^{(f)}$ are non-empty and disjoint.
\end{mlemmarep}
\begin{proof}
    It holds that $\A_{S_x}\cap \S^x_{t,i} \neq \emptyset$ since $d_{S_x}(s^x_{t,1}) \le K$.
    The region $\S^x_{t,i}$ is partially blocked by $H_1^{(1)},\ldots,H_5^{(1)}$.
    In particular, $s^x_{t,1} = (x-3,y)$ and the centres of $H_1^{(1)},\ldots,H_5^{(1)}$ are given by $c(H_1^{(1)}) = (x-3,y+1)$, $c(H_2^{(1)}) = (x-2,y+1)$, $c(H_3^{(1)}) = (x-2,y)$, $c(H_4^{(1)}) = (x-2,y-1)$ and $c(H_5^{(1)}) = (x-3,y-1)$.
    For any $i \in \set{1,\ldots,5}$, $d_{H_i^{(1)}}(s_{t,1}^x) \ge 1$ holds, implying that $s_{t,1}^x \notin \cup_{i=1}^5 \B_{H_i^{(1)}}$.
    It follows that $\F_{S_x}^{(i)} = (\A_{S_x}\cap \S^x_{t,i})\setminus \cup_{j=1}^5 \B_{H_j^{(i)}}$ is non-empty for $i = 1$.
    Moreover, by the definition of $\G_x$, there exists no other heavy disk close to $\S^x_{t,i}$.
    Consequently, $\F_{S_x}^{(i)} = (\A_{S_x}\cap \S^x_{t,i})\setminus \cup_{j=1}^5 \B_{H_j^{(i)}} = \cup_{D'\in\cstate-\{S_x\}}B_{D'}$.
    The proof is analogous for $i = 2$.
    Moreover, $\S^x_{t,1}$ and $\S^x_{t,2}$ are disjoint by the definition of $\G_x$, so $\F_{S_x}^{(1)}$ and $\F_{S_x}^{(2)}$ must also be disjoint.

    We have shown that $\F_{S_x}^{(i)} \neq \emptyset$ for $i \in \set{1,2}$. 
    We only need to prove that $\F_{S_x} \setminus (\F_{S_x}^{(1)} \cup \F_{S_x}^{(2)}) = \emptyset$ holds.
    By the definition of the variable gadget, $S_x$ is surrounded by $6$-heavy disks $\cup_{i,j\in \set{-3,\ldots,3}} \hdisk{(i,j)}$ where $\hdisk{(-3,0)}= \hdisk{s_{t,1}^x} = \emptyset$ and $\hdisk{(3,0)}= \hdisk{s_{t,2}^x} = \emptyset$.
    Suppose instead that there exist such disks at $s_{t,1}^x$ and $s_{t,2}^x$.
    Let $z = \sqrt{(63/64)^2-1/4}$.
    By \Cref{obs:vertical_condition,obs:horizontal_condition}, any point $(x',y')$ such that $x-3-z \le x' \le x+3+z$ and $y-3 - z \le y' \le y+3 + z$ is in $\cup_{i,j\in \set{-3,\ldots,3}} \B_{\hdisk{(i,j)}}$.
    This would imply $\F_{S_x} = \emptyset$ as $r(\A_{S_x}) = 3$.
    We have already shown that if there exists no heavy disk centred at $s_{t,1}^x$ and $s_{t,2}^x$, then there exist two non-empty disjoint subsets in $\F_{S_x}$, namely $\F_{S_x}^{(1)}$ and $\F_{S_x}^{(2)}$.
    Therefore, $\F_{S_x} \setminus (\F_{S_x}^{(1)} \cup \F_{S_x}^{(2)}) = \emptyset$ holds.
\end{proof}
\Cref{lem:sx_two_subsets} implies that the movement of $S_x$ is restricted to the two subsets $\F_{S_x}^{(1)}$ and $\F_{S_x}^{(2)}$. 
We now show that depending on where $S_x$ is moved, the feasible areas of $1$- and $2$-heavy disks of the variable gadget become restricted.

\begin{mlemmarep}\label{lem:block_disk_restricted}
    Let $\G_{x}$ be a variable gadget such that $c(S_x) = (x,y)$ and $H_1^{(i)},\ldots,H_5^{(i)}$ the $6$-heavy disks surrounding $s^x_{t,i}$ for $i \in \set{1,2}$.
    If $S_x$ is moved to $\F_{S_x}^{(1)}$, then $\bdisk{S^x_{t,1}}$ must be moved to a point $p$ that makes $\bdisk{S^x_{t,1}}$ intersect with disks $\ldisk{S^x_{c,1}}$, $\ldisk{S^x_{c,2}}$ and $\ldisk{S^x_{c,3}}$.
\end{mlemmarep}
\begin{proof}
    We need to consider the feasible areas of disks to show the statement.
    For simplicity, we assume that feasible areas of the involved disks are properly defined rectangles.
    These rectangles completely contain the actual feasible areas of disks.
    We then show that the statement is true for the rectangles (and hence also true for the actual feasible areas). 
    In particular, we describe the movement of disks and show that whenever we move the disks to any point of the rectangles, the statement holds true.
    An illustration of the proof is shown in \Cref{fig:block_disk_restricted}.
    
    By \Cref{lem:sx_two_subsets}, $\F_{S_x}^{(1)} = (\A_{S_x}\cap \S^x_{t,1})\setminus \cup_{j=1}^5 \B_{H_j^{(1)}}$. 
    Without loss of generality, assume that $(x,y) = 0$.
    Recall that disks $H_1^{(1)},\ldots,H_5^{(1)}$ are the five $6$-heavy disks surrounding $\S^x_{t,1}$ with centres $c(H_1^{(1)}) = (-3,1)$, $c(H_2^{(1)}) = (-2,1)$, $c(H_3^{(1)}) = (-2,0)$, $c(H_4^{(1)}) = (-2,-1)$ and $c(H_5^{(1)}) = (-3,-1)$.
    We prove the statement for a rectangle $\R$ such that $\F_{S_x}^{(1)} \subseteq \R$. 
    In particular, let $\R$ be the rectangle defined by points $(-3,1/50),(-3+1/50,1/50),(-3+1/50,-1/50), (-3,1/50)$.
    It can be checked that $\F_{S_x}^{(1)}$ is contained in $\R$ by calculating the distances to $c(H_1^{(1)})$, $c(H_3^{(1)})$ and $c(H_5^{(1)})$.
    Let $q = (q_x,q_y) \in \R$ be the point where $S_x$ is moved.
    We show the feasible area of $\bdisk{S_{t,1}^x}$ after moving $S_x$ to $q$.
    Let $p = (p_x, p_y)$ be an arbitrary point of $\F_{\bdisk{S_{t,1}^x}}$.
    Recall that if $S_x$ is moved, then $\ismoved{S_x} = 1$ and $r(\B_{S_x}) = 1$.
    Thus $p_x \le q_x - 1 = -4 + 1/50$ holds since $\bdisk{S_{t,1}^x}$ is initially placed at $(-3-1/2,0)$ and $q_x \le -3 +1/50$ by the definition of $\R$.
    Moreover, we have that $c(\ldisk{S^x_{c,1}}) = (-4,3/4)$, $c(\ldisk{S^x_{c,2}}) = (-4-3/4,0)$ and $c(\ldisk{S^x_{c,3}}) = (0,-3/4)$.
    Consequently $-4 \le p_x$ also holds since $r(\ldisk{S_{c,2}^x}) = 3/4$.
    
    By the above argument, the inequality $-4 \le p_x \le -4+1/50$ holds. 
    We use this range of values for $p_x$ and define a rectangle $\R'$ that contains $\F_{\bdisk{S_{t,1}^x}}$ and thus $p$.
    Notice that $\F_{\bdisk{S_{t,1}^x}} = \A_{\bdisk{S_{t,1}^x}} \setminus (\cup_{i=1}^3 \B_{\ldisk{S^x_{c,i}}} \cup \B_{S_x})$. 
    We describe $\R'$ by using the limits of $\F_{\bdisk{S_{t,1}^x}}$. 
    First, let $p' = (-4+1/50,y')$ be a point. 
    We set $y'$ such that $\lVert c(\ldisk{S^x_{c,1}}),(-4+1/50,y')\rVert_\alpha = 3/4$.
    The point $p'$ is on the boundary of $\B_{\ldisk{S^x_{c,1}}}$.
    It can be checked that $y' < 1/100$, thus we conveniently set $y' = 1/100$ as an upper bound for $\R'$.
    The same argument can be used for $\ldisk{S^x_{c,3}}$ and a $y' = -1/100$ as a lower bound for $\R'$.
    On the other side, note that $\cap_{i \in \set{1,2,3}} \B_{\ldisk{S^x_{c,i}}} = \set{(-4,0)}$. 
    In other words, $\B_{\bdisk{S^x_{t,1}}}$ cannot be moved to a point $(-4,y')$ for $y' \neq 0$.
    Hence, we can use a rectangle $\R'$ defined by points $(-4,1/100)$, $(-4+1/50,1/100)$, $(-4+1/50,-1/100)$ and $(-4,-1/100)$ for which $F_{\bdisk{S^x_{t,1}}}$ is a subset.
    It can be checked that for any point in $p \in \R'$, the distance $\lVert p,c(\ldisk{S^x_{c,i}})\rVert_\alpha < 1$ for $i \in \set{1,2,3}$.
    In other words, whenever $\bdisk{S_{t,1}^x}$ is moved to a point $p \in \R'$, it intersects with link disks $\ldisk{S^x_{c,1}}$, $\ldisk{S^x_{c,2}}$ and $\ldisk{S^x_{c,3}}$.
    Therefore, $\bdisk{S^x_{t,1}}$ must be moved to a point $p$ that makes $\bdisk{S^x_{t,1}}$ intersect with disks $\ldisk{S^x_{c,1}}$, $\ldisk{S^x_{c,2}}$ and $\ldisk{S^x_{c,3}}$.
\end{proof}

\begin{figure}[bt]
    \centering
    \includegraphics[scale=1,page=11]{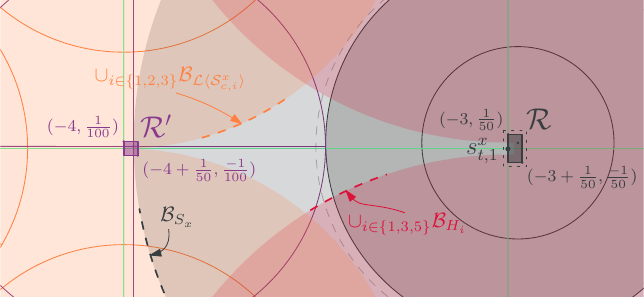}
    \caption{Illustration of \Cref{lem:block_disk_restricted}. The orange regions is the union of blocked zones of disks $\ldisk{S^x_{c,i}}$ for $i\in \set{1,2,3}$, whereas the grey region is the blocked zone of $S_x$. The feasible area $\F_{\bdisk{S^x_{t,1}}}$ is confined to the rectangle $\R'$.}
    \label{fig:block_disk_restricted}
\end{figure}

\begin{mlemmarep}\label{lem:feasible_areas_blocked}
    If $S_x$ is moved to $\F_{S_x}^{(1)}$, then $\F_{D_{c_i}} = \emptyset$ for $i \in \set{1,2,3}$.
\end{mlemmarep}
\begin{proof}
    Let $\R'$ be the rectangle defined by points $(-4,1/100)$, $(-4+1/50,1/100)$, $(-4+1/50,-1/100)$ and $(-4,-1/100)$.
    As proven in \Cref{lem:block_disk_restricted}, the disk $\bdisk{S_{t,1}^x}$ intersects with link disks $\ldisk{S^x_{c,1}}$, $\ldisk{S^x_{c,2}}$ and $\ldisk{S^x_{c,3}}$ when moved to a point $p \in \R'$ (see \Cref{fig:feasible_areas_blocked_1}).
    Suppose that $\bdisk{S_{t,1}^x}$ is actually moved to a point $p \in \R'$.
    We show the statement in the same fashion as in \Cref{lem:block_disk_restricted} using rectangles $\R^1,\R^2,\R^3$ containing $\F_{\ldisk{S^x_{c,1}}}$, $\F_{\ldisk{S^x_{c,2}}}$ and $\F_{\ldisk{S^x_{c,3}}}$, respectively.
    
    We start by defining $\R^1$ (see \Cref{fig:feasible_areas_blocked_2}).
    Let $q$ be a point such that $q \in \F_{\ldisk{S^x_{c,1}}}$.
    We have $q_y \le 1$ since $r(\A_{\ldisk{S^x_{c,1}}}) = 1/4$ and $c(\ldisk{S^x_{c,1}}) = (-4,3/4)$.
    The farthest point in $\R'$ from $c(\ldisk{S^x_{c,1}})$ is the bottom right corner $(-4+1/50,-1/100)$.
    We show a lower bound for $q_y$ assuming that $p = (-4+1/50,-1/100)$.
    Let $q'$ be the point with the lowest $y$-axis coordinate value that is in $\F_{\ldisk{S^x_{c,1}}}$. 
    The values of $q'$ are given by the intersection of the boundaries of $\B_{\bdisk{S_{t,1}^x}}$ and $\B_{\hdisk{(-5,1)}}$ (or $\B_{\hdisk{(-3,1)}}$).
    It can be checked that $q'_y$ approximately equals $0.9893$. Thus, we reasonably set the lower bound of $q_y$ to $1-1/50$.
    Given that $1-1/50 \le q_y\le 1$, the range of $q_x$ is given by the intersection of boundaries of blocked zones $\B_{\hdisk{(-5,1)}}$ and $\B_{\hdisk{(-3,1)}}$ with the boundary of $\B_{\bdisk{S_{t,1}^x}}$.
    It can be checked that $q_x$ has a value $-4\pm 0.015\dots$, thus we reasonably set the bound for $q_x$ to $-4 - 1/50 \le q_x \le -4+1/50$.

    We have the points that define $\R^1$. 
    Now we use $\R^1$ and show that if $\ldisk{S^x_{c,1}}$ is moved to any point $q\in \R^1$, then $\bdisk{S^x_{c,1}}$ must be moved to a point $p^1 \in \F_{\bdisk{S^x_{c,1}}}$ that makes $\F_{D_{c_1}} = \emptyset$.
    When $\ldisk{S^x_{c,1}}$ is moved to $q$, $\ismoved{\ldisk{S^x_{c,1}}} = 1$ and $r(\B_{\ldisk{S^x_{c,1}}}) =1$. 
    So the point $p^1$ must satisfy $\lVert p^1,q\rVert_\alpha \ge 1$.
    Moreover, $\B_{\ldisk{S^x_{c,1}}}$ intersects $\A_{\bdisk{S^x_{c,1}}}$, thus the lowest possible values for $p^1_y$ are given by the intersection of $\B_{\ldisk{S^x_{c,1}}}$ and $\B_{\hdisk{(-5,2)}}$ when $q = (-4+1/50,1-1/50)$ and the intersection of $\B_{\ldisk{S^x_{c,1}}}$ and $\B_{\hdisk{(-3,2)}}$ when $q = (-4-1/50,1-1/50)$.
    In particular, $p^1_y \ge 1.979\ldots > 2-3/100$.
    The $x$-axis values of the points on the boundary of $\B_{\hdisk{(-5,2)}}$ and $\B_{\hdisk{(-3,2)}}$ for which $y = 2-3/100$ are $-4- 0.0161\ldots$ and $-4+0.0161\ldots$, which are bounded by $-4- 1/50$ and $-4+ 1/50$, respectively.
    These two points are the farthest point from $s^x_{c,1}$ for which $\bdisk{S^x_{c,1}}$ can be relocated.
    \Cref{fig:feasible_areas_blocked_3} illustrates the range of values for $p^1$.
    Before moving $\bdisk{S^x_{c,1}}$, the closest points in $\F_{D_{c_1}}$ to $D_{c_1}$ are the intersection points of $\B_{\hdisk{(-5,3)}}$ and $\B_{\hdisk{(-3,3)}}$ with $\B_{\hdisk{(-4,4-1/4)}}$. 
    In particular, these points are $(-4.04\ldots,2.76\ldots)$ and $(-3.95\ldots,2.76\ldots)$.
    We reasonably round these points to $(-4.05,2.77)$ and $(-3.95,2.77)$, respectively.
    When $\bdisk{S^x_{c,1}}$ is moved to $(-4\pm 1/50, 2-3/100)$, 
    $\ismoved{\bdisk{S^x_{c,1}}} = 1$ and thus $r(\B_{\bdisk{S^x_{c,1}}}) = 1$. 
    It can be checked that both $(-4.05,2.77)$ and $(-3.95,2.77)$ are contained in $\B_{\bdisk{S^x_{c,1}}}$.
    See \Cref{fig:feasible_areas_blocked_3} for the case when $\bdisk{S^x_{c,1}}$ is moved to $(-4+ 1/50, 2-3/100)$ and \Cref{fig:feasible_areas_blocked_4} to see that the points defined are contained in $\B_{\bdisk{S^x_{c,1}}}$.
    Moreover, both points are the farthest points to $c(\bdisk{S^x_{c,1}})$ in $\F_{D_{c_1}}$ before moving $\bdisk{S^x_{c,1}}$. 
    Consequently, we conclude that $\F_{D_{c_1}} = \emptyset$.
    The proof is analogous for $\R^2$ and $\R^3$ by rotating the given coordinates by $\pi/2$ and $\pi$ degrees, respectively.
    We started by moving $\bdisk{S_{t,1}^x}$ to a point $p \in \R'$, which is given by moving $S_x$ to $\F_{S_x}^{(1)}$ by \Cref{lem:block_disk_restricted}.
    Therefore, if $S_x$ is moved to $\F_{S_x}^{(1)}$, then $\F_{D_{c_i}} = \emptyset$ for $i \in \set{1,2,3}$.
\end{proof}

\begin{figure}[!htb]
    \centering
    \includegraphics[scale=1,page=18]{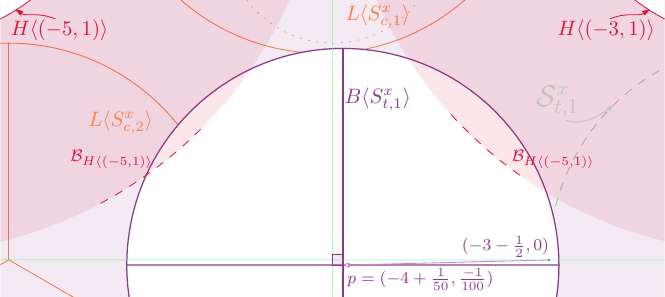}
    \caption{The disk $\bdisk{S_{t,1}^x}$ is moved to the point $p= (-4+1/50,-1/100) \in \R'$.}
    \label{fig:feasible_areas_blocked_1}
\end{figure}
\begin{figure}[!htb]
    \centering
    \includegraphics[scale=1,page=19]{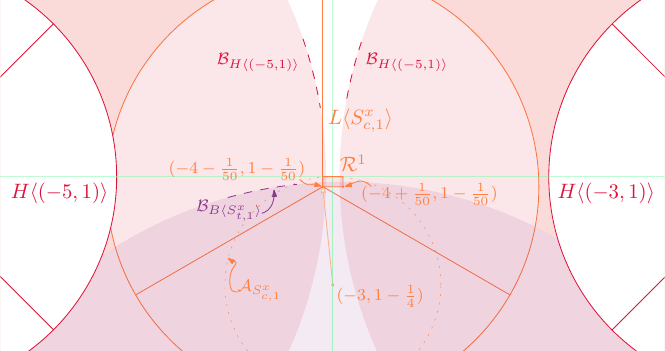}
    \caption{The disk $\ldisk{S_{c,1}^x}$ is moved to the point $(-4-1/50,1-1/50)$ contained in the rectangle $\R^1$ defined by points $(x,y)$ such that $-4-1/50 \le x \le -4+1/50$ and $1-1/50 \le y \le 1$.}
    \label{fig:feasible_areas_blocked_2}
\end{figure}
\begin{figure}[!htb]
    \centering
    \includegraphics[scale=1,page=20]{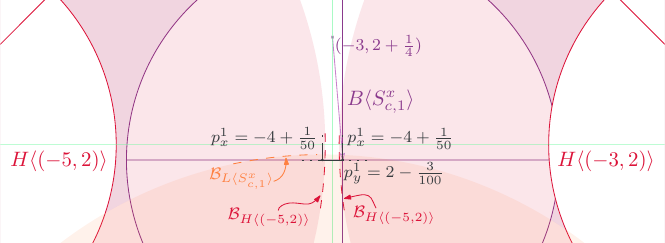}
    \caption{Lowest value for $p_y^1$ and range of values for $p_x^1$ assuming that $\ldisk{S_{c,1}^x}$ was moved to the point $(-4-1/50,1-1/50)$. The disk $\bdisk{S^x_{c,1}}$ is moved to the point $(-4+1/50,2-3/100)$.}
    \label{fig:feasible_areas_blocked_3}
\end{figure}
\begin{figure}[!htb]
    \centering
    \includegraphics[scale=1,page=21]{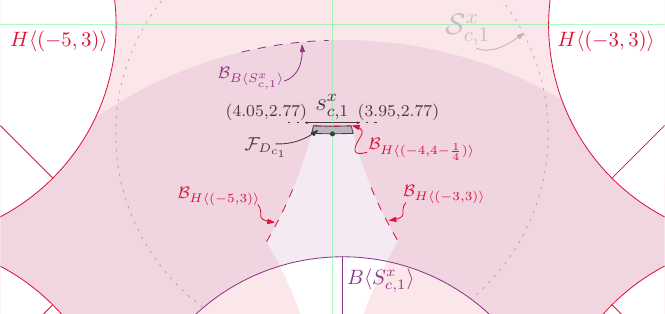}
    \caption{The closest points to $D_{c_1}$ in $\F_{D_{c_1}}$ before moving $\bdisk{S_{c,1}^x}$. When $\bdisk{S_{c,1}^x}$ is moved to the point $(-4\pm 1/50, 2-3/100)$, both points are contained in $\B_{\bdisk{S_{c,1}^x}}$ and $\F_{D_{c_1}}$ becomes $\emptyset$.}
    \label{fig:feasible_areas_blocked_4}
\end{figure}

Lastly, we show that if $S_x$ is moved to $\F_{S_x}^{(2)}$ (resp. $\F_{S_x}^{(1)}$), then there exist three spaces for moving $D_{c_i}$ to the true side for $i\in \set{1,2,3}$ (resp. the false side for $i\in \set{4,5,6}$).
\begin{mlemmarep}\label{lem:feasible_areas_available}
    If $S_x$ is moved to $\F_{S_x}^{(2)}$, then $\F_{D_{c_i}} \neq \emptyset$ for $i \in \set{1,2,3}$.
    Moreover, $\F_{\bdisk{S}^x_{t,1}} \neq \emptyset$.
\end{mlemmarep}

\begin{proof}
    Let $\G_{x}$ be the variable gadget of an arbitrary variable $x$.
    We show that the disks in $\G_{x}$ can be moved such that $\F_{D_{c_i}} \neq \emptyset$ for $i \in \set{1,2,3}$.
    Recall that the disks $H_1^{(1)},\ldots,H_5^{(1)}$ are the five $6$-heavy disks surrounding $\S^x_{t,1}$.
    First, we have $\S^x_{t,1} \setminus \cup_{j=1}^5 \B_{H_j^{(1)}} \neq \emptyset$ since $S_x$ was moved to $\F_{S_x}^{(2)}$.
    Moreover, $d_{\bdisk{S^x_{t,1}}}(s^x_{t,1}) = K$ holds, thus we move $\bdisk{S^x_{t,1}}$ to $s^x_{t,1}$.
    This implies $\F_{\bdisk{S}^x_{t,1}} \neq \emptyset$.
    The intersection between $\bdisk{S^x_{t,1}}$ and disks $\ldisk{S^x_{c,1}},\ldisk{S^x_{c,3}}$ is removed, so these disks can remain unmoved as well as $\ldisk{S^x_{c,2}}$.
    Notice that $s^x_{c,i}\notin \B_{\bdisk{S^x_{c,i}}}$ holds for $i\in \set{1,2,3}$.
    Moreover, $s^x_{c,i}$ is not contained in the zones blocked by heavy disks surrounding $\S^x_{c,i}$ by \Cref{obs:vertical_condition,obs:horizontal_condition}.
    We also know that $s^x_{c,i} \in \F_{D_{c_i}}$ since $d_{D_{c_i}}(s^x_{c,i}) \le K$.
    Therefore $\F_{D_{c_i}} \neq \emptyset$ for $i \in \set{1,2,3}$.
\end{proof}

\subsubsection{Clause Components}
\Cref{lem:cg_two_subsets,lem:tc_three_subsets,lem:block_disk_restricted,lem:feasible_areas_blocked} ensure that any undesired movement of the disks does not significantly alter the correctness of the reduction, whereas \Cref{lem:feasible_areas_available} provides a valid way to move the disks into the free slots of the variable gadget.
We are now ready to introduce the \emph{clause component} and show how the gadgets are connected to each other.

A \emph{clause component} $\G^c_{i,j,k} \subseteq \cstate$ is a collection of disks that represent the clause gadget $\G_c$ for a clause $c$ formed by variables $x_i,x_j,x_k$ connected to three variable gadgets $\G_{x_i},\G_{x_j},\G_{x_k}$ by consecutive cell gadgets.
As we mentioned earlier, the gadgets are connected by using arms, as depicted in \Cref{fig:variable_gadget}.
Arms are also formed by consecutive cell gadgets, but there exist cells that do not follow this definition.
We call these cells \emph{irregular cell gadgets} and their holes \emph{irregular interior holes}.
\Cref{fig:irreg_interior_holes} shows the three irregular cell gadgets present in the arms with their interior holes.

\begin{figure}[!htb]
    \centering
    \includegraphics[scale=1,page=30]{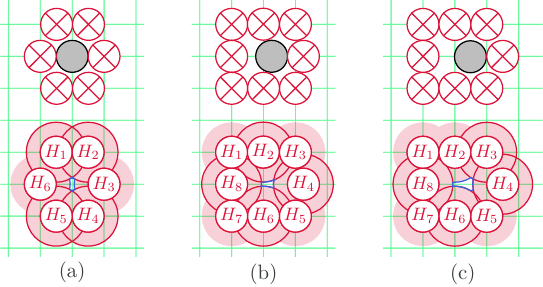}
    \caption{Top: Irregular cell gadgets of arms; Bottom: Irregular interior holes of their respective cell gadget. The marked points are the farthest pair of points for each interior hole.}
    \label{fig:irreg_interior_holes}
\end{figure}

We show that \Cref{lem:cg_two_subsets} can be extended to irregular cell gadgets.

\begin{mlemmarep}\label{lem:irreg_holes_one_disk}
    The interior hole $\H_{(x,y)}$ of an arbitrary irregular cell gadget $\G_{(x,y)}$ of an arm admits exactly one disk centred within it.
\end{mlemmarep}
\begin{proof}
    We prove the statement for each irregular cell gadget.
    Let $\G_{(x,y)}$ be the irregular cell gadget in \Cref{fig:irreg_interior_holes}(a).
    The disks $H_1,\ldots,H_6$ surrounding $D$ are $6$-heavy disks such that $c(H_1) = (x-1/2,y+1)$, $c(H_2) = (x+1/2,y+1)$, $c(H_3) = (x+1,y)$, $c(H_4) = (x+1/2,y-1)$, $c(H_5) = (x-1/2,y-1)$ and $c(H_6) = (x-1,y)$.
    Let $\S$ be the convex polygon containing $\H_{(x,y)}$ defined by the intersection points of the boundaries of the blocked zones contained in $\H_{(x,y)}$.
    We check that $\S$ satisfies $\diam{S} < 1$.
    To aim for simplicity, we only give the farthest pair of points of $\S$.
    The coordinates of the cell gadget can be checked in \Cref{apx:coordinates}.
    Let $p,p' \in \H_{(x,y)}$ be the intersection points of boundaries of $(\B_{H_1},\B_{H_2})$ and $(\B_{H_4},\B_{H_5})$, respectively.
    It can be checked that the farthest pair of points of $\S$ is $p$ and $p'$, for which $\lVert p,p'\rVert_2$ equals $0.304\ldots < 0.31$.
    That is, $\diam{\S} < 1$.
    By \Cref{obs:convex_polygon_one_disk}, $\S$ admits exactly one disk centred within it.
    Consequently, $\H_{(x,y)}$ also admits exactly one disk centred within it, since $\H_{(x,y)} = \S \setminus (\B_{H_1} \cup \cdots \cup \B_{H_6})$.
    
    Let $\G_{(x,y)}$ be the irregular cell gadget in \Cref{fig:irreg_interior_holes}(b) and $\S$ the convex polygon for $\H_{(x,y)}$.
    Let $p,p' \in \H_{(x,y)}$ be the intersection points of boundaries of $(\B_{H_2},\B_{H_4})$ and $(\B_{H_6},\B_{H_8})$, respectively.
    It can be checked that the farthest pair of points of $\S$ is $p$ and $p'$, for which $\lVert p,p'\rVert_2$ equals $0.101\ldots < 0.11$.
    Hence, $\S$ admits exactly one disk centred within it by \Cref{obs:convex_polygon_one_disk}.
    Consequently, $\H_{(x,y)}$ also admits exactly one disk centred within it.

    Lastly, $\G_{(x,y)}$ be the irregular cell gadget in \Cref{fig:irreg_interior_holes}(c) and $\S$ the convex polygon for $\H_{(x,y)}$.
    Let $p,p' \in \H_{(x,y)}$ be the intersection points of boundaries of $(\B_{H_3},\B_{H_{4}})$ and $(\B_{H_6},\B_{H_8)}$, respectively.
    It can be checked that the farthest pair of points of $\S$ is $p$ and $p'$, for which $\lVert p,p'\rVert_2$ equals $0.562\ldots < 0.57$.
    Hence, $\S$ admits exactly one disk centred within it by \Cref{obs:convex_polygon_one_disk}.
    Consequently, $\H_{(x,y)}$ also admits exactly one disk centred within it.

    Notice that the irregular cell gadgets in \Cref{fig:interior_holes} are present in the arm for $D_{c_3}$ (see \Cref{fig:variable_gadget}).
    The rest arms contains the gadget of \Cref{fig:interior_holes}(b) rotated.
    Therefore, all irregular gadgets admits one centre of disk within them.
\end{proof}

Observe that even if a cell gadget is irregular, the interior holes can be used as for normal cell gadgets.
This allows us to define \Cref{cor:cg_two_subsets}, which is a slight extension of \Cref{lem:cg_two_subsets} to characterise irregular cell gadgets.

\begin{corollary}\label{cor:cg_two_subsets}
    Let $\G_{(x,y)}$ and $\G_{(x',y')}$ be two (possibly irregular) cell gadgets with transition disks $D$ and $D'$, respectively.
    If $D$ and $D'$ are consecutive, then the feasible area $\F_{D}$ is equal to $(A_D \cap \H_{(x,y)}) \cup (A_{D} \cap \H_{(x',y')})$.
    Moreover, $A_D \cap \H_{(x,y)}$ and $A_{D'} \cap \H_{(x',y')}$ are non-empty and disjoint.
\end{corollary}

For any variable $x \in \set{x_i,x_j,x_k}$, when $x$ appears as a positive (negative) literal in $c$, $\G_{x}$ is connected to the true (false) side of $\G_c$.
\Cref{fig:component}(c) shows an example of a clause gadget for clause $c = (\overline{x_1} \lor x_2 \lor x_3)$.

\begin{figure}[!htb]
    \centering
    \includegraphics[scale=1,page=24]{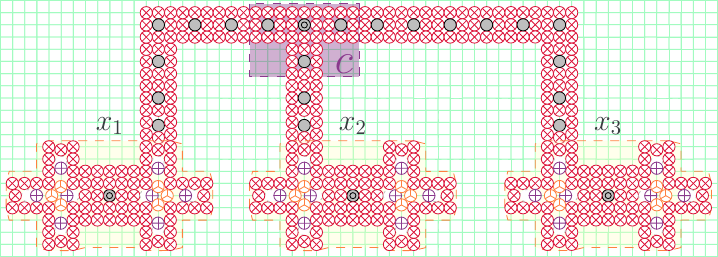}
    \caption{Example of a clause component for clause $c = (\overline{x_1} \lor x_2 \lor x_3)$.}
    \label{fig:component}
\end{figure}

We say \emph{removing the intersection of} $\G^c_{i,j,k}$ to refer to moving the intersection disk of $\G_c$ so that a free slot of $\G_x$ for $x\in \set{x_i,x_j,x_k}$ is occupied, under the condition that the minimum maximum moving distance is $K$.
Removing the intersection of $\G_c$ using a free slot of $\G_x$ is equivalent to assigning a truth value to $x$ that satisfies $c$.
\Cref{lem:clause_satisfiable_md_k} formalises the idea of removing an intersection.

\begin{mlemmarep}\label{lem:clause_satisfiable_md_k}
    Given an arbitrary clause $c$ and its clause component $\G^c_{i,j,k} \subseteq \istate$, $c$ is satisfiable if and only if removing the intersection of $\G^c_{i,j,k}$ can be done with minimum maximum moving distance $K$ for the $L_1$ and $L_2$ distances.
\end{mlemmarep}

\begin{proof}
    Without loss of generality, suppose that $x_i$ appears as a positive literal in $c$ and $c$ is satisfied by $x_i$.
    We show that $\mstatefinal{\G^c_{i,j,k}}$ without intersections can be obtained by moving disks with minimum maximum moving distance $K$.
    Let $\{D_1,\ldots,D_k\}$ be the collection of consecutive transition disks that connect $\G_c$ and $\G_{x_i}$ such that $D_1 = T_c$ and $D_k = D_{c_1}$.
    Since $x_i$ satisfies $c$ when assigned to true, $S_{x_i}$ is moved to $\F_{S_x}^{(2)}$.
    By \Cref{lem:feasible_areas_available}, we have $\F_{D_{c_1}}\neq \emptyset$.
    Thus, we move $D_k$ to $s^{x_i}_{c,1}$ with moving distance $K$.
    Since $\H_{c(D_k)} \setminus \B_{D_k} = \H_{c(D_k)}$,
    $D_{k-1}$ is moved to $c(D_k)$ with moving distance $K$.
    This procedure is repeated until $D_1 = T_c$ is moved to the previous position of $D_2$, which is a consecutive transition disk to $T_c$ in the clause gadget.
    \Cref{cor:cg_two_subsets} and \Cref{lem:tc_three_subsets} ensure that the procedure can be performed.
    Moreover, the procedure meets the conditions of \Cref{lem:holes_one_disk,lem:irreg_holes_one_disk}.
    The disks left in $\G^c_{i,j,k}$ are moved to their centres and consequently $\mstatefinal{\G^c_{i,j,k}}$ does not contain intersections.
    Therefore removing the intersection of $\G_c$ can be done with minimum maximum moving distance $K$.

    In the other direction, assume that removing the intersection of $\G^c_{i,j,k}$ can be done with minimum maximum moving distance of $K$.
    We show that the removal is equal to an assignment of variables that satisfies $c$.
    By \Cref{lem:tc_three_subsets}, $T_c$ must be moved to one of the three positions of its consecutive transition disks in $\G_c$.
    Without loss of generality, suppose that it was moved to the arm connected to the true side of $\G_{x_i}$ through a set $\{D_1,\ldots,D_k\}$ of consecutive transition disks as described before.
    Now $\movedpos{T_c} \in \H_{c(D_2)}$ holds, so $D_2$ must be moved outside $\H_{c(D_2)}$.
    By \Cref{lem:cg_two_subsets}, $\F_{D_2}$ is only the non-empty intersection of $\A_{D_2}$ and the interior hole $\H_{c(D_3)}$, as $\H_{c(D_2)} \setminus \B_{T_c} = \emptyset$.
    The same reasoning can be applied to disks $D_{3},\ldots,D_{k}$ by \Cref{lem:cg_two_subsets,cor:cg_two_subsets}.
    Assume that $D_k = D_{c_1}$.
    The feasible area of $D_k$ consists only of the non-empty intersection of $\A_{D_k}$ and the area of $\S^{x_i}_{c,1}$ that does not intersect blocked zones.
    Hence $D_k$ is moved to $s^{x_i}_{c,1}$.
    Since $\G_c$ is connected to the true side of $\G_{x_i}$, $S_{x_i}$ must be moved to $\F_{S_x}^{(2)}$ by \Cref{lem:feasible_areas_blocked} and $x_i$ appears as a positive literal in $c$ by the definition of clause gadgets.
    Therefore, $c$ is satisfied by an assignment of variables in which $x_i = 1$.
\end{proof}

\subsubsection{Reduction Correctness and Running Time}

We are now ready to characterise the reduction from {\pthreesat}. We first show the correctness of the reduction and then show that it can be obtained in polynomial time.

\begin{mlemmarep}\label{lem:3sat_edgeless_equiv}
    Given an instance $(\Phi,G_\Phi)$ of {\pthreesat}, the minimum maximum moving distance for satisfying $\Pi_{\texttt{edgeless}}$ in $\cstate$ is at most $K$ if and only if $\Phi$ is satisfiable.
\end{mlemmarep}

\begin{proof}
    Assume first that $\Phi$ is satisfiable by values $t_1,\ldots,t_n$ and the number of clauses is $m$.
    The collection of disks $\cstate$ can be partitioned into $m$ clause components that possibly share variable gadgets.
    By \Cref{lem:clause_satisfiable_md_k}, removing the intersections of all clause gadgets in $\cstate$ can be done with minimum maximum moving distance $K$.
    In particular, for each $t_i$, $S_{x_i}$ is moved to $\F_{S_{x_i}}^{(1)}$ if $t_i = 0$ and $\F_{S_{x_i}}^{(2)}$ otherwise.
    \Cref{lem:clause_satisfiable_md_k} ensures that by moving the truth setter disks in this way allows moving the intersection disk of the $m$ clauses to a free space of variable gadgets
    That is, it produces a $\mstatefinal{\D}$ without intersections.
    Therefore $\Pi_{\texttt{edgeless}}$ is satisfied with minimum maximum moving distance $K$.

    In the other direction, assume that $\Pi_{\texttt{edgeless}}$ is satisfied with minimum maximum moving distance $K$ in $\cstate$.
    That is, $\mstatefinal{\D}$ can be constructed by moving disks with minimum maximum moving distance $K$.
    By \Cref{lem:clause_satisfiable_md_k}, it implies that each clause is satisfied by an assignment of a variable of $\Phi$.
    Let $c,c' \in \Phi$ be two arbitrary clauses and $x$ be a variable such that $x$ appears in $c$ and $c'$.
    If $x$ appears in both variables as a positive (negative) literal, $\G_c$ and $\G_{c'}$ are connected to the true (false) side of $\G_{x}$ and \Cref{lem:feasible_areas_blocked,lem:feasible_areas_available} ensure that removing the intersection of both clause gadgets is given by the movement of $S_x$.
    Suppose instead that $x$ appears as a true literal in $c$ and as a false literal in $c'$.
    The gadgets $\G_c$ and $\G_{c'}$ are connected to the true and false side of $\G_x$, respectively.
    Again, \Cref{lem:feasible_areas_blocked,lem:feasible_areas_available} ensure that removing the intersection of $\G_{c}$ using $\G_x$ is not possible whenever the intersection of $\G_{c'}$ is removed using $\G_{x}$, or vice versa.
    By the definition of clause gadgets, it means that satisfying clauses $c$ and $c'$ by a variable $x$ such that $x$ appears as a true literal in $c$ and as a false literal in $c'$ (or the contrary) at the same time is not possible.
    That is, the side not blocked by $S_x$ decides the truth assignment of $x$.
    Consequently, the movement of the truth setter disks of the variable gadgets in $\cstate$ describes a feasible solution for satisfying $\Phi$.
    This concludes the proof.
\end{proof}

The last lemma shows that the construction of the instance can be done in polynomial time.

\begin{mlemmarep}\label{lem:reduction_poly_time}
    Given an instance $(\Phi,G_{\Phi})$ of {\pthreesat} with $n$ variables and $m$ clauses, the instance ($\cstate,K$) of {\ggedmm} can be obtained in $\poly(\eta)$ where $\eta = f(n,m)$.
\end{mlemmarep}

\begin{proof}
    By \Cref{lem:3sat_edgeless_equiv}, there exists a reduction from {\pthreesat} that converts an instance $\Phi$ to a collection of disks $\cstate$ such that $\Phi$ is satisfiable if and only if the minimum maximum moving distance for satisfying $\Pi_{\texttt{edgeless}}$ in $\cstate$ is $K$.
    Thus we only need to prove that $\cstate$ can be obtained in polynomial time from $\Phi$ and $G_{\Phi}$.
    An arbitrary $\cstate$ consists of horizontally aligned variable gadgets with clause gadgets vertically connected to them from up and down.
    The positions and connections of clause components are given by the representation of $G_{\Phi}$.
    All disks are in a unique position and do not intersect any other disk except for the $n + m$ intersection disks and the $n$ blocking disks intersecting link disks.
    We define a grid $\G$ of area $A_{\G} = W_{\G}H_{\G}$ that contains $\cstate$, where $W_{\G}$ and $H_{\G}$ denote the width and height of $\G$ expressed by number of disks, respectively. We show that $\eta = f(n,m)$ polynomially by the size of the input $n+m$.
    
    The maximum area occupied by a variable gadget is a constant value, denoted by $A_v = W_vH_v$.
    By the definition of the variable gadget, $W_v \le 35$ and $H_v \le 19$ hold.
    We add a separation of $10$ between variable gadgets. Consequently, $W_{\G} \le n(W_v+S)\le 35n\times 10$ holds. 
    A clause gadget $c$ occupies an area of $A_g = W_gH_g = 9\times 6$. 
    The number of `nested' clause components in $\G$ is limited by $m$.
    If variable gadgets are aligned at $y = 0$, we locate clause gadgets at $y=20$ for non-nested components and $y = 20+9m'$ for components enclosing $m'-1$ components, where $m' < m$.
    Consequently, $H_{\G}$ is bounded by $2(20+9m)$ and hence $A_{\G} \le 350n \times (40+18m)$.
    This implies that the number of disks required to construct $\cstate$ is also polynomially limited by $n$ and $m$.
    That is, $\eta = f(n,m)$ describes a polynomial over $n$ and $m$.
    Therefore, $\cstate$ can be constructed in polynomial time.
\end{proof}

With \Cref{lem:reduction_poly_time}, we restate \Cref{thm:edgeless_np_hard} below to conclude this section.
We remark that the strongly \NP-hardness comes from assuming $K=1$, which implies that the distance weights used are constant values. Therefore the values of the instance can be bounded by the input size.

\edgelessNPHard*

\end{toappendix}


\ifConf
    \begin{toappendix}
\clearpage
\section{Coordinates of gadgets}\label{apx:coordinates}

\begin{table}[!bht]
\caption{Coordinates for a cell gadget $\G_{(0,0)}$}
\label{tab:cell_gadget_coords}
\centering
\maxsizebox{\textwidth}{\textheight}{%
\begin{tabular}{@{}ccccccccccccccc@{}}
\toprule
\multicolumn{15}{l}{\textbf{Cell Gadget}} \\ \midrule
\multirow{2}{*}{\textbf{$x$}} & \multirow{2}{*}{\textbf{$y$}} & \multirow{2}{*}{\textbf{Type}} & \multirow{2}{*}{\textbf{$x$}} & \multirow{2}{*}{\textbf{$y$}} & \multirow{2}{*}{\textbf{Type}} & \multirow{2}{*}{\textbf{$x$}} & \multirow{2}{*}{\textbf{$y$}} & \multirow{2}{*}{\textbf{Type}} & \multirow{2}{*}{\textbf{$x$}} & \multirow{2}{*}{\textbf{$y$}} & \multirow{2}{*}{\textbf{Type}} & \multirow{2}{*}{\textbf{$x$}} & \multirow{2}{*}{\textbf{$y$}} & \multirow{2}{*}{\textbf{Type}} \\
 &  &  &  &  &  &  &  &  &  &  &  &  &  &  \\
$-1$ & $1$ & \multicolumn{1}{c|}{$6$-heavy} & $0$ & $1$ & \multicolumn{1}{c|}{$6$-heavy} & $1$ & $1$ & \multicolumn{1}{c|}{$6$-heavy} & $-1$ & $0$ & \multicolumn{1}{c|}{$6$-heavy} & $0$ & $0$ & Transition \\
$1$ & $0$ & \multicolumn{1}{c|}{$6$-heavy} & $-1$ & $-1$ & \multicolumn{1}{c|}{$6$-heavy} & $0$ & $-1$ & \multicolumn{1}{c|}{$6$-heavy} & $1$ & $-1$ & \multicolumn{1}{c|}{$6$-heavy} &  &  &  \\ \bottomrule
\end{tabular}%
}
\end{table}

\begin{table}[!bht]
\caption{Coordinates for a clause gadget $\G_{c}$ such that $c(T_c) = (0,0)$.}
\label{tab:clause_gadget_coords}
\centering
\maxsizebox{\textwidth}{\textheight}{%
\begin{tabular}{@{}ccccccccccccccc@{}}
\toprule
\multicolumn{15}{l}{\textbf{Clause Gadget}} \\ \midrule
\multirow{2}{*}{\textbf{$x$}} & \multirow{2}{*}{\textbf{$y$}} & \multirow{2}{*}{\textbf{Type}} & \multirow{2}{*}{\textbf{$x$}} & \multirow{2}{*}{\textbf{$y$}} & \multirow{2}{*}{\textbf{Type}} & \multirow{2}{*}{\textbf{$x$}} & \multirow{2}{*}{\textbf{$y$}} & \multirow{2}{*}{\textbf{Type}} & \multirow{2}{*}{\textbf{$x$}} & \multirow{2}{*}{\textbf{$y$}} & \multirow{2}{*}{\textbf{Type}} & \multirow{2}{*}{\textbf{$x$}} & \multirow{2}{*}{\textbf{$y$}} & \multirow{2}{*}{\textbf{Type}} \\
 &  &  &  &  &  &  &  &  &  &  &  &  &  &  \\
$1$ & $1$ & \multicolumn{1}{c|}{$6$-heavy} & $2$ & $1$ & \multicolumn{1}{c|}{$6$-heavy} & $3$ & $1$ & \multicolumn{1}{c|}{$6$-heavy} & $4$ & $1$ & \multicolumn{1}{c|}{$6$-heavy} & $-4$ & $0$ & $6$-heavy \\
$-3$ & $0$ & \multicolumn{1}{c|}{Transition} & $-2$ & $0$ & \multicolumn{1}{c|}{$6$-heavy} & $-1$ & $0$ & \multicolumn{1}{c|}{$6$-heavy} & $0$ & $0$ & \multicolumn{1}{c|}{Transition} & $0$ & $0$ & $6$-heavy \\
$1$ & $0$ & \multicolumn{1}{c|}{$6$-heavy} & $2$ & $0$ & \multicolumn{1}{c|}{$6$-heavy} & $3$ & $0$ & \multicolumn{1}{c|}{Transition} & $4$ & $0$ & \multicolumn{1}{c|}{$6$-heavy} & $-4$ & $-1$ & $6$-heavy \\
$-3$ & $-1$ & \multicolumn{1}{c|}{$6$-heavy} & $-2$ & $-1$ & \multicolumn{1}{c|}{$6$-heavy} & $-1$ & $-1$ & \multicolumn{1}{c|}{$6$-heavy} & $0$ & $-1$ & \multicolumn{1}{c|}{$6$-heavy} & $1$ & $-1$ & $6$-heavy \\
$2$ & $-1$ & \multicolumn{1}{c|}{$6$-heavy} & $3$ & $-1$ & \multicolumn{1}{c|}{$6$-heavy} & $4$ & $-1$ & \multicolumn{1}{c|}{$6$-heavy} & $-1$ & $-2$ & \multicolumn{1}{c|}{$6$-heavy} & $0$ & $-2$ & $6$-heavy \\
$1$ & $-2$ & \multicolumn{1}{c|}{$6$-heavy} & $-1$ & $-3$ & \multicolumn{1}{c|}{$6$-heavy} & $0$ & $-3$ & \multicolumn{1}{c|}{Transition} & $1$ & $-3$ & \multicolumn{1}{c|}{$6$-heavy} & $-1$ & $-4$ & $6$-heavy \\
$0$ & $-4$ & \multicolumn{1}{c|}{$6$-heavy} & $1$ & $-4$ & \multicolumn{1}{c|}{$6$-heavy} &  &  & \multicolumn{1}{c|}{} &  &  & \multicolumn{1}{c|}{} &  &  &  \\ \bottomrule
\end{tabular}%
}
\end{table}

\begin{table}[!bht]
\caption{Coordinates for the central part of a variable gadget $\G_{x}$ such that $c(S_x) = (0,0)$.}
\label{tab:variable_gadget_central_coords}
\centering
\maxsizebox{\textwidth}{\textheight}{%
\begin{tabular}{@{}ccccccccccccccc@{}}
\toprule
\multicolumn{15}{l}{\textbf{Variable Gadget (central part)}} \\ \midrule
\multirow{2}{*}{\textbf{$x$}} & \multirow{2}{*}{\textbf{$y$}} & \multirow{2}{*}{\textbf{Type}} & \multirow{2}{*}{\textbf{$x$}} & \multirow{2}{*}{\textbf{$y$}} & \multirow{2}{*}{\textbf{Type}} & \multirow{2}{*}{\textbf{$x$}} & \multirow{2}{*}{\textbf{$y$}} & \multirow{2}{*}{\textbf{Type}} & \multirow{2}{*}{\textbf{$x$}} & \multirow{2}{*}{\textbf{$y$}} & \multirow{2}{*}{\textbf{Type}} & \multirow{2}{*}{\textbf{$x$}} & \multirow{2}{*}{\textbf{$y$}} & \multirow{2}{*}{\textbf{Type}} \\
 &  &  &  &  &  &  &  &  &  &  &  &  &  &  \\
$4$ & $3+3/4$ & \multicolumn{1}{c|}{$6$-heavy} & $-5$ & $3$ & \multicolumn{1}{c|}{$6$-heavy} & $-3$ & $3$ & \multicolumn{1}{c|}{$6$-heavy} & $3$ & $3$ & \multicolumn{1}{c|}{$6$-heavy} & $5$ & $3$ & $6$-heavy \\
$-4$ & $2+1/4$ & \multicolumn{1}{c|}{$1$-heavy} & $4$ & $2+1/4$ & \multicolumn{1}{c|}{$1$-heavy} & $-5$ & $2$ & \multicolumn{1}{c|}{$6$-heavy} & $-3$ & $2$ & \multicolumn{1}{c|}{$6$-heavy} & $-2$ & $2$ & $6$-heavy \\
$-1$ & $2$ & \multicolumn{1}{c|}{$6$-heavy} & $0$ & $2$ & \multicolumn{1}{c|}{$6$-heavy} & $1$ & $2$ & \multicolumn{1}{c|}{$6$-heavy} & $2$ & $2$ & \multicolumn{1}{c|}{$6$-heavy} & $3$ & $2$ & $6$-heavy \\
$5$ & $2$ & \multicolumn{1}{c|}{$6$-heavy} & $-8$ & $1$ & \multicolumn{1}{c|}{$6$-heavy} & $-7$ & $1$ & \multicolumn{1}{c|}{$6$-heavy} & $-6$ & $1$ & \multicolumn{1}{c|}{$6$-heavy} & $-5$ & $1$ & $6$-heavy \\
$-3$ & $1$ & \multicolumn{1}{c|}{$6$-heavy} & $-2$ & $1$ & \multicolumn{1}{c|}{$6$-heavy} & $-1$ & $1$ & \multicolumn{1}{c|}{$6$-heavy} & $0$ & $1$ & \multicolumn{1}{c|}{$6$-heavy} & $1$ & $1$ & $6$-heavy \\
$2$ & $1$ & \multicolumn{1}{c|}{$6$-heavy} & $3$ & $1$ & \multicolumn{1}{c|}{$6$-heavy} & $5$ & $1$ & \multicolumn{1}{c|}{$6$-heavy} & $6$ & $1$ & \multicolumn{1}{c|}{$6$-heavy} & $7$ & $1$ & $6$-heavy \\
$8$ & $1$ & \multicolumn{1}{c|}{$6$-heavy} & $-4$ & $3/4$ & \multicolumn{1}{c|}{$2$-heavy} & $4$ & $3/4$ & \multicolumn{1}{c|}{$2$-heavy} & $-7-3/4$ & $0$ & \multicolumn{1}{c|}{$6$-heavy} & $-6$ & $0$ & $1$-heavy \\
$-4-3/4$ & $0$ & \multicolumn{1}{c|}{$2$-heavy} & $-3-1/2$ & $0$ & \multicolumn{1}{c|}{$1$-heavy} & $-2$ & $0$ & \multicolumn{1}{c|}{$6$-heavy} & $-1$ & $0$ & \multicolumn{1}{c|}{$6$-heavy} & $0$ & $0$ & Transition \\
$0$ & $0$ & \multicolumn{1}{c|}{$6$-heavy} & $1$ & $0$ & \multicolumn{1}{c|}{$6$-heavy} & $2$ & $0$ & \multicolumn{1}{c|}{$6$-heavy} & $3+1/2$ & $0$ & \multicolumn{1}{c|}{$1$-heavy} & $4+3/4$ & $0$ & $2$-heavy \\
$6$ & $0$ & \multicolumn{1}{c|}{$1$-heavy} & $7+3/4$ & $0$ & \multicolumn{1}{c|}{$6$-heavy} & $-4$ & $-3/4$ & \multicolumn{1}{c|}{$2$-heavy} & $4$ & $-3/4$ & \multicolumn{1}{c|}{$2$-heavy} & $-8$ & $-1$ & $6$-heavy \\
$-7$ & $-1$ & \multicolumn{1}{c|}{$6$-heavy} & $-6$ & $-1$ & \multicolumn{1}{c|}{$6$-heavy} & $-5$ & $-1$ & \multicolumn{1}{c|}{$6$-heavy} & $-3$ & $-1$ & \multicolumn{1}{c|}{$6$-heavy} & $-2$ & $-1$ & $6$-heavy \\
$-1$ & $-1$ & \multicolumn{1}{c|}{$6$-heavy} & $0$ & $-1$ & \multicolumn{1}{c|}{$6$-heavy} & $1$ & $-1$ & \multicolumn{1}{c|}{$6$-heavy} & $2$ & $-1$ & \multicolumn{1}{c|}{$6$-heavy} & $3$ & $-1$ & $6$-heavy \\
$5$ & $-1$ & \multicolumn{1}{c|}{$6$-heavy} & $6$ & $-1$ & \multicolumn{1}{c|}{$6$-heavy} & $7$ & $-1$ & \multicolumn{1}{c|}{$6$-heavy} & $8$ & $-1$ & \multicolumn{1}{c|}{$6$-heavy} & $-5$ & $-2$ & $6$-heavy \\
$-3$ & $-2$ & \multicolumn{1}{c|}{$6$-heavy} & $-2$ & $-2$ & \multicolumn{1}{c|}{$6$-heavy} & $-1$ & $-2$ & \multicolumn{1}{c|}{$6$-heavy} & $0$ & $-2$ & \multicolumn{1}{c|}{$6$-heavy} & $1$ & $-2$ & $6$-heavy \\
$2$ & $-2$ & \multicolumn{1}{c|}{$6$-heavy} & $3$ & $-2$ & \multicolumn{1}{c|}{$6$-heavy} & $5$ & $-2$ & \multicolumn{1}{c|}{$6$-heavy} & $-4$ & $-2-1/4$ & \multicolumn{1}{c|}{$1$-heavy} & $4$ & $-2-1/4$ & $1$-heavy \\
$-5$ & $-3$ & \multicolumn{1}{c|}{$6$-heavy} & $-3$ & $-3$ & \multicolumn{1}{c|}{$6$-heavy} & $3$ & $-3$ & \multicolumn{1}{c|}{$6$-heavy} & $5$ & $-3$ & \multicolumn{1}{c|}{$6$-heavy} & $-4$ & $-3-3/4$ & $6$-heavy \\
$4$ & $-3-3/4$ & \multicolumn{1}{c|}{$6$-heavy} & $-5$ & $-4$ & \multicolumn{1}{c|}{$6$-heavy} & $-3$ & $-4$ & \multicolumn{1}{c|}{$6$-heavy} & $3$ & $-4$ & \multicolumn{1}{c|}{$6$-heavy} & $5$ & $-4$ & $6$-heavy \\ \bottomrule
\end{tabular}%
}
\end{table}

\begin{table}[!bht]
\caption{Coordinates of arm 1 for $D_{c_1}$ when $S_x$ is centred at $(0,0)$.}
\label{tab:variable_gadget_arm1_coords}
\centering
\maxsizebox{\textwidth}{\textheight}{%
\begin{tabular}{@{}ccccccccccccccc@{}}
\toprule
\multicolumn{15}{l}{\textbf{Arm 1 (when $S_x$ centred at $(0,0)$)}} \\ \midrule
\multirow{2}{*}{\textbf{$x$}} & \multirow{2}{*}{\textbf{$y$}} & \multirow{2}{*}{\textbf{Disk Type}} & \multirow{2}{*}{\textbf{$x$}} & \multirow{2}{*}{\textbf{$y$}} & \multirow{2}{*}{\textbf{Disk Type}} & \multirow{2}{*}{\textbf{$x$}} & \multirow{2}{*}{\textbf{$y$}} & \multirow{2}{*}{\textbf{Disk Type}} & \multirow{2}{*}{\textbf{$x$}} & \multirow{2}{*}{\textbf{$y$}} & \multirow{2}{*}{\textbf{Disk Type}} & \multirow{2}{*}{\textbf{$x$}} & \multirow{2}{*}{\textbf{$y$}} & \multirow{2}{*}{\textbf{Disk Type}} \\
 &  &  &  &  &  &  &  &  &  &  &  &  &  &  \\
$-5$ & $9$ & \multicolumn{1}{c|}{$6$-heavy} & $-4$ & $9$ & \multicolumn{1}{c|}{$6$-heavy} & $-3$ & $9$ & \multicolumn{1}{c|}{$6$-heavy} & $-5$ & $8$ & \multicolumn{1}{c|}{$6$-heavy} & $-4$ & $8$ & Transition \\
$-3$ & $8$ & \multicolumn{1}{c|}{$6$-heavy} & $-5$ & $7$ & \multicolumn{1}{c|}{$6$-heavy} & $-4$ & $7$ & \multicolumn{1}{c|}{$6$-heavy} & $-3$ & $7$ & \multicolumn{1}{c|}{$6$-heavy} & $-5$ & $6$ & $6$-heavy \\
$-3$ & $6$ & \multicolumn{1}{c|}{$6$-heavy} & $-4$ & $5+3/4$ & \multicolumn{1}{c|}{Transition} & $-5$ & $5$ & \multicolumn{1}{c|}{$6$-heavy} & $-3$ & $5$ & \multicolumn{1}{c|}{$6$-heavy} & $-4$ & $4+3/4$ & $6$-heavy \\ \bottomrule
\end{tabular}%
}
\end{table}

\begin{table}[!bht]
\caption{Coordinates of arm for $D_{c_2}$ when $S_x$ is centred at $(0,0)$.}
\label{tab:variable_gadget_arm2_coords}
\centering
\maxsizebox{\textwidth}{\textheight}{%
\begin{tabular}{@{}ccccccccccccccc@{}}
\toprule
\multicolumn{15}{l}{\textbf{Arm 2 (when $S_x$ centred at $(0,0)$)}} \\ \midrule
\multirow{2}{*}{\textbf{$x$}} & \multirow{2}{*}{\textbf{$y$}} & \multirow{2}{*}{\textbf{Disk Type}} & \multirow{2}{*}{\textbf{$x$}} & \multirow{2}{*}{\textbf{$y$}} & \multirow{2}{*}{\textbf{Disk Type}} & \multirow{2}{*}{\textbf{$x$}} & \multirow{2}{*}{\textbf{$y$}} & \multirow{2}{*}{\textbf{Disk Type}} & \multirow{2}{*}{\textbf{$x$}} & \multirow{2}{*}{\textbf{$y$}} & \multirow{2}{*}{\textbf{Disk Type}} & \multirow{2}{*}{\textbf{$x$}} & \multirow{2}{*}{\textbf{$y$}} & \multirow{2}{*}{\textbf{Disk Type}} \\
 &  &  &  &  &  &  &  &  &  &  &  &  &  &  \\
$-9$ & $8$ & \multicolumn{1}{c|}{$6$-heavy} & $-11$ & $7$ & \multicolumn{1}{c|}{$6$-heavy} & $-10$ & $7$ & \multicolumn{1}{c|}{$6$-heavy} & $-9$ & $7$ & \multicolumn{1}{c|}{$6$-heavy} & $-11$ & $6$ & $6$-heavy \\
$-10$ & $6$ & \multicolumn{1}{c|}{$6$-heavy} & $-9$ & $6$ & \multicolumn{1}{c|}{$6$-heavy} & $-11$ & $5$ & \multicolumn{1}{c|}{$6$-heavy} & $-10$ & $5$ & \multicolumn{1}{c|}{Transition} & $-9$ & $5$ & $6$-heavy \\
$-11$ & $4$ & \multicolumn{1}{c|}{$6$-heavy} & $-10$ & $4$ & \multicolumn{1}{c|}{$6$-heavy} & $-9$ & $4$ & \multicolumn{1}{c|}{$6$-heavy} & $-11$ & $3$ & \multicolumn{1}{c|}{$6$-heavy} & $-10$ & $3$ & $6$-heavy \\
$-9$ & $3$ & \multicolumn{1}{c|}{$6$-heavy} & $-11$ & $2$ & \multicolumn{1}{c|}{$6$-heavy} & $-10$ & $2$ & \multicolumn{1}{c|}{Transition} & $-9$ & $2$ & \multicolumn{1}{c|}{$6$-heavy} & $-11$ & $1$ & $6$-heavy \\
$-10$ & $1$ & \multicolumn{1}{c|}{$6$-heavy} & $-9$ & $1$ & \multicolumn{1}{c|}{$6$-heavy} & $-11$ & $0$ & \multicolumn{1}{c|}{$6$-heavy} & $-9-3/4$ & $0$ & \multicolumn{1}{c|}{Transition} & $-8-3/4$ & $0$ & $6$-heavy \\
$-11$ & $-1$ & \multicolumn{1}{c|}{$6$-heavy} & $-10$ & $-1$ & \multicolumn{1}{c|}{$6$-heavy} & $-9$ & $-1$ & \multicolumn{1}{c|}{$6$-heavy} &  &  & \multicolumn{1}{c|}{} &  &  & \\\bottomrule
\end{tabular}%
}
\end{table}

\begin{table}[!bht]
\caption{Coordinates of arm for $D_{c_3}$ when $S_x$ is centred at $(0,0)$.}
\label{tab:variable_gadget_arm3_coords}
\centering
\maxsizebox{\textwidth}{\textheight}{%
\begin{tabular}{@{}ccccccccccccccc@{}}
\toprule
\multicolumn{15}{l}{\textbf{Arm 3 (when $S_x$ centred at $(0,0)$)}} \\ \midrule
\multirow{2}{*}{\textbf{$x$}} & \multirow{2}{*}{\textbf{$y$}} & \multirow{2}{*}{\textbf{Disk Type}} & \multirow{2}{*}{\textbf{$x$}} & \multirow{2}{*}{\textbf{$y$}} & \multirow{2}{*}{\textbf{Disk Type}} & \multirow{2}{*}{\textbf{$x$}} & \multirow{2}{*}{\textbf{$y$}} & \multirow{2}{*}{\textbf{Disk Type}} & \multirow{2}{*}{\textbf{$x$}} & \multirow{2}{*}{\textbf{$y$}} & \multirow{2}{*}{\textbf{Disk Type}} & \multirow{2}{*}{\textbf{$x$}} & \multirow{2}{*}{\textbf{$y$}} & \multirow{2}{*}{\textbf{Disk Type}} \\
 &  &  &  &  &  &  &  &  &  &  &  &  &  &  \\
$-15$ & $8$ & \multicolumn{1}{c|}{$6$-heavy} & $-17$ & $7$ & \multicolumn{1}{c|}{$6$-heavy} & $-16$ & $7$ & \multicolumn{1}{c|}{$6$-heavy} & $-15$ & $7$ & \multicolumn{1}{c|}{$6$-heavy} & $-17$ & $6$ & $6$-heavy \\
$-16$ & $6$ & \multicolumn{1}{c|}{$6$-heavy} & $-15$ & $6$ & \multicolumn{1}{c|}{$6$-heavy} & $-17$ & $5$ & \multicolumn{1}{c|}{$6$-heavy} & $-16$ & $5$ & \multicolumn{1}{c|}{Transition} & $-15$ & $5$ & $6$-heavy \\
$-17$ & $4$ & \multicolumn{1}{c|}{$6$-heavy} & $-16$ & $4$ & \multicolumn{1}{c|}{$6$-heavy} & $-15$ & $4$ & \multicolumn{1}{c|}{$6$-heavy} & $-17$ & $3$ & \multicolumn{1}{c|}{$6$-heavy} & $-16$ & $3$ & $6$-heavy \\
$-15$ & $3$ & \multicolumn{1}{c|}{$6$-heavy} & $-17$ & $2$ & \multicolumn{1}{c|}{$6$-heavy} & $-16$ & $2$ & \multicolumn{1}{c|}{Transition} & $-15$ & $2$ & \multicolumn{1}{c|}{$6$-heavy} & $-17$ & $1$ & $6$-heavy \\
$-16$ & $1$ & \multicolumn{1}{c|}{$6$-heavy} & $-15$ & $1$ & \multicolumn{1}{c|}{$6$-heavy} & $-17$ & $0$ & \multicolumn{1}{c|}{$6$-heavy} & $-16$ & $0$ & \multicolumn{1}{c|}{$6$-heavy} & $-15$ & $0$ & $6$-heavy \\
$-17$ & $-1$ & \multicolumn{1}{c|}{$6$-heavy} & $-16$ & $-1$ & \multicolumn{1}{c|}{Transition} & $-15$ & $-1$ & \multicolumn{1}{c|}{$6$-heavy} & $-17$ & $-2$ & \multicolumn{1}{c|}{$6$-heavy} & $-16$ & $-2$ & $6$-heavy \\
$-15$ & $-2$ & \multicolumn{1}{c|}{$6$-heavy} & $-17$ & $-3$ & \multicolumn{1}{c|}{$6$-heavy} & $-16$ & $-3$ & \multicolumn{1}{c|}{$6$-heavy} & $-15$ & $-3$ & \multicolumn{1}{c|}{$6$-heavy} & $-17$ & $-4$ & $6$-heavy \\
$-16$ & $-4$ & \multicolumn{1}{c|}{Transition} & $-15$ & $-4$ & \multicolumn{1}{c|}{$6$-heavy} & $-4$ & $-4-3/4$ & \multicolumn{1}{c|}{$6$-heavy} & $-17$ & $-5$ & \multicolumn{1}{c|}{$6$-heavy} & $-16$ & $-5$ & $6$-heavy \\
$-15$ & $-5$ & \multicolumn{1}{c|}{$6$-heavy} & $-14$ & $-5$ & \multicolumn{1}{c|}{$6$-heavy} & $-13$ & $-5$ & \multicolumn{1}{c|}{$6$-heavy} & $-12$ & $-5$ & \multicolumn{1}{c|}{$6$-heavy} & $-11$ & $-5$ & $6$-heavy \\
$-10$ & $-5$ & \multicolumn{1}{c|}{$6$-heavy} & $-9$ & $-5$ & \multicolumn{1}{c|}{$6$-heavy} & $-8$ & $-5$ & \multicolumn{1}{c|}{$6$-heavy} & $-7$ & $-5$ & \multicolumn{1}{c|}{$6$-heavy} & $-6$ & $-5$ & $6$-heavy \\
$-5$ & $-5$ & \multicolumn{1}{c|}{$6$-heavy} & $-3$ & $-5$ & \multicolumn{1}{c|}{$6$-heavy} & $-4$ & $-5-3/4$ & \multicolumn{1}{c|}{Transition} & $-17$ & $-6$ & \multicolumn{1}{c|}{$6$-heavy} & $-15-1/2$ & $-6$ & Transition \\
$-14-1/2$ & $-6$ & \multicolumn{1}{c|}{$6$-heavy} & $-13-1/2$ & $-6$ & \multicolumn{1}{c|}{$6$-heavy} & $-12-1/2$ & $-6$ & \multicolumn{1}{c|}{Transition} & $-11-1/2$ & $-6$ & \multicolumn{1}{c|}{$6$-heavy} & $-10-1/2$ & $-6$ & $6$-heavy \\
$-9-1/2$ & $-6$ & \multicolumn{1}{c|}{Transition} & $-8-1/2$ & $-6$ & \multicolumn{1}{c|}{$6$-heavy} & $-7-1/2$ & $-6$ & \multicolumn{1}{c|}{$6$-heavy} & $-6-1/2$ & $-6$ & \multicolumn{1}{c|}{Transition} & $-5$ & $-6$ & $6$-heavy \\
$-3$ & $-6$ & \multicolumn{1}{c|}{$6$-heavy} & $-17$ & $-7$ & \multicolumn{1}{c|}{$6$-heavy} & $-16$ & $-7$ & \multicolumn{1}{c|}{$6$-heavy} & $-15$ & $-7$ & \multicolumn{1}{c|}{$6$-heavy} & $-14$ & $-7$ & $6$-heavy \\
$-13$ & $-7$ & \multicolumn{1}{c|}{$6$-heavy} & $-12$ & $-7$ & \multicolumn{1}{c|}{$6$-heavy} & $-11$ & $-7$ & \multicolumn{1}{c|}{$6$-heavy} & $-10$ & $-7$ & \multicolumn{1}{c|}{$6$-heavy} & $-9$ & $-7$ & $6$-heavy \\
$-8$ & $-7$ & \multicolumn{1}{c|}{$6$-heavy} & $-7$ & $-7$ & \multicolumn{1}{c|}{$6$-heavy} & $-6$ & $-7$ & \multicolumn{1}{c|}{$6$-heavy} & $-5$ & $-7$ & \multicolumn{1}{c|}{$6$-heavy} & $-4$ & $-7$ & $6$-heavy \\
$-3$ & $-7$ & \multicolumn{1}{c|}{$6$-heavy} &  &  & \multicolumn{1}{c|}{} &  &  & \multicolumn{1}{c|}{} &  &  & \multicolumn{1}{c|}{} &  &  & \\\bottomrule
\end{tabular}%
}
\end{table}

    \end{toappendix}
\fi

\section{Concluding Remarks}\label{sec:conclu}

The main contribution of this paper is two-fold.
First, we continued the study of {\gged} originally presented in~\cite{HonoratoDroguett2024}, showing complexity results for satisfying several properties for sparse graphs on interval graphs.
In particular, we showed that satisfying properties $\Pi_{\texttt{edgeless}}$, $\Pi_{\texttt{acyc}}$ and $\overline{\Pi_{k\texttt{-clique}}}$ is solvable in $O(n\log n)$ time on unit interval graphs.
In contrast, we showed that the problem becomes strongly \NP-hard on weighted interval graphs for satisfying the same properties.
Second, we defined {\ggedmm} as a variation of the above problem and showed that it is strongly {\NP-hard} for satisfying $\Pi_{\texttt{edgeless}}$ on weighted unit disk graphs over the $L_1$ and $L_2$ distances.
%

There are several directions for further research.
Our results provide a comprehensive picture of the complexity of {\gged} on interval graphs. 
In particular, we showed that the problem becomes hard even in lower dimensions when the input is not restricted by interval size and distance weight.
As a result, a potential future work is to study the complexity when exclusively one of the restrictions is applied.
Another interesting direction is to study the model for satisfying $\Pi_{\texttt{edgeless}}$ in higher dimensions.
Related works \cite{fomin2023,Fomin2025,Fiala2005} suggest that our model on more complex intersection graphs becomes intractable for some of the properties presented in this work. 
In general, we deal with the edit operation that moves the objects of the given intersection graph. 
However, the model is not restricted to this operation. 
Determining {\gged} using other geometric edit operations (such as shrinking or rotating objects) is left for future research for all intersection graphs and graph properties presented in this work.

\bibliography{ref}

\clearpage

\appendix

\renewcommand{\thesubsection}{\arabic{section}}
\ifFull
    
\fi


\end{document}